\documentclass[11pt]{article}
\usepackage{fullpage}
\usepackage{times}
\usepackage{amsmath,amsfonts,amssymb,amsthm}
\usepackage{enumerate}
\usepackage{xcolor}
\usepackage{xspace}
\usepackage{tikz}
\usepackage{paralist}
\usepackage{framed}
\usepackage{xparse}
\usepackage{hyperref}
\usepackage{subcaption}
\usepackage{xparse}
\usetikzlibrary{decorations.pathreplacing, decorations.pathmorphing}

\newtheorem{theorem}{Theorem}
\newtheorem{lemma}[theorem]{Lemma}
\newtheorem{claim}[theorem]{Claim}

\newtheorem{proposition}[theorem]{Proposition}

\newtheoremstyle{restate}{}{}{\itshape}{}{\bfseries}{~(restated).}{.5em}{\thmnote{#3}}
\theoremstyle{restate}
\newtheorem*{restate}{}

\theoremstyle{definition}
\newtheorem{definition}[theorem]{Definition}
\theoremstyle{remark}
\newtheorem*{remark}{Remark}

\allowdisplaybreaks

\newcommand{\E}{\mathop{\mathbb{E}}}
\newcommand{\DKL}[2]{\bD_{\mathrm{KL}}(#1\,\|\,#2)}

\newcommand{\Dchisq}[2]{\bD_{\mathrm{\chisq}}(#1\,\|\,#2)}

\DeclareDocumentCommand{\dev}{ m m o }{\theta_{#2}(#1\IfValueT{#3}{\,@\,#3})}
\DeclareDocumentCommand{\thec}{ m m o }{\theta_{#2}(#1\IfValueT{#3}{\,@\,#3})}
\DeclareDocumentCommand{\chis}{ m m o }{\chi^2_{#2}(#1\IfValueT{#3}{\,@\,#3})}

\newcommand{\bD}{\mathbf{D}}

\newcommand{\bM}{\mathbf{M}}
\newcommand{\bm}{\mathbf{m}}

\newcommand{\bR}{\mathbf{R}}

\newcommand{\bX}{\mathbf{X}}

\newcommand{\bY}{\mathbf{Y}}

\newcommand{\cB}{\mathcal{B}}

\newcommand{\cM}{\mathcal{M}}

\newcommand{\cS}{\mathcal{S}}

\newcommand{\cU}{\mathcal{U}}

\newcommand{\cX}{\mathcal{X}}
\newcommand{\cY}{\mathcal{Y}}

\newcommand{\cost}{\mathrm{cost}}

\newcommand{\chisq}{\chi^2}

\newcommand{\odd}{\mathrm{odd}\,\xspace}
\newcommand{\even}{\mathrm{even}\,\xspace}

\newcommand{\tmp}{\mathrm{tmp}}

\newcommand{\adv}{\mathrm{adv}}
\newcommand{\pa}{\mathrm{pt}}
\newcommand{\rect}{\mathrm{rec}}
\newcommand{\newn}{\mathrm{new}}
\newcommand{\pen}{\pi_{{\newn}}}

\newcommand{\acc}{\mathtt{acc}}
\newcommand{\suc}{\mathrm{suc}}

\newcommand{\lowprob}{\textnormal{low-prob}}
\newcommand{\good}{\textnormal{low-cost}}
\newcommand{\reducetheta}{\textnormal{high-$\theta$}}
\newcommand{\reducechisA}{\textnormal{high-$\chi^2$-A}}
\newcommand{\reducechisB}{\textnormal{high-$\chi^2$-B}}
\newcommand{\reducecost}{\textnormal{high-cost}}

\newcommand{\disc}{\mathrm{disc}}

\newsavebox{\mybox}
\NewDocumentEnvironment{code}{mO{Algorithm}O{()}O{#1}}
	{\begin{center}
	\begin{lrbox}{\mybox}\footnotesize%
	\begin{minipage}{5.5in}
	\ifcsname c@pctr#1\endcsname\else\newcounter{pctr#1}\textbf{#2 $\mathtt{#4}#3$}:\fi
	\setlength{\topsep}{0pt}
	\begin{compactenum}
	\setcounter{enumi}{\value{pctr#1}}
	}{\setcounter{pctr#1}{\value{enumi}}
	\end{compactenum}
	\end{minipage}
	\end{lrbox}\fbox{\usebox{\mybox}}
	\end{center}}
\newcommand{\ctn}{\hfill (to be cont'd)}

\title{Strong XOR Lemma for Communication with Bounded Rounds}
\author{Huacheng Yu\thanks{Department of Computer Science, Princeton University. \href{mailto:yuhch123@gmail.com}{\texttt{yuhch123@gmail.com}}. Supported by Simons Junior Faculty Award.}}
\date{}

\begin{document}
	\maketitle
	\thispagestyle{empty}
	\begin{abstract}
		In this paper, we prove a strong XOR lemma for bounded-round two-player randomized communication.
		For a function $f:\cX\times \cY\rightarrow\{0,1\}$, the $n$-fold XOR function $f^{\oplus n}:\cX^n\times \cY^n\rightarrow\{0,1\}$ maps $n$ input pairs $(X_1,\ldots,X_n,Y_1,\ldots,Y_n)$ to the XOR of the $n$ output bits $f(X_1,Y_1)\oplus \cdots \oplus f(X_n, Y_n)$.
		We prove that if every $r$-round communication protocols that computes $f$ with probability $2/3$ uses at least $C$ bits of communication, then any $r$-round protocol that computes $f^{\oplus n}$ with probability $1/2+\exp(-O(n))$ must use $n\cdot \left(r^{-O(r)}\cdot C-1\right)$ bits.
		When $r$ is a constant and $C$ is sufficiently large, this is $\Omega(n\cdot C)$ bits.
		It matches the communication cost and the success probability of the trivial protocol that computes the $n$ bits $f(X_i,Y_i)$ independently and outputs their XOR, up to a constant factor in $n$.

		A similar XOR lemma has been proved for $f$ whose communication lower bound can be obtained via bounding the discrepancy~\cite{Shaltiel03}.
		By the equivalence between the discrepancy and the correlation with $2$-bit communication protocols~\cite{VW08}, our new XOR lemma implies the previous result.
	\end{abstract}

	\newpage 
	\thispagestyle{empty}
	\tableofcontents

	\newpage
	\setcounter{page}{1}

\section{Introduction}

In computational complexity, XOR lemmas study the relation between the complexity of a $\{0,1\}$-valued function $f(x)$ and the complexity of the $n$-fold XOR function $f^{\oplus n}$ where 
\[
	f^{\oplus n}(x_1,\ldots,x_n)=f(x_1)\oplus \cdots\oplus f(x_n)
\]
and $\oplus$ is the XOR.
A classic example is Yao's XOR lemma for circuits~\cite{Yao82a}, which states if $f$ cannot be computed with probability $2/3$ on a random input by size-$s$ circuits, then $f^{\oplus n}$ cannot be computed with probability $1/2+\exp(-\Omega(n))$ on a random input by size-$s'$ circuits for some $s'<s$ (and small $n$).
Such lemmas can be used to create very hard functions in a blackbox way, which can only be computed barely better than random guessing, from functions that are ``just'' hard to compute with constant probability.
This approach of hardness amplification has been used in one-way functions~\cite{Yao82a,Levin87}, pseudorandom generators~\cite{Impagliazzo95,IW97}, and more recently, streaming lower bounds~\cite{AN21,CKPSSY21}.

In general, suppose computing a function $f$ with probability $2/3$ requires resource $s$ in some model of computation (e.g., circuit size, running time, query complexity, communication cost, etc).
Then the trivial way to compute $f^{\oplus n}$ is to compute each $f(x_i)$ using resource $s$ independently, and output their XOR.
It uses resource $n\cdot s$ in total, and each instance is correct with probability $2/3$, hence, their XOR is correct with probability $1/2+\exp(-\Theta(n))$: For two \emph{independent} random bits $b_1,b_2$, if $\Pr[b_1=0]=1/2+\alpha_1/2$ and $\Pr[b_2=0]=1/2+\alpha_2/2$, then 
\begin{align*}\Pr[b_1\oplus b_2=0]=(1/2+\alpha_1/2)(1/2+\alpha_2/2)+(1/2-\alpha_1/2)(1/2-\alpha_2/2)=1/2+\alpha_1\alpha_2/2;
\end{align*}
let $b_i=0$ if and only if $f(x_i)$ is computed correctly, applying the above calculation inductively gives the claimed probability.
A \emph{strong XOR lemma} asserts that to achieve $1/2+\exp(-O(n))$ success probability, one must use $\Omega(n\cdot s)$ resource --- the trivial solution is essentially optimal.

Now suppose that we are given $n\cdot s/2$ resource in total, and we want to compute $f^{\oplus n}$.
If we try to solve the $n$ copies independently, then no matter how we distribute the resource among the $n$ copies, at least half of them will get no more than $s$.
The function $f^{\oplus}$ is computed correctly with probability at most $1/2+\exp(-\Omega(n))$.
Of course, since all $n$ inputs $(x_1,\ldots,x_n)$ are given together, we can potentially process them jointly.
This may correlate the $n$ copies, and in particular, it may correlate the correctness of computing each $f(x_i)$. 
Hence, one difficulty in proving the strong XOR lemma from the technical point of view is that in the above calculation of the probability of XOR of two independent bits, the linear terms perfectly cancel only because $b_1$ and $b_2$ are independent;
when they are not independent, we may get a linear term remaining, and do not reduce the probability bias as desired.
In computational models where one cannot expect the independence between the copies throughout the computation, a success probability lower bound of $1/2+\exp(-\Omega(n))$ (hence, a strong XOR lemma) is generally difficult to prove.

In this paper, we prove a strong XOR lemma for the two-player randomized communication complexity with bounded rounds: Alice and Bob receive $X$ and $Y$ respectively, they alternatively send a total of $r$ messages to each other with the goal of computing $f(X,Y)$.
For $f^{\oplus n}$, Alice receives $(X_1,\ldots,X_n)$ and Bob receives $(Y_1,\ldots,Y_n)$, and they wish to compute $f(X_1,Y_1)\oplus\cdots\oplus f(X_n,Y_n)$ after $r$ rounds of communication.
Each player has half of the inputs for all copies, and can send messages that arbitrarily depend on them, which can nontrivially correlate the $n$ instances.
Nevertheless, we show that one cannot do much better than simply solving all $n$ copies in parallel.

Let $\bR^{(r)}_{p}(f)$ be the minimum number of bits of communication needed in $r$ messages in order to compute $f(X, Y)$ correctly with probability $p$.
We prove the following theorem.

\newcommand{\thmmainnondistcont}{
	For any $\{0,1\}$-valued function $f$, we have
	\[
		\bR^{(r)}_{1/2+2^{-n}}(f^{\oplus n})\geq n\cdot \left(r^{-O(r)}\cdot \bR^{(r)}_{2/3}(f)-1\right).
	\]}
\begin{theorem}\label{thm_main_nondist}
	\thmmainnondistcont
\end{theorem}
In particular, when $r$ is a constant, it implies that $\bR^{(r)}_{1/2+2^{-n}}(f^{\oplus n})\geq \Omega\left(n\cdot \left(\bR^{(r)}_{2/3}(f)-O(1)\right)\right)$.\footnote{Observe that since $\bR^{(r)}_{0.51}(f)\leq \bR^{(r)}_{0.99}(f)\leq O(\bR^{(r)}_{0.51}(f))$, the constant $2/3$ does not matter as long as it is in $(1/2, 1)$. Our proof will also show that the base in $2^{-n}$ can be any constant.}
To the best of our knowledge, such an XOR lemma was not known even for one-way communication and without the factor of $n$.

As pointed in~\cite{BBCR13}, the ``$-O(1)$'' term is needed.
This is because for $f(X,Y)=X\oplus Y$, we have $\bR^{(r)}_{2/3}(f)=2$.
On the other hand, $f^{\oplus n}$ can also be computed with $2$ bits of communication by simply (locally) computing $\bigoplus_{i=1}^n X_i$ and $\bigoplus_{i=1}^n Y_i$ and exchanging the values.

We obtain Theorem~\ref{thm_main_nondist} via the following \emph{distributional} strong XOR lemma.
Let $\suc_{\mu}({f; C_A, C_B, r})$ be the maximum success probability of an $r$-round protocol $\pi$ computing $f(X, Y)$ where
\begin{compactitem}
	\item Alice sends at most $C_A$ bits in every odd round,
	\item Bob sends at most $C_B$ bits in every even round, and
	\item $(X, Y)$ is sampled from $\mu$.
\end{compactitem}

\newcommand{\thmmaincont}{
	Let $c>0$ be a sufficiently large constant.
	Fix $\alpha\in(0, r^{-cr})$ and $C_A,C_B\geq 2c\log(r/\alpha)$.
	Let $f:\cX\times \cY\rightarrow \{0,1\}$ be a function, and $\mu$ be a distribution over $\cX\times \cY$.
	Suppose $f$ satisfies $$\suc_{\mu}({f; C_A, C_B, r})\leq 1/2+\alpha/2,$$
	then for any integer $n\geq 2$, we have
	\[
		\suc_{\mu^n}(f^{\oplus n}; 2^{-8}r^{-1}n\cdot C_A, 2^{-8}r^{-1}n\cdot C_B, r)\leq \frac{1}{2}+\frac{\alpha^{2^{-12}n}}2.
	\]
}

\begin{theorem}\label{thm_main}
	\thmmaincont
\end{theorem}

This distributional strong XOR lemma states that for any fixed input distribution $\mu$ and function $f$, to compute $f^{\oplus n}$ when the $n$ inputs are sampled independently from $\mu$, either the advantage is exponentially small in $\Omega(n)$, or one of the players need to communicate at least $\Omega(n/r)$ times more than one copy.
This also gives a strong XOR lemma in the asymmetric communication, where we separately count how many bits Alice and Bob send.

It is worth noting that Shaltiel~\cite{Shaltiel03} proved a similar strong XOR lemma for functions whose communication lower bound can be obtained via bounding the \emph{discrepancy}.
By the equivalence between the discrepancy and the correlation with 2-bit protocols~\cite{VW08}, Theorem~\ref{thm_main} implies their result.
See Appendix~\ref{app_disc} for a more detailed argument.

\bigskip 

Note that a simple argument shows that Theorem~\ref{thm_main} implies Theorem~\ref{thm_main_nondist} (see also Section~\ref{sec_setup}).
Therefore, we will focus on the distributional version, and assume that the $n$ input pairs are sampled independently from some distribution $\mu$.

Our proof of the distributional version is inspired by the \emph{information complexity}~\cite{CSWY01}.
We define a new complexity measure for protocols, the \emph{$\chisq$-cost}, which is related to the \emph{internal information cost}~\cite{BYJKS04,BBCR13}.
Roughly speaking, it replaces the KL-divergence in the internal information cost with the $\chisq$-divergence, which can be viewed as the ``exponential'' version of KL.
This provides better concentration, which is needed in our argument.
Throughout the proof, we will also work with distributions that are ``close to'' communication protocols, i.e., the speaker's message may slightly depend on the receiver's input.
Such distributions have also been studied in the proof of direct product theorems~\cite{JPY12,BRWY13a,BRWY13}.
We will provide more details in Section~\ref{sec_overview}.

\subsection{Related work}
As we mentioned earlier, Shaltiel~\cite{Shaltiel03} proved a strong XOR lemma for functions whose communication lower bound can be obtained via bounding the discrepancy.
Sherstov~\cite{Sherstov11} extended this bound to generalized discrepancy and quantum communication complexity.

Barak, Braverman, Chen and Rao~\cite{BBCR13} obtained an XOR lemma for the information complexity and then an XOR lemma for communication (with worst parameters) via information compression.
However, their XOR lemma does not give exponentially small advantage.
They proved that if $f$ is hard to compute with information cost $C$, then $f^{\oplus n}$ is hard to compute with information cost $O(n\cdot C)$.
In fact, the starting point of our proof is an alternative view of their argument, which we will outline in Section~\ref{sec_overview_bbcr}.

Viola and Wigderson~\cite{VW08} proved a strong XOR lemma for multi-player $c$-bit communication for small $c$.
As pointed out in their paper, it implies the XOR lemma by Shaltiel~\cite{Shaltiel03}.
XOR lemmas have also been proved in circuit complexity~\cite{Yao82a,Levin87,Impagliazzo95,IW97,GNW11}, query complexity~\cite{Shaltiel03,Sherstov11,Drucker12,BKLS20}, streaming~\cite{AN21} and for low degree polynomials~\cite{VW08}.

\bigskip

Direct product and direct sum theorems, which are results of similar types, have also been studied in the literature.
They ask to return the outputs of all $n$ copies instead of their XOR.
Direct sum theorems state that the problem cannot be solved with the same probability unless $\Omega(n)$ times more resource is used, while direct product theorems state that the problem can only be solved with probability \emph{exponentially small} in $\Omega(n)$ unless $\Omega(n)$ times more resource is used.
The direct sum theorem for information complexity is known~\cite{CSWY01,BYJKS04,BBCR13}.
A direct sum theorem for communication complexity with suboptimal parameters can be obtained via information compression~\cite{BBCR13}.
A direct sum theorem for bounded-round communication has been proved~\cite{BR11}, and we use a similar argument in one component of the proof (see Section~\ref{sec_overview_compression} and Section~\ref{sec_compression}).
Direct product theorems for communication complexity (with suboptimal parameters via information compression), bounded-round communication and from information complexity to communication complexity have also been studied~\cite{JPY12,BRWY13,BW15}.


\section{Technical Overview}\label{sec_overview}
\subsection{An alternative view of~\cite{BBCR13}}\label{sec_overview_bbcr}
The starting point of our proof is an alternative view of the XOR lemma in~\cite{BBCR13} for \emph{information complexity}, which does not give an exponentially small advantage.
Running a protocol on an input pair sampled from some fixed input distribution defines a joint distribution over the input pairs and the transcripts.
Information complexity studies that in this joint distribution, how much information the transcript reveals about the inputs.
The (internal) information cost is defined as 
\[
	I(X; \bM\mid Y, M_0)+I(Y; \bM\mid X, M_0),
\]
where $\bM=(M_0, M_1,\ldots,M_r)$ is the transcript and $M_0$ is the public random bits.\footnote{In the usual definition, the public random string is not part of the transcript. We add it for simplicity of notations. This does not change the values of the mutual information terms as it is already in the condition.}
We assume that Alice sends all the odd $M_i$ and Bob sends all the even $M_i$.
The internal information cost of a protocol is always at most its communication cost.
It is also known that for bounded-round communication, the internal information complexity is roughly equal to the communication complexity~\cite{BR11} (up to some additive error probability).

For the XOR lemma for information complexity, we consider input pair $X=(X_1,\ldots,X_n)$ and $Y=(Y_1,\ldots,Y_n)$ sampled from $\mu^n$.
Suppose there is a protocol $\pi$ computing $f^{\oplus n}$ with information cost $I$, we want to show that $f$ can be computed with information cost $\approx I/n$.

To this end, we show that $\pi$ can be ``decomposed'' into a protocol $\pi_{<n}$ computing $f^{\oplus n-1}$ with information cost $I_1$ and a protocol $\pi_{n}$ computing $f$ with information cost $I_2$ such that $I_1+I_2\approx I$, as follows (see also Figure~\ref{fig_decompose}).
\begin{itemize}
	\item For $\pi_{<n}$, given $n-1$ input pairs, the players view them as $X_{<n}$ and $Y_{<n}$ as part of the inputs for $\pi$, where $X_{<n}$ denotes $(X_1,\ldots,X_{n-1})$ and $Y_{<n}$ denotes $(Y_1,\ldots,Y_{n-1})$;
	then the players publicly sample $X_n\sim \mu_X$, and Bob privately samples $Y_n$ conditioned on $X_n$;
	the players run $\pi$ to compute $f^{\oplus n}(X, Y)$;
	Bob sends one extra bit indicating $f(X_n, Y_n)$.
	\item For $\pi_n$, given one input pair, the players view it as $X_n$ and $Y_n$;
	then the players publicly sample $Y_{<n}\sim \mu_Y^{n-1}$, and Alice privately samples $X_{<n}$ conditioned on $Y_{<n}$;
	the players run $\pi$ to compute $f^{\oplus n}(X, Y)$;
	Alice sends one extra bit indicating $\oplus_{i=1}^{n-1}f(X_i, Y_i)$.
\end{itemize}

\begin{figure}
	\centering
	\begin{subfigure}{0.5\textwidth}
	\centering
	\begin{tikzpicture}
		\node at (-10pt, 5pt) {$X$};
		\draw (0,0) rectangle (100pt, 10pt);
		\draw [fill=gray] (90pt,0) rectangle (100pt, 10pt);
		\fill (0, 0) rectangle (90pt, 10pt);
		\node at (-10pt, -15pt) {$Y$};
		\draw (0,-20pt) rectangle (100pt, -10pt);
		\fill (0,-20pt) rectangle (90pt, -10pt);
	\end{tikzpicture}
	\caption{Protocol $\pi_{<k}$}
	\end{subfigure}%
	\begin{subfigure}{0.5\textwidth}
	\centering
	\begin{tikzpicture}
		\node at (-10pt, 5pt) {$X$};
		\draw (0,0) rectangle (100pt, 10pt);
		\fill (90pt,0) rectangle (100pt, 10pt);
		\node at (-10pt, -15pt) {$Y$};
		\draw (0,-20pt) rectangle (100pt, -10pt);
		\fill (90pt, -20pt) rectangle (100pt, -10pt);
		\draw [fill=gray] (0,-20pt) rectangle (90pt, -10pt);
	\end{tikzpicture}
	\caption{Protocol $\pi_{k}$}
	\end{subfigure}
	\caption[Decomposition of $\pi$]{Decomposition of $\pi$: \tikz[baseline=2pt]{\fill (0pt,0pt) rectangle (10pt, 10pt);} is the inputs, \tikz[baseline=2pt]{\draw[fill=gray] (0pt,0pt) rectangle (10pt, 10pt);} is sampled publicly, \tikz[baseline=2pt]{\draw (0pt,0pt) rectangle (10pt, 10pt);} is sampled privately.}\label{fig_decompose}
\end{figure}

If $\pi$ computes $f^{\oplus n}$ correctly, then the two protocols compute $f^{\oplus n-1}$ and $f$ correctly respectively.
For the information cost of $\pi_{<n}$ (if we exclude the last bit indicating $f(X_n,Y_n)$), the first term is equal to $I(X_{<n};\bM\mid Y_{<n},X_n,M_0)$, since $X_n$ is sampled using public random bits.
It is also equal to $I(X_{<n};\bM\mid Y,X_n,M_0)$ due to the \emph{rectangle property} of communication protocols.
For the information cost of $\pi_n$ (if we exclude the last bit), the first term is equal to $I(X_n;\bM\mid Y,M_0)$ since $Y_{<n}$ is sampled using public random bits.
Therefore, the first terms sum up to exactly $I(X;\bM\mid Y,M_0)$, the first term in the information cost of $\pi$, by the chain rule of mutual information.
Similarly, the second terms sum up to $I(Y;\bM\mid X, M_0)$, the second term in the information cost of $\pi$.

Hence, including the last bits in the protocols, we have $I_1+I_2\leq I+O(1)$.
Thus, by repeatedly applying this argument, we obtain a protocol for $f$ with information cost $I/n+O(1)$, as desired.
Note that in this decomposition, the players \emph{do not} need to sample the private parts explicitly.
As long as they can send the messages from the same distribution (e.g., by directly sampling the messages conditioned on the previous messages and their own inputs), the information costs and correctness are not affected.

\bigskip
The original paper proves the same result by explicitly writing out the protocol for $f$ obtained after applying the above decomposition $i$ times for a random $i\in[n]$, and proving the expected cost is as claimed.
The two proofs are essentially equivalent for this statement.\footnote{The orginal proof embeds the input to $f$ into a random coordinate $i$ of $f^{\oplus n}$, and samples $X_{>i}$ and $Y_{<i}$ using public random bits.}
However, as we will see later, our new view is more flexible, allowing for more sophisticated manipulations when doing the decomposition.

\subsection{Obtaining exponentially small advantage}
The above decomposition preserves the success probability.
However, if we start from a protocol for $f^{\oplus n}$ with exponentially small advantage, then we will not be able to obtain a protocol for $f$ with success probability $2/3$, which is required in order to prove the strong XOR lemma.\footnote{\label{foot_counterexample}In fact, this is inherent for information complexity, since the strong XOR lemma for \emph{information complexity} does not hold. This is because the information complexity is an average measure, and it is at most the \emph{expected} communication. A protocol can choose to compute \emph{all} $f(X_i, Y_i)$ with probability $\varepsilon$ and output a random bit otherwise, which achieves success probability $1/2+\varepsilon$ and expected communication $\varepsilon n$ times the one-copy cost. However, the reader is encouraged to continue reading this subsection pretending that they do not know about this counterexample.}

Let $\adv(f(X, Y)\mid \bR)$ denote the advantage for $f(X, Y)$ conditioned on $\bR$, which is defined as
\[
	\left|2\Pr[f(X, Y)=1\mid \bR]-1\right|,
\]
i.e., the advantage is $\alpha$ if the conditional probability is either $1/2+\alpha/2$ or $1/2-\alpha/2$.

Now let us take a closer look at the two protocols $\pi_{<n}$ and $\pi_n$ (see Figure~\ref{fig_decompose_2}).
For $\pi_{<n}$, in \emph{Bob's view} at the end of the communication, he knows his input $Y_{<n}$, the publicly sampled $X_n$ and the transcript $\bM$.
Hence, he is able to predict $f^{\oplus n-1}(X_{<n}, Y_{<n})$ with advantage $\adv(f^{\oplus n-1}(X_{<n}, Y_{<n})\mid X_n, Y_{<n}, \bM)$.
By letting Bob send one extra bit indicating his prediction, the advantage of the protocol achieves the same.
For $\pi_n$, in \emph{Alice's view} at the end of the communication, she knows her input $X_n$, the publicly sampled $Y_{<n}$ and the transcript $\bM$.
Hence, she is able to predict $f(X_n, Y_n)$ with advantage $\adv(f(X_n, Y_n)\mid X_n, Y_{<n}, \bM)$.
By letting Alice send one extra bit indicating her prediction, the advantage of the protocol achieves the same.

Now an important observation is that $X_{<n}$ and $Y_n$ are independent conditioned on $(X_n, Y_{<n}, \bM)$, by the rectangle property of communication protocols.
Hence, $f^{\oplus n-1}(X_{<n}, Y_{<n})$ and $f(X_n, Y_n)$ are also independent conditioned on $(X_n, Y_{<n}, \bM)$.
Since $f^{\oplus n}(X, Y)=f^{\oplus n-1}(X_{<n}, Y_{<n})\oplus f(X_n, Y_n)$, by the probability of XOR of two independent bits, we have
\begin{equation}\label{eqn_intro_adv_prod}
	\begin{aligned}\adv(f^{\oplus n}(X, Y)\mid X_n, Y_{<n}, \bM)&=\adv(f^{\oplus n-1}(X_{<n}, Y_{<n})\mid X_n, Y_{<n}, \bM) \\
	&\hspace{70pt}\times \adv(f(X_n, Y_n)\mid X_n, Y_{<n}, \bM).
	\end{aligned}
\end{equation}

\begin{figure}
	\centering
	\begin{subfigure}{0.5\textwidth}
	\centering
	\begin{tikzpicture}
		\node at (-10pt, 5pt) {$X$};
		\draw (0,0) rectangle (100pt, 10pt);
		\draw [fill=gray] (90pt,0) rectangle (100pt, 10pt);
		\fill (0, 0) rectangle (90pt, 10pt);
		\node at (-10pt, -15pt) {$Y$};
		\draw (0,-20pt) rectangle (100pt, -10pt);
		\draw [fill=gray] (0,-20pt) rectangle (90pt, -10pt);
	\end{tikzpicture}
	\caption{Bob's view in $\pi_{<k}$}
	\end{subfigure}%
	\begin{subfigure}{0.5\textwidth}
	\centering
	\begin{tikzpicture}
		\node at (-10pt, 5pt) {$X$};
		\draw (0,0) rectangle (100pt, 10pt);
		\draw [fill=gray] (90pt,0) rectangle (100pt, 10pt);
		\node at (-10pt, -15pt) {$Y$};
		\draw (0,-20pt) rectangle (100pt, -10pt);
		\fill (90pt, -20pt) rectangle (100pt, -10pt);
		\draw [fill=gray] (0,-20pt) rectangle (90pt, -10pt);
	\end{tikzpicture}
	\caption{Alice's view in $\pi_{k}$}
	\end{subfigure}
	\caption[Decomposition of $\pi$]{Decomposition of $\pi$: \tikz[baseline=2pt]{\fill (0pt,0pt) rectangle (10pt, 10pt);} is unknown inputs, \tikz[baseline=2pt]{\draw[fill=gray] (0pt,0pt) rectangle (10pt, 10pt);} is known.}\label{fig_decompose_2}
\end{figure}

This suggests the following strategy for the decomposition:
\begin{itemize}
	\item if the information cost of $\pi_n$ is large, then the information cost of $\pi_{<n}$ must be much smaller than that of $\pi$;
	\item if the information cost of $\pi_n$ is small and its advantage for $f$ is large, then we have obtained a good protocol for $f$;
	\item if the information cost of $\pi_n$ is small and its advantage for $f$ is small, then by~\eqref{eqn_intro_adv_prod}, the advantage of $\pi_{<n}$ must be larger than that of $\pi$ by some factor.
\end{itemize}
Hence, in each decomposition, if we don't already obtain a good protocol for $f$, then when decrementing $n$ to $n-1$, we must either significantly decrease the information cost, or increase the advantage by a multiplicative factor.
If we start with a protocol with a low cost and a mild-exponentially small advantage for $f^{\oplus n}$, then we must obtain a good protocol for $f$ by applying this decomposition iteratively.

It turns out that the main difficulty in applying the above strategy is to formalize the last bullet point.
Note that the expected advantage of $\pi_n$ (after Alice sending the one extra bit indicating her prediction) is $\E\left[\adv(f(X_n, Y_n)\mid X_n, Y_{<n}, \bM)\right]$, the expected advantage of $\pi_{<n}$ (after Bob sending the one extra bit indicating his prediction) is $\E\left[\adv(f^{\oplus n-1}(X_{<n}, Y_{<n})\mid X_n, Y_{<n}, \bM)\right]$, and the expected advantage of $\pi$ is $\E\left[\adv(f^{\oplus n}(X, Y)\mid \bM)\right]$, which is at most $\E\left[\adv(f^{\oplus n}(X, Y)\mid X_n, Y_{<n}, \bM)\right]$.

When we say that the advantage of $\pi_n$ for $f$ is small in the last bullet point, we can only guarantee that this expectation is small.
Equation~\eqref{eqn_intro_adv_prod}, which is a pointwise equality, does not directly give any useful bounds on the expectations.
For example, it is possible that both $\E\left[\adv(f(X_n, Y_n)\mid X_n, Y_{<n}, \bM)\right]$ and $\E\left[\adv(f^{\oplus n-1}(X_{<n}, Y_{<n})\mid X_n, Y_{<n}, \bM)\right]$ are very small, but $\adv(f(X_n, Y_n)\mid X_n, Y_{<n}, \bM)$ and $\adv(f^{\oplus n-1}(X_{<n}, Y_{<n})\mid X_n, Y_{<n}, \bM)$ are always equal to zero or one at the same time, both concentrated on a small probability set.
Then we have $\E[\adv(f^{\oplus n}(X, Y)\mid X_n, Y_{<n}, \bM)]=\E[\adv(f^{\oplus n-1}(X_{<n}, Y_{<n})\mid X_n, Y_{<n}, \bM)]$, the advantage may not increase at all.
In this case, the advantage $\adv(f^{\oplus n}(X, Y)\mid X_n, Y_{<n}, \bM)$ is also concentrated on the same small probability set.

On the other hand, observe that if $\adv(f^{\oplus n}(X, Y)\mid X_n, Y_{<n}, \bM)$ takes roughly the same value (say, $\varepsilon$) most of the time, then we do obtain an advantage increase:
\begin{align*}
	&\kern-2em\E[\adv(f^{\oplus n-1}(X_{<n}, Y_{<n})\mid X_n, Y_{<n}, \bM)] \\
	&=\E[\varepsilon/\adv(f(X_n, Y_n)\mid X_n, Y_{<n}, \bM)] \\
	&\geq \varepsilon/\E[\adv(f(X_n, Y_n)\mid X_n, Y_{<n}, \bM)]
\end{align*}
by the convexity of $1/x$.

This motivates us to consider the following two extreme cases:
\begin{enumerate}
	\item $\adv(f^{\oplus n}(X, Y)\mid X_n, Y_{<n}, \bM)$ is roughly uniformly distributed among all $(X_n, Y_{<n}, \bM)$;
	\item $\adv(f^{\oplus n}(X, Y)\mid X_n, Y_{<n}, \bM)$ is concentrated on a tiny fraction of the triples $(X_n, Y_{<n}, \bM)$.
\end{enumerate}
Basically following what we just argued, the above strategy directly applies in the first case.
The second case is related to the \emph{direct product theorems}, where we also want to analyze protocols that is correct with exponentially small probability.
This is because one possible strategy for the players is to compute all $f(X_i, Y_i)$ correctly with some probability $\varepsilon$ and output a random bit otherwise. 
We must at least show that in this case, $\varepsilon\leq \exp(-\Omega(n))$.

\subsection{Generalized protocols}

For the second case above, we follow one strategy for direct product theorems~\cite{BRWY13}.
When the advantage $\adv(f^{\oplus n}(X, Y)\mid X_n, Y_{<n}, \bM)$ is concentrated on a small set $U$ of triples $(X_n, Y_{<n}, \bM)$, we restrict our attention to $U$ by \emph{conditioning} $\pi$ on $U$.
However, this immediately creates two issues.

The first issue is that although $\pi\mid U$ is a well-defined distribution, it is not necessarily a \emph{protocol}, since conditioning on an arbitrary event may break the independence between a message and the receiver's input, e.g., $M_1$ may no longer be independent of $Y$conditioned on $X$.\footnote{Conditioning on an event also distorts the input distribution, which needs to be handled. But for simplicity, we omit it in the overview.}

This issue was also encountered in the direct product theorem proofs.
Instead of studying standard protocols, we focus on \emph{generalized protocols}, where we allow each message to depend on both player's inputs, and we wish to restrict the correlation between the odd $M_i$ and Bob's input and the correlation between the even $M_i$ and Alice's input.
In the previous work, it bounds
\[
	\theta(\pi):=\sum_{\odd i} I(M_i; Y\mid X, M_{<i})+\sum_{\even i} I(M_i; X\mid Y, M_{<i}),
\]
the mutual information between the message and the receiver's input.

Intuitively, the $\theta$-value measures how close to a standard protocol a generalized protocol is.
It turns out that the $\theta$-value of a standard protocol conditioned on a not-too-small probability event is small; on the other hand, when the $\theta$-value is small, it is statistically close to a standard protocol.
Furthermore, an important feature of $\theta(\pi)$ is that the decomposition of $\pi$ into $\pi_{<n}$ and $\pi_n$ also satisfies that $\theta(\pi)=\theta(\pi_{<n})+\theta(\pi_n)$.
Hence, when doing the decomposition, we can hope to obtain a generalized protocol for $f$ that is very close to a standard protocol.

\bigskip

The second issue is that conditioning on a small probability event $U$ could greatly increase the information cost, from $I$ to $\Omega(I/\Pr[U])$. 
Since $I$ is close to the communication cost, such a multiplicative loss in each step of decomposition is unaffordable.
Such a loss occurs because the mutual information is an average measure (an expectation), which does not provide any concentration (also recall the counterexample in footnote~\ref{foot_counterexample} where the communication cost and the advantage are both concentrated on an $\varepsilon$-probability event, when we condition on this event, both the expected communication cost and the advantage increase by a factor of $1/\varepsilon$).
More specifically, consider the first term in the information cost, $I(X;\bM\mid Y)$ (omit the public random bits for now).
For standard protocols, it is equal to 
\[
	\sum_{x,y,\bm} \pi(x, y, \bm)\cdot \log \left(\frac{\pi(x\mid \bm, y)}{\pi(x\mid y)}\right)=\E_{\pi}\left[\log \left(\frac{\pi(X\mid \bM, Y)}{\pi(X\mid Y)}\right)\right]=\E_{\pi}\left[\log \left(\frac{\pi(X\mid \bM, Y)}{\mu(X\mid Y)}\right)\right].
\]

If we only have a bound on this expectation, then inevitably its value can greatly increase after conditioning on a small probability event, not to say that the logarithm inside the expectation is \emph{not} nonnegative, so it can get worse than what Markov's inequality gives.

\bigskip

We also note that the argument in the previous subsections crucially uses the \emph{rectangle property} of the communication protocols, which does not necessarily hold for generalized protocols.
This turns out not to be a real issue, since throughout the argument, we will maintain the rectangle property \emph{at all leaves}, which is sufficient for the argument to go through (see also Section~\ref{sec_overview_compression}).

\subsection{$\theta$-cost and $\chisq$-costs}
Our novel solution to the second issue above is to focus on the ``exponential version'' of the information cost, i.e., for the first term, 
\[
	\chis{\pi}{\mu,A}:=\E_{\pi}\left[\frac{\pi(X\mid \bM, Y)}{\mu(X\mid Y)}\right],
\]
which we call the $\chisq$-cost by Alice.
The $\chisq$-cost by Bob, $\chis{\pi}{\mu,B}$, is defined similarly for the second term in the information cost (see Definition~\ref{def_chi_cost}).

This notion of the cost has the following benefits.
\begin{itemize}
	\item For a (deterministic) standard protocol with $C$ bits of communication, $\chis{\pi}{\mu,A}\leq 2^C$.
	Hence, it corresponds to the exponential of the communication cost.
	\item When conditioning on a small probability event $U$, we can essentially ensure that it increases by a factor of $O(1/\pi(U))$ (Lemma~\ref{lem_resolve_event} gives a more generalized statement).
	Effectively, this only adds $\log(1/\pi(U))$ to the communication cost, which becomes affordable.
\end{itemize}

Note that the mutual information is the expected KL-divergence, and the $\chisq$-cost is the expected $\chisq$-divergence (plus one).
Similarly, we also define an ``exponential version'' of $\theta(\pi)$, which we call the $\theta$-cost of $\pi$ (see Definition~\ref{def_theta_cost}).
It also ensures that the value does not increase significantly when conditioning on a small probability event.

On the other hand, going from mutual information to its ``exponential version'' loses many of its good properties, most importantly, the chain rule.
The next crucial observation is that the chain rule for mutual information in fact holds \emph{pointwisely}, which enables us to work with the $\chisq$-costs.

More specifically, let $X,Y,Z$ be three random variables with joint distribution $\pi$, the chain rules says $I(X; Y, Z)=I(X; Y)+I(X;Z\mid Y)$.
By writing the mutual information as an expectation, this is
\[
	\E\left[\log \left(\frac{\pi(Y, Z\mid X)}{\pi(Y, Z)}\right)\right]=\E\left[\log \left(\frac{\pi(Y\mid X)}{\pi(Y)}\right)\right]+\E\left[\log \left(\frac{\pi(Z\mid X, Y)}{\pi(Z\mid Y)}\right)\right].
\]
This equality holds pointwisely in the sense that for any concrete values $(x,y,z)$, the equality holds for the logarithms inside the expectation
\[
	\log \left(\frac{\pi(y, z\mid x)}{\pi(y, z)}\right)=\log \left(\frac{\pi(y\mid x)}{\pi(y)}\right)+\log \left(\frac{\pi(z\mid x, y)}{\pi(z\mid y)}\right)
\]
by the definition of conditional probability.

Therefore, the ``exponential version'' also holds pointwisely:
\[
	\frac{\pi(y, z\mid x)}{\pi(y, z)}=\frac{\pi(y\mid x)}{\pi(y)}\cdot \frac{\pi(z\mid x, y)}{\pi(z\mid y)}.
\]
This is what we use in replacement of the chain rule for mutual information.
See the next subsection for more details.

\subsection{Proof outline}
We now give an outline of the proof of the following statement: 
Given an $r$-round standard protocol $\pi$ for $f^{\oplus n}$ with communication cost $o(n\cdot C)$ that succeeds with advantage $\alpha^{o(n)}$ on the inputs sampled from $\mu^n$, we can obtain an $r$-round generalized protocol $\rho$ for $f$ with $\chisq$-costs $\approx 2^C$, $\theta$-cost $\approx 1/\alpha$ and advantage $\approx \alpha$.
We will then discuss how to convert such a generalized protocol to a standard protocol with low communication cost in the next subsection.

We first show that $\pi$ is also a generalized protocol with $\chisq$-cost $2^{o(nC)}$ and $\theta$-cost $1$ (in the proof of Lemma~\ref{lem_potential_standard}).
Next, we decompose $\pi$ into $\pi_{<n}$ for $f^{\oplus n-1}$ and $\pi_n$ for $f$, and prove that the product of the $\theta$-cost [resp. $\chisq$-costs] of $\pi_{<n}$ and $\pi_n$ is that of $\pi$ pointwisely (Section~\ref{sec_decompose}).
Now if the advantage of $\pi$ is not roughly evenly distributed, we will identify an event $U$ such that the advantage conditioned on $U$ is much higher than the average advantage, and more importantly, the advantage within $U$ becomes roughly evenly distributed (not concentrated on any small probability event in $U$) (Section~\ref{sec_identify_U}).
Conditioning on $U$ increases the advantage while also increases the $\theta$-cost and $\chisq$-costs, it turns out that they all increase by about the same factor.
Next, we partition the sample space of $\pi$ into $S_{\reducecost}, S_{\good}$ and $S_{\lowprob}$ such that
\begin{itemize}
	\item in $S_{\reducecost}$, $\pi_n$ has high $\theta$-cost or high $\chisq$-cost (excluding some corner cases), say $\geq 1/\alpha$ for $\theta$-cost or $\geq 2^C$ for $\chisq$-cost,
	\item in $S_{\good}$, $\pi_n$ has low $\theta$-cost and low $\chisq$-cost (also excluding some corner cases),
	\item $S_{\lowprob}$ is the rest, which will happen with very low probability.
\end{itemize}
Since the advantage is not concentrated on any small probability in $U$, then (at least) one of $S_{\reducecost}$ or $S_{\good}$ will have advantage about as high as the advantage of $U$.
If $S_{\reducecost}$ has the advantage as high as $U$, then we prove that by the pointwise equality for the costs, $\pi_{<n}\mid S_{\reducecost}$ must have a much smaller cost than $\pi\mid U$, while they have roughly the same advantage (Section~\ref{sec_high_cost}).
If $S_{\good}$ has the advantage as high as $U$, then if $\pi_n\mid S_{\good}$ has high advantage, then we obtain a desired generalized protocol for $f$ with low costs and high advantage; otherwise we prove that $\pi_{<n}\mid S_{\good}$ has a much higher advantage than $\pi\mid U$ (as the advantage of $\pi$ is roughly evenly distribution within $U$), while they have roughly the same costs (Section~\ref{sec_low_cost}).

To summarize the above argument, if we don't already find a desired generalized protocol for $f$, then when decrementing $n$ to $n-1$, we first condition on an event $U$, increasing costs and advantage simultaneously by about the same (while arbitrary) factor, then either we reduce the $\theta$-cost by a factor of $\geq 1/\alpha$, or we reduce the $\chisq$-costs by a factor of $\geq 2^C$, or we increase the advantage by a factor of $\geq 1/\alpha$.
Since we start with $\chisq$-costs $2^{o(nC)}$, $\theta$-cost $1$ and advantage $\alpha^{o(n)}$, we cannot repeat this for $n$ steps without finding a desired protocol for $f$.
More formally, we will measure the progress by using a potential function that depends on the costs and advantage of the current protocol, and show that each time we decrement from $k$ to $k-1$, how much the potential must decrease (Section~\ref{sec_setup}).

\subsection{Convert a generalized protocol to a standard protocol}\label{sec_overview_compression}
Finally, we need to show that the existence of a good generalized protocol implies the existence of a good standard protocol.
We prove that if an $r$-round generalized protocol $\rho$ has $\chisq$-costs $2^C$, $\theta$-cost $1/\alpha$ and advantage $\alpha$, then there is an $r$-round standard protocol $\tau$ with communication cost $\approx C$ and advantage $\approx \alpha^3$.
Together with what we summarized in the last subsection, we obtain the strong XOR lemma for $r$-round communication.

\cite{BR11} converts a \emph{standard} protocol $\rho$ with constant rounds to a standard protocol with communication matching the internal information cost of $\rho$.
Using a similar argument, we can convert $\rho$ to a standard protocol with communication $\approx C$.
By the convexity of $2^x$, $\chisq$-cost of $2^C$ implies internal information cost of at most $C$.
It turns out that the (almost) same argument applies in our case, for generalized protocol $\rho$.

Then the next crucial observation is that we can ensure the generalized protocol $\rho$ that we obtain from the arguments in the previous subsection has the \emph{rectangle property} with respect to $\mu$.
Roughly speaking, it means that for all transcripts $\bM$, if we look at the ratio of the probabilities $\frac{\rho(X, Y\mid \bM)}{\mu(X, Y)}$, it is a product function of $X$ and $Y$, i.e., it is equal to $g_A(X)\cdot g_B(Y)$ for some functions $g_A,g_B$ that may depend on $\bM$.
Note that a standard protocol has the rectangle property, since each message depends only on either $X$ or $Y$, and the same property holds even conditioned on any prefix of the transcript $M_{<i}$.
A generalized protocol may not have this property in general, but we can ensure that the protocol we obtain has this product structure conditioned on any \emph{complete} transcript $\bM$.

After generating a transcript $\bM$ using~\cite{BR11}, the rectangle property allows the players to locally ``re-adjust'' the probabilities (via rejection sampling) so that after the readjustment, the probability of a triple $(X, Y, \bM)$ is proportional to the ``right'' probability $\rho(X, Y, \bM)$, which in turn, gives the advantage proportional to that of $\rho$.

The probability that is sacrificed in the rejection sampling depends on how far $\rho$ is from a standard protocol, i.e., the $\theta$-cost of $\rho$.
It turns out that the above argument gives an overall advantage of at least $\alpha^2$ divided by the $\theta$-cost of $\rho$.
See Section~\ref{sec_compression} for the formal proof.

\section{Notations and Definitions for Generalized Protocols}
\subsection{Notations and standard probabilities}
Throughout the paper, all logarithms have base $2$.
We use $[n]$ to denote the set $\{1,\ldots,n\}$.
Let $f:\cX\times \cY\rightarrow\{0,1\}$ be a binary-valued function.
We use $f^{\oplus n}$ to denote the function $f^{\oplus n}:\cX^n\times \cY^n\rightarrow\{0,1\}$ such that
\[
	f^{\oplus n}(X_1,\ldots,X_n,Y_1,\ldots,Y_n)=\bigoplus_{i=1}^n f(X_i, Y_i),
\]
where $\oplus$ is the XOR operation.

Let $X$ be a vector $(X_1,\ldots,X_n)$.
We denote the prefix $(X_1,\ldots,X_i)$ by $X_{\leq i}$.
Similarly, $X_{<i}$ denotes $(X_1,\ldots,X_{i-1})$ and $X_{>i}$ denotes $(X_{i+1},\ldots,X_n)$.
For vectors $X$ where we start the index from $0$, $X_{\leq i}$ denotes $(X_0,\ldots,X_i)$.

\bigskip

Let $\pi$ be a distribution over triples $(X, Y, \bM)\in\cX\times \cY\times \cM$, where $\bM=(M_0,\ldots,M_r)$.
For an event $W\subseteq \cX\times \cY\times \cM$, we use $\pi(W)$ to denote its probability.
For a random variable $\bM$, we use $\pi(\bM)$ to denote the probability of $\bM$ in distribution $\pi$, which by itself is a random variable that depends on the value of $\bM$.
It is similar for multiple variables, e.g., $\pi(X, M_{<i})$ denotes the probability of $(X, M_{<i})$.

Let $S$ be a set of possible values of several variables, say, $S$ is a set of possible values of $(X, M_{<i})$.
We use $\pi(S)$ to denote the probability that $(X, M_{<i})\in S$, i.e., $\pi(S)=\Pr_{\pi}\left[(X, M_{<i})\in S\right]=\pi(\{(X, Y, \bM): (X, M_{<i})\in S\})$.
When there is no ambiguity, we may abuse the notation, and use $S$ to denote the event that $(X, M_{<i})\in S$, which is the set $\{(X, Y, \bM): (X, M_{<i})\in S\}$, e.g., if $T$ is a set of possible values of $(Y, M_{<j})$, then $S\cap T$ is the event that $(X, M_{<i})\in S\wedge (Y, M_{<j})\in T$, which is the set $\{(X, Y, \bM): (X, M_{<i})\in S\wedge (Y, M_{<j})\in T\}$.

\bigskip
The $\chisq$-divergence of two distributions is defined as follows.
\begin{definition}[$\chisq$-divergence]
	Let $P$ and $Q$ be two distributions over a sample space $\cX$.
	The $\chi^2$-divergence from $Q$ to $P$ is
	\[
		\Dchisq{P}{Q}=\sum_{x\in\cX} \frac{P(x)^2}{Q(x)}-1=\E_{x\sim P}\left[\frac{P(x)}{Q(x)}\right]-1.
	\]
\end{definition}

The KL-divergence of two distributions is defined as follows.
\begin{definition}[KL-divergence]
	Let $P$ and $Q$ be two distributions over a sample space $\cX$.
	The KL-divergence from $Q$ to $P$ is
	\[
		\DKL{P}{Q}=\sum_{x\in\cX}P(x)\log \left(\frac{P(x)}{Q(x)}\right)=\E_{x\sim P}\left[\log \left(\frac{P(x)}{Q(x)}\right)\right].
	\]
\end{definition}

A simple calculation gives the following proposition.
\begin{proposition}\label{prop_xor}
	Let $R_1,R_2\in\{0,1\}$ be two \emph{independent} random variables such that $\Pr[R_1=0]=\frac{1}{2}+\frac{\sigma_1}{2}$ and $\Pr[R_2=0]=\frac{1}{2}+\frac{\sigma_2}{2}$.
	Then $\Pr[R_1\oplus R_2=0]=\frac{1}{2}+\frac{\sigma_1\sigma_2}{2}$.
\end{proposition}


\subsection{Generalized communication protocols}

For most standard communication protocols discussed in this paper, we pair it with an input distribution, and study the joint distribution.

\begin{definition}[standard protocols]\label{def_standard}
	An $r$-round \emph{standard protocol} $\pi$ for input distribution $\mu$ is a distribution over triples
	\[
		(X, Y, \bM)\in \cX\times \cY\times \cM,
	\]
	where the transcript $\bM=(M_0,\ldots,M_r)$, and each $M_i$ is chosen from a prefix-free set of strings that only depends on $M_{<i}$.
	Moreover, $(X, Y)\sim \mu$; the public random string $M_0$ is independent of $(X, Y)$; for odd $i\geq 1$, $M_i$ (a message by Alice) is independent of $Y$ conditioned on $X$ and $M_{<i}$; for even $i\geq 1$, $M_i$ (a message by Bob) is independent of $X$ conditioned on $Y$ and $M_{<i}$.
	The output of $\pi$ is a function of $\bM$.
\end{definition}

Now we define $\suc_{\mu}(f; C_A, C_B, r)$ to be the maximum success probability of a protocol computing $f$ under the communication cost constraints.
\begin{definition}
	Let $f: \cX\times\cY\rightarrow \{0,1\}$ be a function, and $\mu$ be a distribution over $\cX\times\cY$.
	For $C_A,C_B, r\geq 1$, let
	\[
		\suc_{\mu}(f; C_A, C_B, r)
	\]
	be the supremum over all $r$-round \emph{standard} protocols $\pi$ where Alice sends at most $C_A$ bits and Bob sends at most $C_B$ bits in a round, the probability that the output of $\pi$ is equal to $f(X, Y)$ when $(X, Y)$ is sampled from $\mu$.
\end{definition}
\begin{remark}
	Without loss of generality, we may assume that $M_r\in\{0,1\}$.
	This is because the output of the protocol is a function of $\bM$.
	Instead of $M_r$, we could always let Bob send the output, which has only one bit.
\end{remark}

Next, we define generalized protocols.

\begin{definition}[generalized protocols]\label{def_gen_protocol}
An $r$-round \emph{generalized protocol} $\pi$ is a distribution over triples
\[
	(X, Y, \bM)\in \cX\times\cY\times\cM,
\]
where $\bM=(M_0, M_1,\ldots,M_r)$, and each $M_i$ is chosen from a prefix-free set of strings that only depends on $M_{<i}$.
\end{definition}
In this paper, we also only consider generalized protocols with $M_r\in\{0,1\}$.

Clearly, a standard protocol is also a generalized protocol.
One still should think $M_0$ as the public random bits, and think $M_i$ as a message sent by Alice if $i$ is odd, and sent by Bob if $i$ is even.
The messages and public random bits are allowed to be arbitrarily correlated with both players' inputs.

We do not explicitly define the output of a generalized protocol in this paper.
When we study the correctness of a generalized protocol when computing some function $f$, we characterize it using the \emph{advantage}.
\begin{definition}[advantage]
	Let $\pi$ be a generalized protocol, and $f$ be a binary-valued random variable (e.g., $f=f(X,Y)$ for $f:\cX\times \cY\rightarrow\{0,1\}$).
	The advantage of $\pi$ for $f$ conditioned on $W$ is
	\[
		\adv_{\pi}(f\mid W):=\left|2\pi(f=0\mid W)-1\right|=\left|2\pi(f=1\mid W)-1\right|.
	\]
	We may omit the subscript $\pi$ when there is no ambiguity.
\end{definition}
Note that fixing $\pi$ and $f$, $\adv_{\pi}(f\mid W)$ is a function of $W$.
For a standard protocol with input distribution $\mu$, the larges probability that the output can equal $f$ when the transcript is $\bM$ is 
\[
	\frac{1}{2}+\frac{1}{2}\cdot \adv_{\pi}(f(X, Y)\mid \bM).
\]
Thus, the overall success probability is always at most
\[
	\frac{1}{2}+\frac{1}{2}\cdot \E_{\bM\sim\pi}\left[\adv_{\pi}(f(X, Y)\mid \bM)\right].
\]
For general protocols, we will also use $\E_{\bM\sim\pi}\left[\adv_{\pi}(f(X, Y)\mid \bM)\right]$ to characterize the success probability.

The (conditional) advantage is superadditive when weighted by the probability of the condition.

\begin{lemma}\label{lem_adv_super_additivity}
	Let $W_1, W_2$ be \emph{disjoint} events and $\bR$ be a set of random variables, then
	\begin{align*}
		&\pi(W_1\cup W_2)\cdot\E_{\pi\mid W_1\cup W_2}\left[\adv(f(X, Y)\mid \bR, W_1\cup W_2)\right] \\
		&\quad\,\,\leq \pi(W_1)\cdot\E_{\pi\mid W_1}\left[\adv(f(X, Y)\mid \bR, W_1)\right]+\pi(W_2)\cdot\E_{\pi\mid W_2}\left[\adv(f(X, Y)\mid \bR, W_2)\right].
	\end{align*}
\end{lemma}
\begin{proof}
	By definition, we have
	\begin{align*}
		&\kern-2em\pi(W_1)\cdot\E_{\pi\mid W_1}\left[\adv(f(X, Y)\mid \bR, W_1)\right] +\pi(W_2)\cdot\E_{\pi\mid W_2}\left[\adv(f(X, Y)\mid \bR, W_2)\right] \\
		&=\pi(W_1)\cdot \sum_{\bR}\pi(\bR\mid W_1)\cdot \left|2\pi(f(X, Y)=1\mid \bR, W_1)-1\right| \\
		&\strut\qquad+\pi(W_2)\cdot \sum_{\bR}\pi(\bR\mid W_2)\cdot \left|2\pi(f(X, Y)=1\mid \bR, W_2)-1\right| \\
		&=\sum_{\bR}\left|2\pi(f(X, Y)=1, \bR, W_1)-\pi(\bR, W_1)\right|+\sum_{\bR}\left|2\pi(f(X, Y)=1, \bR, W_2)-\pi(\bR, W_2)\right|
		\intertext{which by the fact that $W_1$ and $W_2$ are disjoint, is}
		&\geq \sum_{\bR}\left|2\pi(f(X, Y)=1, \bR, W_1\cup W_2)-\pi(\bR, W_1\cup W_2)\right| \\
		&=\pi(W_1\cup W_2)\cdot \E_{\pi\mid W_1\cup W_2}\left[\adv(f(X, Y)\mid \bR, W_1\cup W_2)\right].
	\end{align*}
\end{proof}

The following proposition states that knowing more could only increase the expected advantage.
\begin{proposition}\label{prop_adv_know_more}
	Let $\bR_1,\bR_2$ be two random variables, then
	\[
		\E_{\bR_1\sim \pi}\left[\adv(f\mid \bR_1)\right]\leq\E_{\bR_1, \bR_2\sim \pi}\left[\adv(f\mid \bR_1,\bR_2)\right].
	\]
\end{proposition}
\begin{proof}
	We have
	\begin{align*}
		&\kern-2em\E_{\bR_1, \bR_2\sim \pi}\left[\adv(f\mid \bR_1,\bR_2)\right] \\
		&=\sum_{\bR_1,\bR_2}\pi(\bR_1,\bR_2)\cdot \left|2\pi(f=1\mid \bR_1,\bR_2)-1\right| \\
		&=\sum_{\bR_1}\pi(\bR_1)\sum_{\bR_2}\pi(\bR_2\mid \bR_1)\cdot \left|2\pi(f=1\mid \bR_1,\bR_2)-1\right| \\
		&\geq \sum_{\bR_1}\pi(\bR_1)\cdot \left|\sum_{\bR_2}\pi(\bR_2\mid \bR_1)\cdot 2\pi(f=1\mid \bR_1,\bR_2)-\sum_{\bR_2}\pi(\bR_2\mid \bR_1)\right| \\
		&=\sum_{\bR_1}\pi(\bR_1)\cdot \left|2\pi(f=1\mid \bR_1)-1\right| \\
		&=\E_{\bR_1\sim \pi}\left[\adv(f\mid \bR_1)\right].
	\end{align*}
\end{proof}

\subsection{The $\theta$-cost and $\chisq$-costs}
In a standard protocol, Alice's message must be independent of Bob's input conditioned on Alice's input and the previous messages, and vice versa, while we allow arbitrary correlation in a generalized protocol.
The \emph{$\theta$-cost} of a generalized protocol measures this correlation.
\begin{definition}[$\theta$-cost]\label{def_theta_cost}
	The $\theta$-cost of $\pi$ with respect to $\mu$ at $(X, Y, \bM)$ is
	\begin{align*}
		\thec{\pi}{\mu}[X, Y, \bM]&:=\frac{\pi(X, Y\mid M_0)}{\mu(X, Y)}\cdot\prod_{\odd i\in [r]}\frac{\pi(M_i\mid X, Y, M_{<i})}{\pi(M_i\mid X, M_{<i})}\cdot \prod_{\even i\in[r]}\frac{\pi(M_i\mid X, Y, M_{<i})}{\pi(M_i\mid Y, M_{<i})} \\
		&=\frac{\pi(X, Y, \bM)}{\pi(M_0)\cdot \mu(X, Y)\cdot\prod_{\odd i\in[r]}\pi(M_i\mid X, M_{<i})\cdot\prod_{\even i\in[r]}\pi(M_i\mid Y, M_{<i})}.
	\end{align*}
	The $\theta$-cost of $\pi$ with respect to $\mu$ is
	\begin{align*}
		\thec{\pi}{\mu}&:=\E_{(X, Y, \bM)\sim\pi}\left[\thec{\pi}{\mu}[X,Y,\bM]\right].
	\end{align*}
	For an event $W$, the {$\theta$-cost} of $\pi$ respect to $\mu$ conditioned on $W$ is
	\begin{align*}
		\theta_{\mu}(\pi\mid W):=\E_{(X, Y, \bM)\sim\pi\mid W}\left[\thec{\pi}{\mu}[X, Y, \bM]\right].
	\end{align*}
\end{definition}

\begin{remark}
	We emphasize that $\thec{\pi\mid W}{\mu}$ is different from $\thec{\pi_W}{\mu}$ for $\pi_W$ being the distribution of $\pi$ conditioned on $W$.
	According to the definitions, although $(X, Y, \bM)$ is sampled from $\pi\mid W$ in both cases, the quantity inside the expectation is different.
	For $\thec{\pi\mid W}{\mu}$, we still measure the $\theta$-cost at $(X, Y, \bM)$ according to distribution $\pi$, while for $\thec{\pi_W}{\mu}$, we measure the $\theta$-cost at $(X, Y, \bM)$ according to $\pi_W$.
\end{remark}

\begin{remark}
	Let $\tau$ be the protocol obtained by ``making $\pi$ standard.''
	That is, $\tau(X,Y)$ is equal to $\mu(X, Y)$; $\tau(M_0)$ is $\pi(M_0)$, independent of $(X, Y)$.
	Each odd $M_i$ is sampled according to $\pi(M_i\mid X, M_{<i})$ independent of $Y$, and each even $M_i$ is sampled according to $\pi(M_i\mid Y, M_{<i})$ independent of $X$.
	Then $\tau$ is a standard protocol such that
	\[
		\tau(X,Y,\bM)={\pi(M_0)\cdot \mu(X, Y)\cdot\prod_{\odd i\in[r]}\pi(M_i\mid X, M_{<i})\cdot\prod_{\even i\in[r]}\pi(M_i\mid Y, M_{<i})}.
	\]
	The $\theta$-cost of $\pi$ is simply the $\chi^2$-divergence from $\tau$ to $\pi$ plus one.
\end{remark}

By the above connection to $\chisq$-divergence, we have the following proposition.
\begin{proposition}\label{prop_expect_theta_inv}
	For any protocol $\pi$, we have
	\[
		\E_{\pi}\left[\thec{\pi}{\mu}[X, Y, \bM]^{-1}\right]=1.
	\]
\end{proposition}
\begin{proof}
	Let $\tau$ be the protocol by making $\pi$ standard as described in the remark above. 
	Then we have $$\thec{\pi}{\mu}[X, Y, \bM]=\frac{\pi(X, Y, \bM)}{\tau(X, Y, \bM)}.$$
	Therefore,
	\[
		\E_{\pi}\left[\thec{\pi}{\mu}[X, Y, \bM]^{-1}\right]=\sum_{(X, Y, \bM)}\pi(X, Y, \bM)\cdot \frac{\tau(X, Y, \bM)}{\pi(X, Y, \bM)}=1.
	\]
\end{proof}

By standard bounds on conditional probabilities, we have the following bound on the conditional cost.
\begin{proposition}\label{prop_thec_condition_increase}
	For events $W_1, W_2$, we have
	\[
		\thec{\pi\mid W_1\cap W_2}{\mu}\leq \frac{\thec{\pi\mid W_1}{\mu}}{\pi(W_2\mid W_1)}.
	\]
\end{proposition}
\begin{proof}
	Since $\thec{\pi}{\mu}[X, Y, \bM]$ is nonnegative, we have
	\begin{align*}
		\thec{\pi\mid W_1\cap W_2}{\mu}&=\E_{(X, Y, \bM)\sim\pi\mid W_1\cap W_2}\left[\thec{\pi}{\mu}[X, Y, \bM]\right] \\
		&=\sum_{(X, Y, \bM)} \pi(X, Y, \bM\mid W_1\cap W_2)\cdot \thec{\pi}{\mu}[X, Y, \bM] \\
		&\leq \sum_{(X, Y, \bM)} \frac{\pi(X, Y, \bM\mid W_1)}{\pi(W_2\mid W_1)}\cdot \thec{\pi}{\mu}[X, Y, \bM] \\
		&=\frac{\thec{\pi\mid W_1}{\mu}}{\pi(W_2\mid W_1)}.
	\end{align*}
\end{proof}

Next, the \emph{$\chisq$-cost} measures the ``communication cost'' of a protocol: how different Alice's input becomes in Bob's view at the end of the communication compared to that in the input distribution.
\begin{definition}[$\chisq$-cost]\label{def_chi_cost}
	Let $\pi$ be a generalized protocol, and $\mu$ be a distribution over the inputs.
	The $\chi^2$-cost of $\pi$ by Alice with respect to $\mu$ at $(X, Y, \bM)$ is
	\begin{align*}
		\chis{\pi}{\mu,A}[X, Y, \bM]&:=\frac{\pi(X\mid \bM, Y)}{\mu(X\mid Y)};
	\end{align*}
	the $\chi^2$-cost of $\pi$ by Bob with respect to $\mu$ at $(X, Y, \bM)$ is
	\begin{align*}
		\chis{\pi}{\mu,B}[X, Y, \bM]&:=\frac{\pi(Y\mid \bM, X)}{\mu(Y\mid X)}.
	\end{align*}
	The $\chi^2$-costs of $\pi$ with respect to $\mu$ are
	\begin{align*}
		\chis{\pi}{\mu,A}&:=\E_{(X,Y,\bM)\sim\pi}\left[\chis{\pi}{\mu,A}[X, Y, \bM]\right], \\
		\chis{\pi}{\mu,B}&:=\E_{(X,Y,\bM)\sim\pi}\left[\chis{\pi}{\mu,B}[X, Y, \bM]\right].
	\end{align*}
	For an event $W$, the $\chisq$-costs of $\pi$ with respect to $\mu$ conditioned on $W$ are
	\begin{align*}
		\chis{\pi\mid W}{\mu,A}&:=\E_{(X,Y,\bM)\sim\pi\mid W}\left[\chis{\pi}{\mu,A}[X, Y, \bM]\right],\\
		\chis{\pi\mid W}{\mu,B}&:=\E_{(X,Y,\bM)\sim\pi\mid W}\left[\chis{\pi}{\mu,B}[X, Y, \bM]\right].
	\end{align*}
\end{definition}
\begin{remark}
	Similar to the $\theta$-cost, $\chis{\pi\mid W}{\mu,A}$ is also different from $\chis{\pi_W}{\mu,A}$.
	The $\chi^2$-cost of $\pi$ by Alice is the expected $\chi^2$-divergence from $\mu_{X\mid Y}$ to $\pi_{X\mid Y, \bM}$ plus one.
	Similarly, the $\chi^2$-cost of $\pi$ by Bob is the expected $\chi^2$-divergence from $\mu_{Y\mid X}$ to $\pi_{Y\mid X, \bM}$ plus one.
	Observe that for standard protocols, if we measure the expected KL-divergence instead of the $\chisq$-divergence, then we obtain the internal information costs:
	\[
		\sum_{\odd i}I(X; M_i\mid Y, M_{<i})\qquad\textrm{and}\qquad\sum_{\even i}I(Y; M_i\mid X, M_{<i}).
	\]
\end{remark}

Similar proofs to Proposition~\ref{prop_expect_theta_inv} and Proposition~\ref{prop_thec_condition_increase} give us the following two propositions.
\begin{proposition}\label{prop_expect_chis_inv}
	For any protocol $\pi$, we have
	\[
		\E_{\pi}\left[\chis{\pi}{\mu,A}[X, Y, \bM]^{-1}\right]=1,
	\]
	and
	\[
		\E_{\pi}\left[\chis{\pi}{\mu,B}[X, Y, \bM]^{-1}\right]=1.
	\]
\end{proposition}

\begin{proposition}\label{prop_chis_condition_increase}
	For any events $W_1,W_2$, we have
	\[
		\chis{\pi\mid W_1\cap W_2}{\mu,A}\leq \frac{\chis{\pi\mid W_1}{\mu,A}}{\pi(W_2\mid W_1)},
	\]
	and
	\[
		\chis{\pi\mid W_1\cap W_2}{\mu,B}\leq \frac{\chis{\pi\mid W_1}{\mu,B}}{\pi(W_2\mid W_1)}.
	\]
\end{proposition}

\subsection{Rectangle properties in generalized protocols}

We will maintain the \emph{rectangle property} for the generalized protocols throughout the proof.
\begin{definition}[Rectangle property]\label{def_rect}
	A generalized protocol $\pi$ has the \emph{rectangle property} with respect to $\mu$, if there exists nonnegative functions $g_1:\cX\times \cM\rightarrow \mathbb{R},g_2:\cY\times \cM\rightarrow \mathbb{R}$ such that
	\[
		\pi(X, Y, \bM)={\mu(X, Y)\cdot g_1(X, \bM)\cdot g_2(Y, \bM)}.
	\]
	Let $W$ be an event, $(\pi \mid W)$ has the rectangle property with respect to $\mu$ if there exists nonnegative functions $g_1:\cX\times \cM\rightarrow \mathbb{R},g_2:\cY\times \cM\rightarrow \mathbb{R}$ such that
	\[
		\pi(X, Y, \bM\mid W)=\mu(X, Y)\cdot g_1(X, \bM)\cdot g_2(Y, \bM).
	\]
\end{definition}

Equivalently, $\pi$ has the rectangle property if for every transcript $\bM$, the posterior distribution $\pi_{X,Y\mid \bM}$ is equal to $\mu$ rescaled by some product function with one factor depending only on $X$ and another factor depending only on $Y$.
Note that this property holds for any standard protocol, since each message $M_i$ conditioned on $M_{<i}$ only depends on one of $X$ and $Y$.
Hence, for standard protocols, we have such product structure even conditioned on any prefix $M_{<i}$.

When decomposing a protocol for $k$ instances, we need the following definition of the \emph{partial rectangle property}.
\begin{definition}[Partial rectangle property]\label{def_partial_rect}
	Let $\pi$ be a generalized protocol such that $X=(X_1,\ldots,X_k)$ and $Y=(Y_1,\ldots,Y_k)$.
	$\pi$ satisfies the \emph{partial rectangle property} with respect to $\mu^k$ if there exists nonnegative functions $g_1, g_2, g_3$ such that
	\[
		\pi(X, Y, \bM)=\mu^{k}(X, Y)\cdot g_1(X, \bM)\cdot g_2(Y, \bM)\cdot g_3(X_k, Y_{<k}, \bM).
	\]
	Let $W$ be an event, $(\pi \mid W)$ has the partial rectangle property with respect to $\mu^k$ if there exists nonnegative functions $g_1,g_2,g_3$ such that
	\[
		\pi(X, Y, \bM\mid W)=\mu^{k}(X, Y)\cdot g_1(X, \bM)\cdot g_2(Y, \bM)\cdot g_3(X_k, Y_{<k}, \bM).
	\]
\end{definition}

\begin{proposition}\label{prop_independent_x<k_yk}
	If $\pi$ has the partial rectangle property, then $X_{<k}$ and $Y_k$ are independent conditioned on $X_k, Y_{<k}, \bM$;
	If $\pi\mid W$ has the partial rectangle property, then $X_{<k}$ and $Y_k$ are independent conditioned on $X_k, Y_{<k}, \bM, W$.
\end{proposition}
\begin{proof}
	If $\pi$ has the partial rectangle property, then
	\begin{align*}
		&\kern-2em\pi(X_{<k},Y_k\mid X_k, Y_{<k}, \bM) \\
		&=\mu^{k}(X, Y)\cdot g_1(X, \bM)\cdot g_2(Y, \bM)\cdot g_3(X_k, Y_{<k}, \bM)\cdot \pi(X_k, Y_{<k}, \bM)^{-1} \\
		&=\left(\mu^{k-1}(X_{<k}, Y_{<k})\cdot g_1(X, \bM)\right)\cdot \left(\mu(X_k, Y_k)\cdot g_2(Y, \bM)\cdot g_3(X_k, Y_{<k}, \bM)\cdot \pi(X_k, Y_{<k}, \bM)^{-1}\right).
	\end{align*}
	Note that given $(X_k, Y_{<k}, \bM)$, the first factor only depends on $X_{<k}$, and the second factor only depends on $Y_k$.
	By normalizing the two factors, we obtain that
	\[
		\pi(X_{<k},Y_k\mid X_k, Y_{<k}, \bM)=\pi(X_{<k}\mid X_k, Y_{<k}, \bM)\cdot \pi(Y_k\mid X_k, Y_{<k}, \bM).
	\]
	The proof for $\pi\mid W$ is almost identical.
	We omit the details.
\end{proof}

\bigskip

For a protocol $\pi$, we define the follow sets that are related to the rectangle property and the partial rectangle property.
\begin{definition}\label{def_cU_pi}
	Let $\cU_{X, M}(\pi)$ be the set consisting of all possible pairs $(X, \bM)$.
	Let $\cU_{Y, M}(\pi)$ be the set consisting of all possible pairs $(Y, \bM)$.
	Let $\cU_{X_k, Y_{<k}, M}(\pi)$ be the set consisting of all possible triples $(X_k, Y_{<k}, \bM)$.

	Let $\cS_{\rect}(\pi)$ be the collection of all possible events $S$ such that there exist $S_{X, M}\subseteq \cU_{X, M}, S_{Y, M}\subseteq \cU_{Y, M}$, and $$S=\{(X, Y, \bM): (X, \bM)\in S_{X, M}\wedge (Y, \bM)\in S_{Y, M}\}.$$

	Let $\cS_\pa(\pi)$ be the collection of all possible events $S$ such that there exist $S_{X, M}\subseteq \cU_{X, M}, S_{Y, M}\subseteq \cU_{Y, M}, S_{X_k, Y_{<k}, M}\subseteq \cU_{X_k, Y_{<k}, M}$, and $$S=\{(X, Y, \bM): (X, \bM)\in S_{X, M}\wedge (Y, \bM)\in S_{Y, M}\wedge (X_k, Y_{<k}, \bM)\in S_{X_k, Y_{<k}, M}\}.$$

	We may omit $\pi$ and use $\cS_{\rect}$, $\cS_{\pa}$ when there is no ambiguity.
\end{definition}

Intuitively, $\cS_{\rect}(\pi)$ is the collection of events conditioned on which, $\pi$ remains to have the rectangle property.
Similarly, $\cS_{\pa}(\pi)$ is the collection of events conditioned on which, $\pi$ remains to have the partial rectangle property.

\begin{proposition}\label{prop_cs_partial}
	We have the following properties about $\cS_{\rect}$ and $\cS_{\pa}$:
	\begin{enumerate}[(i)]
		\item if $\pi$ has the rectangle property, then for any $S\in\cS_{\rect}$, $(\pi\mid S)$ has the rectangle property;
		\item if $\pi$ has the partial rectangle property, then for any $S\in\cS_{\pa}$, $(\pi\mid S)$ has the partial rectangle property;
		\item $\cS_{\rect}\subseteq \cS_{\pa}$;
		\item both $\cS_{\rect}$ and $\cS_{\pa}$ are closed under intersection.
	\end{enumerate}
\end{proposition}
\begin{proof}
	For (i), suppose $S=\{(X, Y, \bM): (X, \bM)\in S_{X, M}\wedge (Y, \bM)\in S_{Y, M}\}$.
		Then $$\pi(X, Y, \bM\mid S)=\pi(X, Y, \bM)\cdot \mathbf{1}_{S_{X, M}}(X, \bM)\cdot \mathbf{1}_{S_{Y, M}}(Y, \bM)\cdot \pi(S)^{-1}.$$
		Thus, if $\pi$ has the rectangle property, then $(\pi\mid S)$ has the rectangle property.
	
	Similarly for (ii), suppose 
		\[
		S=\{(X, Y, \bM): (X, \bM)\in S_{X, M}\wedge (Y, \bM)\in S_{Y, M}\wedge (X_k, Y_{<k}, \bM)\in S_{X_k, Y_{<k}, M}\}.
		\]
		Then \[
			\pi(X, Y, \bM\mid S)=\pi(X, Y, \bM)\cdot \mathbf{1}_{S_{X, M}}(X, \bM)\cdot \mathbf{1}_{S_{Y, M}}(Y, \bM)\cdot \mathbf{1}_{S_{X_k,Y_{<k},M}}(X_k, Y_{<k}, \bM)\cdot \pi(S)^{-1}.
		\]
		Thus, if $\pi$ has the partial rectangle property, then $(\pi\mid S)$ has the partial rectangle property.

	(iii) and (iv) follow from the definitions.
\end{proof}


\section{Main Setup}\label{sec_setup}

In this section, we set up the main framework for proving our main theorems.

\begin{restate}[Theorem~\ref{thm_main_nondist}]
	\thmmainnondistcont
\end{restate}

\begin{restate}[Theorem~\ref{thm_main}]
	\thmmaincont
\end{restate}

We first show that Theorem~\ref{thm_main} implies Theorem~\ref{thm_main_nondist}.

\begin{proof}[Proof of Theorem~\ref{thm_main_nondist}]
Fix a function $f$, suppose there is an $r$-round protocol $\pi$ for $f^{\oplus n}$ with communication cost $T$ and success probability $1/2+2^{-n}$.
Let $\alpha=r^{-2cr}$ for a sufficiently large $c$, then $\pi$ has success probability more than $1/2+\alpha^{2^{-12}n}/2$.
By setting $C_A=C_B=2^8r\cdot T/n+2c\log(r/\alpha)=O(r\cdot T/n+r\log r)$, Theorem~\ref{thm_main} implies $\suc_{\mu}(f; C_A, C_B, r)>1/2+\alpha/2$, i.e., for distribution $\mu$, there is an $r$-round protocol with communication cost at most $O(r\cdot T/n+r\log r)$ in each round and success probability at least $1/2+r^{-O(r)}$.

Since this holds for \emph{any} $\mu$, by Yao's minimax lemma, there is an $r$-round randomized protocol with $O(r\cdot T/n+r\log r)$ communication in each round and success probability at least $1/2+r^{-O(r)}$ for all inputs.
By simply running such a protocol $r^{O(r)}$ times in parallel and outputting the majority, we obtain an $r$-round protocol with $r^{O(r)}\cdot (T/n+1)$ total communication and success probability $2/3$.
Thus, we obtain $\bR_{2/3}^{(r)}(f)\leq r^{O(r)}\cdot (\bR_{1/2+2^{-n}}(f^{\oplus n})/n+1)$.
Rearranging the terms gives Theorem~\ref{thm_main_nondist}.
\end{proof}

In the rest of the paper, we will focus on proving Theorem~\ref{thm_main}.
Let us fix a sufficiently large constant $c>0$, parameters $C_A,C_B,r,\alpha$, function $f$ and input distribution $\mu$ satisfying its premises.
As mentioned in Section~\ref{sec_overview}, we will first define a potential function based on the costs and advantage, and then show that the potential function value decreases as we decrement $n$.

\begin{definition}[Potential functions]\label{def_potential}
	For an $r$-round generalized protocol $\pi$ for $f^{\oplus n}$ and an event $W$, we define the potential function $\phi_n(\pi\mid W)$ (and $\phi_n^{\cost},\phi_n^{\adv}$) as follows:
	\begin{align*}
		\phi_n(\pi\mid W)&=\underbrace{\log \thec{\pi\mid W}{\mu^n}+\frac{\log(1/\alpha)}{C_A-c\log(r/\alpha)}\cdot \log \chis{\pi\mid W}{\mu,A}+\frac{\log(1/\alpha)}{C_B-c\log(r/\alpha)}\cdot \log\chis{\pi\mid W}{\mu,B}}_{\phi_n^{\cost}(\pi\mid W)} \\
		&\qquad+\underbrace{32\log \left(\E_{\pi\mid W}\left[\adv_{\pi}(f^{\oplus n}(X, Y)\mid \bM, W)\right]^{-1}\right)}_{\phi_n^{\adv}(\pi\mid W)}.
	\end{align*}
	We also define $\phi_{n, \pa}(\pi\mid W)$ (and $\phi_{n,\pa}^{\adv}$) as follows:
	\begin{align*}
		\phi_{n, \pa}(\pi\mid W)&=\log \thec{\pi\mid W}{\mu^n}+\frac{\log(1/\alpha)}{C_A-c\log(r/\alpha)}\cdot \log \chis{\pi\mid W}{\mu,A}+\frac{\log(1/\alpha)}{C_B-c\log(r/\alpha)}\cdot \log\chis{\pi\mid W}{\mu,B} \\
		&\qquad+\underbrace{32\log \left(\E_{\pi\mid W}\left[\adv_{\pi}(f^{\oplus n}(X, Y)\mid X_n, Y_{<n}, \bM, W)\right]^{-1}\right)}_{\phi_{n,\pa}^{\adv}(\pi\mid W)}.
	\end{align*}
	When $W$ is the whole sample space, we may simply write $\phi_n(\pi)$ or $\phi_{n,\pa}(\pi)$.
\end{definition}

The first three terms in both potential functions $\phi_{n}^{\cost}$ are the (normalized) costs of $\pi$.
They are small if $\pi$ has low $\theta$-cost and low $\chisq$-costs.
The last term in both potential functions depends on the expected advantage.
$\phi_n$ uses the standard advantage, while $\phi_{n,\pa}$ uses the advantage conditioned not only on the transcript, but also $X_n$ and $Y_{<n}$.
As we will see later, it is used when decomposing $\pi$.
The last term is small if the protocol has high advantage.
By Proposition~\ref{prop_adv_know_more}, knowing more could only increase the expected advantage.
Hence, $\phi_{n,\pa}(\pi)$ is always at most $\phi_n(\pi)$.

We have the following lower bound on the potential of $\pi$ conditioned on $W$.
In particular, when $W$ is the whole sample space, the potential function is nonnegative.

\begin{lemma}\label{lem_potential_low_bound}
	For any $\pi$, event $W$ and any $n\geq 1$, we must have
	\[
		\phi_{n}(\pi\mid W)\geq -3\log(1/\pi(W)).
	\]
\end{lemma}
\begin{proof}
	For the $\theta$-cost, by the convexity of $1/x$, we have
	\begin{align*}
		\thec{\pi\mid W}{\mu^n}^{-1}&=\E_{\pi\mid W}\left[\thec{\pi}{\mu^n}[X, Y, \bM]\right]^{-1} \\
		&\leq \E_{\pi\mid W}\left[\thec{\pi}{\mu^n}[X, Y, \bM]^{-1}\right]
		\intertext{which by the fact that $\thec{\pi}{\mu^n}[X, Y, \bM]$ is nonnegative, is}
		&\leq \pi(W)^{-1}\cdot \E_{\pi}\left[\thec{\pi}{\mu^n}[X, Y, \bM]^{-1}\right]
		\intertext{which by Proposition~\ref{prop_expect_theta_inv}, is}
		&=\pi(W)^{-1}.
	\end{align*}
	Hence, $\log\thec{\pi\mid W}{\mu^n}\geq -\log(1/\pi(W))$.
	Similarly, we also have $\log\chis{\pi\mid W}{\mu^n,A}\geq -\log(1/\pi(W))$, and $\log\chis{\pi\mid W}{\mu^n,B}\geq -\log(1/\pi(W))$.
	By the fact that $\log(1/\alpha)\leq C_A-c\log(r/\alpha)$ and $\log(1/\alpha)\leq C_B-c\log(r/\alpha)$, we have
	\[
		\phi_n^{\adv}(\pi\mid W)\geq -3\log(1/\pi(W)).
	\]

	The advantage is always at most $1$.
	Therefore, the last term is nonnegative.
	Hence, $\phi_n(\pi\mid W)\geq -3\log(1/\pi(W))$.
\end{proof}

The following lemma shows an upper bound on the potential of a \emph{deterministic} standard protocol $\pi$ computing $f^{\oplus n}$.

\begin{lemma}\label{lem_potential_standard}
	Let $\pi$ be a \emph{deterministic} standard protocol where Alice sends at most $T_A$ bits in each (odd) round and Bob sends at most $T_B$ bits in each (even) round.
	If it computes $f^{\oplus n}$ with probability $\frac{1}{2}+\frac{\sigma}{2}$ under input distribution $\mu^n$, then
	\[
		\phi_n(\pi)\leq \lceil r/2\rceil \cdot \frac{T_A\cdot\log(1/\alpha)}{C_A-c\log(r/\alpha)}+\lfloor r/2\rfloor\cdot\frac{T_B\cdot\log(1/\alpha)}{C_B-c\log(r/\alpha)}+32\log (1/\sigma).
	\]
\end{lemma}

\begin{proof}
	By the property of a standard protocol, $\thec{\pi}{\mu^n}[X, Y, \bM]=1$ for any $X, Y, \bM$ in the support of $\pi$.
	Hence, $\log \thec{\pi}{\mu^n}=0$.

	For the $\chisq$-cost by Alice, we have
	\begin{align*}
		\chis{\pi}{\mu^n,A}&=\E_{(X,Y,\bM)\sim \pi}\left[\frac{\pi(X\mid \bM, Y)}{\mu^n(X\mid Y)}\right] \\
		&=\E_{(X,Y,\bM)\sim \pi}\left[\frac{\pi(X\mid \bM, Y)}{\pi(X\mid Y)}\right] \\
		&=\sum_{(X, Y, \bM)}\frac{\pi(X,Y, \bM)\pi(X\mid \bM, Y)}{\pi(X\mid Y)} \\
		&=\sum_{(X, Y, \bM)}\pi(\bM\mid X, Y)\pi(X\mid Y, \bM)\pi(Y).
	\end{align*}
	Since $\pi$ is a \emph{deterministic} standard protocol, $M_0$ is fixed.
	All even messages $(M_2, M_4,\ldots)$ are sent by Bob such that each $M_i$ is determined by $M_{<i}$ and $Y$.
	Therefore, $(M_2, M_4,\ldots)$ are determined by all odd messages $(M_1, M_3,\ldots)$ and $Y$.
	Denote $(M_2, M_4,\ldots)$ by $M_{\even}$ and $(M_1, M_3,\ldots)$ by $M_{\odd}$, we have
	\begin{align*}
		&\quad\,\,\sum_{(X, Y, \bM)}\pi(\bM\mid X, Y)\pi(X\mid Y, \bM)\pi(Y) \\
		&=\sum_{(X, Y, \bM)}\pi(M_{\even}\mid X, Y, M_{\odd})\pi(M_{\odd}\mid X, Y)\pi(X\mid Y, M_{\odd}, M_{\even})\pi(Y) \\
		&=\sum_{(X, Y, \bM)}\pi(M_{\even}\mid X, Y, M_{\odd})\pi(M_{\odd}\mid X, Y)\pi(X\mid Y, M_{\odd})\pi(Y) \\
		&=\sum_{(X, Y, \bM)}\pi(M_{\odd}\mid X, Y)\pi(X, M_{\even}\mid Y, M_{\odd})\pi(Y) \\
		&\leq \sum_{(X, Y, \bM)}\pi(X, M_{\even}\mid Y, M_{\odd})\pi(Y) \\
		&=\sum_{(Y, M_{\odd})}\pi(Y) \\
		&=\sum_{M_{\odd}} 1 \\
		&\leq 2^{\lceil r/2\rceil T_A},
	\end{align*}
	where the last inequality uses the fact that Alice's messages have at most $T_A$ bits in each (odd) round.
	Similarly, we have $\chis{\mu,B}{\pi}\leq 2^{\lfloor r/2\rfloor T_B}$.

	Finally, by the connection between advantage and success probability, $\E_{\pi}\left[\adv_{\pi}(f^{\oplus n}(X, Y)\mid \bM)\right]\geq \sigma$.
	Hence, 
	\[
		\phi_n(\pi)\leq \lceil r/2\rceil \cdot \frac{T_A\cdot\log(1/\alpha)}{C_A-c\log(r/\alpha)}+\lfloor r/2\rfloor\cdot\frac{T_B\cdot\log(1/\alpha)}{C_B-c\log(r/\alpha)}+32\log (1/\sigma).
	\]
\end{proof}

In the rest of the paper, we will prove the following lemma, which shows that given a protocol for $f^{\oplus k}$, we can construct a protocol for $f^{\oplus k-1}$ with a lower potential.

\newcommand{\leminductioncont}{
	For $k\geq 2$, \\ \strut\qquad \textbf{if} there is a generalized protocol $\pi$ for $f^{\oplus k}$ with the rectangle property with respect to $\mu^k$ and an event $V\in\cS_{\rect}(\pi)$ such that $\pi(V)\geq 2^{-12}$, \\ \strut\qquad \textbf{then} there is a generalized protocol $\pi_{\newn}$ for $f^{\oplus k-1}$ with the rectangle property with respect to $\mu^{k-1}$ and an event $V_{\newn}\in\cS_{\rect}(\pi_{\newn})$ such that $\pi_{\newn}(V_{\newn})\geq 2^{-12}$, and
	\[
		\phi_{k-1}(\pi_{\newn}\mid V_{\newn})\leq \phi_k(\pi\mid V)-\frac{1}{16}\log(1/\alpha).
	\]	
}
\begin{lemma}\label{lem_induction}
	\leminductioncont
\end{lemma}

Our main theorem is a direct corollary of Lemma~\ref{lem_potential_low_bound},~\ref{lem_potential_standard} and~\ref{lem_induction}.
\begin{proof}[Proof of Theorem~\ref{thm_main}]
	Since $C_A, C_B\geq 2c\log(r/\alpha)$, $C_A/2\leq C_A-c\log(r/\alpha)$ and $C_B/2\leq C_B-c\log(r/\alpha)$.
	Suppose there exists an $r$-round protocol $\pi^{(n)}$ where Alice sends at most \[
		2^{-8}r^{-1}n\cdot C_A\leq 2^{-7}r^{-1}n(C_A-c\log (r/\alpha))
	\] bits in each round and Bob sends at most $$2^{-8}r^{-1}n\cdot C_B\leq 2^{-7}r^{-1}n(C_B-c\log (r/\alpha))$$ in each round, which computes $f^{\oplus n}$ correctly with probability $1/2+\sigma/2$ when the input is sampled from $\mu^n$.
	By fixing the randomness, we may assume that $\pi^{(n)}$ is deterministic.
	Then by Lemma~\ref{lem_potential_standard}, we have
	\[
		\phi_n(\pi^{(n)})\leq 2^{-7}n \log (1/\alpha)+32\log (1/\sigma).
	\]

	Now we set $V^{(n)}$ to be the whole sample space of $\pi^{(n)}$.
	Clearly, $\pi^{(n)}$ and $V^{(n)}$ satisfy the premise of Lemma~\ref{lem_induction}.
	By inductively applying Lemma~\ref{lem_induction} a total of $n-1$ times, we obtain a protocol $\pi^{(1)}$ for $f$ and event $V^{(1)}$ such that $\pi^{(1)}(V^{(1)})\geq 2^{-12}$ and
	\[
		\phi_1(\pi^{(1)}\mid V^{(1)})\leq \phi_n(\pi^{(n)})-\frac{n-1}{16}\cdot\log(1/\alpha).
	\]

	On the other hand, Lemma~\ref{lem_potential_low_bound} implies that the LHS is at least $-3\log(1/\pi^{(1)}(V^{(1)}))\geq -36$, implying that
	\[
		\phi_n(\pi^{(n)})\geq \frac{n-1}{16}\cdot\log (1/\alpha)-36\geq 2^{-6}n\log(1/\alpha),
	\] 
	since $n\geq2$ and $\alpha<r^{-cr}$ for a sufficiently large $c$.

	Combining the above upper and lower bounds on $\phi_n(\pi^{(n)})$, we obtain
	\[
		2^{-7}n \log (1/\alpha)+32\log (1/\sigma)\geq 2^{-6}n\log(1/\alpha),
	\]
	implying that $\log(1/\sigma)\geq 2^{-12}n\log (1/\alpha)$, i.e.,
	\[
		\sigma\leq \alpha^{2^{-12}n}.
	\]
	This proves the theorem.
\end{proof}


\section{Decomposition of Generalized Protocols}\label{sec_decompose}

To prove Lemma~\ref{lem_induction}, we will decompose a generalized protocol $\pi$ for $f^{\oplus k}$ into a protocol $\pi_{<k}$ for $f^{\oplus k-1}$ and a protocol $\pi_k$ for $f$ such that the costs of $\pi_{<k}$ and $\pi_k$ ``add up'' to the costs of $\pi$ pointwisely.
For simplicity of notations, we will assume that \emph{$r$ is even} from now on, the case of odd $r$ is similar.

\subsection{Definition of $\pi_{<k}$ and $\pi_k$}
Fix a generalized protocol $\pi$ with the rectangle property with respect to $\mu^k$.
Let $(X, Y, \bM)\sim\pi$.
We view the following tuple as the $r$-round generalized protocol $\pi_{<k}$ on inputs $(X_{<k}, Y_{<k})$
\[
	\left(X_{<k}, Y_{<k}, \left(M_0\circ X_k, M_1, M_2,\ldots, M_{r-1}, M_r\right)\right),
\]
where we append $X_k$ to $M_0$.
We view the following tuple as the $r$-round generalized protocol $\pi_k$ on inputs $(X_k, Y_k)$
\[
	\left(X_k, Y_k, \left(M_0, Y_{<k}\circ M_1, M_2,\ldots, M_{r-1}, M_r\right)\right),
\]
where we prepend $Y_{<k}$ to $M_1$.

It is useful to think that $\pi_{<k}$, $\pi_k$ and $\pi$ are the \emph{same} distribution over the \emph{same} sample space, only their inputs and transcripts are defined in different ways.
Therefore, we may use $\pi_{<k}(W),\pi_k(W),\pi(W)$ interchangeably when measuring the probability of an event $W$.

For simplicity of notations, we use $\bM^{(\pi_{<k})}$ to denote $\left(M_0\circ X_k, M_1, M_2,\ldots, M_r\right)$, the transcript of $\pi_{<k}$, and use $\bM^{(\pi_k)}$ to denote $\left(M_0, Y_{<k}\circ M_1, M_2,\ldots, M_r\right)$, the transcript of $\pi_k$.
$M_i^{(\pi_{<k})}$ and $M_i^{(\pi_k)}$ are defined similarly.
Since $(X, Y, \bM)$ determines $(X_{<k}, Y_{<k}, \bM^{(\pi_{<k})})$, we define $\theta$-cost of $\pi_{<k}$ at $(X, Y, \bM)$ as $$\thec{\pi_{<k}}{\mu^{k-1}}[X, Y, \bM]:=\thec{\pi_{<k}}{\mu^{k-1}}[X_{<k}, Y_{<k}, \bM^{(\pi_{<k})}],$$ where $(X_{<k}, Y_{<k}, \bM^{(\pi_{<k})})$ is the triple determined by $(X, Y, \bM)$.
Note that this cost does not depend on $Y_k$ given the other parts of $(X, Y, \bM)$.
The $\chisq$-costs of $\pi_{<k}$ and the costs of $\pi_k$ at $(X, Y, \bM)$ are defined similarly.

\bigskip




In the remainder of this section, we will analyze $\pi_{<k}$ and $\pi_k$.
First, we observe that the partial rectangle property of $\pi$ implies the rectangle properties of $\pi_{<k}$ and $\pi_k$.
\begin{proposition}\label{prop_pi<k_pk_rect}
	Let $W$ be an event such that $\pi\mid W$ has the partial rectangle property with respect to $\mu^k$.
	Then $\pi_{<k}\mid W$ has the rectangle property with respect to $\mu^{k-1}$, and $\pi_k\mid W$ has the rectangle property with respect to $\mu$.
\end{proposition}
\begin{proof}
	Since $\pi\mid W$ has the partial rectangle property, there exists $g_1,g_2,g_3$ such that
	\[
		\pi(X, Y, \bM\mid W)=\mu^k(X, Y)\cdot g_1(X, \bM)\cdot g_2(Y, \bM)\cdot g_3(X_k, Y_{<k}, \bM).
	\]
	Thus,
	\begin{align*}
		\pi_{<k}(X_{<k}, Y_{<k}, \bM^{(\pi_{<k})}\mid W)&=\pi(X, Y_{<k}, \bM\mid W) \\
		&=\sum_{Y_k}\mu^k(X, Y)\cdot g_1(X, \bM)\cdot g_2(Y, \bM)\cdot g_3(X_k, Y_{<k}, \bM) \\
		&=\mu^{k-1}(X_{<k}, Y_{<k})\cdot g_1(X, \bM)\cdot \left(\sum_{Y_k} \mu(X_k, Y_k)\cdot g_2(Y, \bM)\cdot g_3(X_k, Y_{<k}, \bM)\right).
	\end{align*}
	Note that the second factor is a function of only $X_{<k}$ and $\bM^{(\pi_{<k})}$, the third factor is a function of only $Y_{<k}$ and $\bM^{(\pi_{<k})}$.

	For $\pi_k$, we have
	\begin{align*}
		\pi_{k}(X_{k}, Y_{k}, \bM^{(\pi_k)}\mid W)&=\pi(X_k, Y, \bM\mid W) \\
		&=\sum_{X_{<k}} \mu^k(X, Y)\cdot g_1(X, \bM)\cdot g_2(Y, \bM)\cdot g_3(X_k, Y_{<k}, \bM) \\
		&=\mu(X_k, Y_k)\cdot \left(\sum_{X_{<k}}g_1(X, \bM)\cdot g_3(X_k, Y_{<k}, \bM)\right)\cdot g_2(Y, \bM).
	\end{align*}
	The second factor depends only on $X_k$ and $\bM^{(\pi_k)}$, and the third factor depends only on $Y_k$ and $\bM^{(\pi_k)}$.
	This proves the lemma.
\end{proof}

Similarly, we have the following relation between $\cS_{\rect}(\pi_{<k}),\cS_{\rect}(\pi_k)$ and $\cS_{\pa}(\pi)$.
\begin{proposition}\label{prop_pi<k_pik_S_rect}
	We have $\cS_{\rect}(\pi_{<k})\subseteq \cS_{\pa}(\pi)$ and $\cS_{\rect}(\pi_{k})\subseteq \cS_{\pa}(\pi)$.
\end{proposition}
\begin{proof}
	Let $S\in \cS_{\rect}(\pi_{<k})$ be an event such that
	\[
		S=\{(X, Y, \bM):(X_{<k}, \bM^{(\pi_{<k})})\in S_{X,M}\wedge (Y_{<k}, \bM^{(\pi_{<k})})\in S_{Y,M}\},
	\]
	for sets $S_{X, M}\in \cU_{X,M}(\pi_{<k})$ and $S_{Y,M}\in \cU_{Y,M}(\pi_{<k})$.
	Hence,
	\[
		S=\{(X, Y, \bM): (X, \bM)\in S_{X, M}\wedge (X_k, Y_{<k}, \bM)\in S_{Y, M}\}
	\]
	is a set in $\cS_{\pa}(\pi)$.

	The proof of $\cS_{\rect}(\pi_{k})\subseteq \cS_{\pa}(\pi)$ is similar, and we omit the details.
\end{proof}

\subsection{Decomposition of the costs}
Below is the first main lemma of the decomposition, stating that the product of $\theta$-costs of $\pi_{<k}$ and $\pi_{k}$ is equal to that of $\pi$ pointwisely.
\begin{lemma}\label{lem_sum_thetacost}
	The product of the $\theta$-costs of $\pi_{<k}$ and $\pi_{k}$ at $(X, Y, \bM)$ is $\theta$-cost of $\pi$ at $(X, Y, \bM)$,
	\[
		\dev{\pi_{<k}}{\mu^{k-1}}[X, Y, \bM]\cdot \dev{\pi_{k}}{\mu}[X, Y, \bM]=\dev{\pi}{\mu^k}[X, Y, \bM].
	\]
\end{lemma}

\begin{proof}
    By definition, we have
	\begin{align}
		&\kern1.25em\thec{\pi_{<k}}{\mu^{k-1}}[X, Y, \bM] \nonumber\\
		&=\frac{\pi_{<k}(X_{<k}, Y_{<k},\bM^{(\pi_{<k})}\mid M_0^{(\pi_{<k})})}{\mu^{k-1}(X_{<k}, Y_{<k})}\cdot \prod_{\odd i\in[r]}\frac{1}{\pi_{<k}(M_i^{(\pi_{<k})}\mid X_{<k}, M_{<i}^{(\pi_{<k})})}\cdot\prod_{\even i\in[r]}\frac{1}{\pi_{<k}(M_i^{(\pi_{<k})}\mid Y_{<k}, M_{<i}^{(\pi_{<k})})}\nonumber \\
		&=\frac{\pi(X, Y_{<k}, \bM\mid X_k, M_0)}{\mu^{k-1}(X_{<k}, Y_{<k})}\cdot \prod_{\odd i\in[r]}\frac{1}{\pi(M_i\mid X, M_{<i})}\cdot\prod_{\even i\in[r]}\frac{1}{\pi(M_i\mid X_k, Y_{<k}, M_{<i})} \nonumber\\
		&=\frac{\pi(X_{<k}, Y_{<k}, \bM\mid X_k, M_0)}{\mu^{k-1}(X_{<k}, Y_{<k})}\cdot \prod_{\odd i\in[r]}\frac{1}{\pi(M_i\mid X, M_{<i})}\cdot\prod_{\even i\in[r]}\frac{1}{\pi(M_i\mid X_k, Y_{<k}, M_{<i})}.\label{eqn_theta_pi<k}
	\end{align}

    Similarly,
	\begin{align}
		&\kern1.25em\thec{\pi_k}{\mu}[X, Y, \bM] \nonumber\\
		&=\frac{\pi_k(X_k, Y_k, \bM^{(\pi_k)}\mid M_0^{(\pi_k)})}{\mu(X_k, Y_k)}\cdot \prod_{\odd i\in[r]}\frac{1}{\pi_k(M_i^{(\pi_k)}\mid X_k, M_{<i}^{(\pi_k)})}\cdot\prod_{\even i\in[r]}\frac{1}{\pi_k(M_i^{(\pi_k)}\mid Y_k, M_{<i}^{(\pi_k)})} \nonumber\\
		&=\frac{\pi(X_k, Y, \bM\mid M_0)}{\mu(X_k, Y_k)}\cdot\frac{1}{\pi(Y_{<k}, M_1\mid X_k, M_0)}\cdot \prod_{\odd i\in[3,r]}\frac{1}{\pi(M_i\mid X_k, Y_{<k}, M_{<i})}\cdot\prod_{\even i\in[r]}\frac{1}{\pi(M_i\mid Y, M_{<i})} \nonumber \\
		&=\frac{\pi(X_k, Y, \bM\mid M_0)}{\mu(X_k, Y_k)\cdot\pi(Y_{<k}\mid X_k, M_0)}\cdot \prod_{\odd i\in[r]}\frac{1}{\pi(M_i\mid X_k, Y_{<k}, M_{<i})} \cdot\prod_{\even i\in[r]}\frac{1}{\pi(M_i\mid Y, M_{<i})}. \label{eqn_theta_pik}
	\end{align}

	Combining Equation~\eqref{eqn_theta_pi<k} and~\eqref{eqn_theta_pik}, we have
	\begin{align*}
		&\kern1.25em\thec{\pi_{<k}}{\mu^{k-1}}[X, Y, \bM]\cdot \thec{\pi_k}{\mu}[X, Y, \bM] \\
		&=\frac{\pi(X_{<k}, Y_{<k}, \bM\mid X_k, M_0)}{\mu^{k-1}(X_{<k}, Y_{<k})}\cdot \prod_{\odd i\in[r]}\frac{1}{\pi(M_i\mid X, M_{<i})}\cdot\prod_{\even i\in[r]}\frac{1}{\pi(M_i\mid X_k, Y_{<k}, M_{<i})} \\
		&\strut\qquad\qquad\cdot \frac{\pi(X_k, Y, \bM\mid M_0)}{\mu(X_k, Y_k)\cdot\pi(Y_{<k}\mid X_k, M_0)}\cdot \prod_{\odd i\in[r]}\frac{1}{\pi(M_i\mid X_k, Y_{<k}, M_{<i})}\cdot\prod_{\even i\in[r]}\frac{1}{\pi(M_i\mid Y, M_{<i})} \\
		&=\frac{\pi(X_{<k}, Y_{<k}, \bM\mid X_k, M_0)\pi(X_k, Y, \bM\mid M_0)}{\mu^{k}(X, Y)\pi(Y_{<k}\mid X_k, M_0)}\cdot \frac{1}{\pi(\bM\mid X_k, Y_{<k}, M_0)} \\
		&\strut\qquad\qquad\cdot \prod_{\odd i\in[r]}\frac{1}{\pi(M_i\mid X, M_{<i})}\cdot\prod_{\even i\in[r]}\frac{1}{\pi(M_i\mid Y, M_i)} \\
		&=\frac{\pi(X_{<k}\mid X_k, Y_{<k}, \bM)\pi(X_k, Y, \bM\mid M_0)}{\mu^{k}(X, Y)}\cdot \prod_{\odd i\in[r]}\frac{1}{\pi(M_i\mid X, M_{<i})}\cdot\prod_{\even i\in[r]}\frac{1}{\pi(M_i\mid Y, M_{<i})}.
	\end{align*}
	Then by the rectangle property of $\pi$ and Proposition~\ref{prop_independent_x<k_yk}, $Y_k$ is independent of $X_{<k}$ conditioned on $(X_k, Y_{<k}, \bM)$.
	It is equal to
	\begin{align*}
		&\kern-2em\frac{\pi(X_{<k}\mid X_k, Y, \bM)\pi(X_k, Y, \bM\mid M_0)}{\mu^{k}(X, Y)}\cdot \prod_{\odd i\in[r]}\frac{1}{\pi(M_i\mid X, M_{<i})}\cdot\prod_{\even i\in[r]}\frac{1}{\pi(M_i\mid Y, M_{<i})} \\
		&=\frac{\pi(X, Y, \bM\mid M_0)}{\mu^k(X, Y)}\cdot \prod_{\odd i\in[r]}\frac{1}{\pi(M_i\mid X, M_{<i})}\cdot\prod_{\even i\in[r]}\frac{1}{\pi(M_i\mid Y, M_{<i})}\\
		&=\thec{\pi}{\mu^k}[X, Y, \bM].
	\end{align*}
	This proves the lemma.
\end{proof}

The second main lemma of the decomposition states that the product of the $\chi^2$-costs of $\pi_{<k}$ and $\pi_k$ is also equal to that of $\pi$ pointwisely. 
\begin{lemma}\label{lem_sum_chicost}
	The product of the $\chi^2$-costs of $\pi_{<k}$ and $\pi_k$ at $(X, Y, \bM)$ is the $\chi^2$-cost of $\pi$ at $(X, Y, \bM)$ by Alice and Bob respectively,
	\begin{align*}
		\chis{\pi_{<k}}{\mu^{k-1},A}[X, Y, \bM]\cdot \chis{\pi_{k}}{\mu,A}[X, Y, \bM]&=\chis{\pi}{\mu^{k},A}[X, Y, \bM], \\
		\chis{\pi_{<k}}{\mu^{k-1},B}[X, Y, \bM]\cdot \chis{\pi_{k}}{\mu,B}[X, Y, \bM]&=\chis{\pi}{\mu^{k},B}[X, Y, \bM].
	\end{align*}
\end{lemma}
\begin{proof}
	For the $\chi^2$-cost by Alice, by definition, we have
	\begin{align*}
		\chis{\pi_{<k}}{\mu^{k-1},A}[X, Y, \bM]&=\frac{\pi_{<k}(X_{<k}\mid Y_{<k}, \bM^{(\pi_{<k})})}{\mu^{k-1}(X_{<k}\mid Y_{<k})} \\
		&=\frac{\pi(X_{<k}\mid X_k, Y_{<k}, \bM)}{\mu^{k-1}(X_{<k}\mid Y_{<k})},
	\end{align*}
	and
	\begin{align*}
		\chis{\pi_{k}}{\mu,A}[X, Y, \bM]&=\frac{\pi_{k}(X_{k}\mid Y_{k}, \bM^{(\pi_{k})})}{\mu(X_{k}\mid Y_{k})} \\
		&=\frac{\pi(X_{k}\mid Y, \bM)}{\mu(X_{k}\mid Y_{k})}.
	\end{align*}
	Hence, by partial rectangle property of $\pi$ and Proposition~\ref{prop_independent_x<k_yk}, their product is equal to
	\begin{align*}
		&\kern1.25em\chis{\pi_{<k}}{\mu^{k-1},A}[X, Y, \bM]\cdot \chis{\pi_{k}}{\mu,A}[X, Y, \bM] \\
		&= \frac{\pi(X_{<k}\mid X_k, Y_{<k}, \bM)}{\mu^{k-1}(X_{<k}\mid Y_{<k})}\cdot \frac{\pi(X_{k}\mid Y, \bM)}{\mu(X_{k}\mid Y_{k})} \\
		&= \frac{\pi(X_{<k}\mid X_k, Y, \bM)\cdot \pi(X_{k}\mid Y, \bM)}{\mu^k(X\mid Y)} \\
		&= \frac{\pi(X\mid Y, \bM)}{\mu^k(X\mid Y)} \\
		&= \chis{\pi}{\mu^k,A}[X, Y, \bM].
	\end{align*}

	The $\chi^2$-cost for Bob is similar,
	\begin{align*}
		&\kern1.25em\chis{\pi_{<k}}{\mu^{k-1},B}[X, Y, \bM]\cdot \chis{\pi_{k}}{\mu,B}[X, Y, \bM] \\
		&= \frac{\pi(Y_{<k}\mid X, \bM)}{\mu^{k-1}(Y_{<k}\mid X_{<k})}\cdot \frac{\pi(Y_k\mid X_k, Y_{<k}, \bM)}{\mu(Y_k\mid X_k)} \\
		&= \frac{\pi(Y_{<k}\mid X, \bM)\cdot \pi(Y_k\mid X, Y_{<k}, \bM)}{\mu^k(Y \mid X)} \\
		&= \chis{\pi}{\mu^k,B}[X, Y, \bM].
	\end{align*}
	This proves the lemma.
\end{proof}



\section{Induction: Proof of Lemma~\ref{lem_induction}}

In this section, we will use the decomposition of $\pi$ to prove Lemma~\ref{lem_induction}.

\subsection{Identify event $U$}\label{sec_identify_U}
As we mentioned in Section~\ref{sec_overview}, to obtain a new protocol for $f^{\oplus k-1}$ from $\pi$, we first identify an event $U$ such that the advantage of $\pi$ is not concentrated on any $S$ for $S\in \cS_{\pa}(\pi)$ and $S\subseteq U$.

Let $U\in \cS_\pa(\pi)$ and $U\subseteq V$ be an event that maximizes
\begin{equation}\label{eqn_max_U}
	\pi(U)^{1/2}\cdot \E_{\pi\mid U}\left[\adv(f^{\oplus k}(X, Y)\mid X_k, Y_{<k}, \bM, U)\right].
\end{equation}
Since $\cS_{\pa}(\pi)$ is a discrete set, such $U$ exists.
If there is a tie, we fix $U$ to be any maximizer.
We first show that conditioning on $U$ reduces the potential function value, and the reduction is large when the probability $U$ is small.
\begin{lemma}\label{lem_max_U}
	$\pi\mid U$ has the \emph{partial rectangle property} with respect to $\mu^k$, and 
	\[
		\phi_{k, \pa}(\pi\mid U)\leq \phi_k(\pi\mid V)-13\log (1/\pi(U\mid V)).
	\]
\end{lemma}
\begin{proof}
	By definition, we have
	\begin{align*}
		\phi_{k, \pa}(\pi\mid U)&=\log \thec{\pi\mid U}{\mu^k}+\frac{\log (1/\alpha)}{C_A-c\log (r/\alpha)}\cdot \log \chis{\pi\mid U}{\mu^k,A}+\frac{\log (1/\alpha)}{C_B-c\log (r/\alpha)}\cdot \log \chis{\pi\mid U}{\mu^k,B} \\
		&\qquad+32\log\left(\E_{\pi\mid U}\left[\adv_{\pi}(f^{\oplus k}(X, Y)\mid X_k, Y_{<k}, \bM, U)\right]^{-1}\right).
	\end{align*}
	By Proposition~\ref{prop_thec_condition_increase}, Proposition~\ref{prop_chis_condition_increase} and the fact that $U\subseteq V$, we have
	\begin{align*}
		\log \thec{\pi\mid U}{\mu^k}&\leq \log \thec{\pi\mid V}{\mu^k}+\log (1/\pi(U\mid V)) \\
		\log \chis{\pi\mid U}{\mu^k,A}&\leq \log \chis{\pi\mid V}{\mu^k,A}+\log (1/\pi(U\mid V)) \\
		\log \chis{\pi\mid U}{\mu^k,B}&\leq \log \chis{\pi\mid V}{\mu^k,B}+\log (1/\pi(U\mid V)).		
	\end{align*}
	Then since $U$ is the maximizer of Equation~\eqref{eqn_max_U} and $V\in\cS_{\rect}(\pi)\subseteq\cS_{\pa}(\pi)$,
	\begin{align*}
		&\kern-2em\E_{\pi\mid U}\left[\adv_{\pi}(f^{\oplus k}(X, Y)\mid X_k, Y_{<k}, \bM, U)\right] \\
		&\geq \pi(U)^{-1/2}\cdot \pi(V)^{1/2}\cdot \E_{\pi\mid V}\left[\adv_{\pi}(f^{\oplus k}(X, Y)\mid X_k, Y_{<k}, \bM, V)\right] \\
		&=\pi(U\mid V)^{-1/2}\cdot \E_{\pi\mid V}\left[\adv_{\pi}(f^{\oplus k}(X, Y)\mid X_k, Y_{<k}, \bM, V)\right].
	\end{align*}
	Since knowing less could only decrease the advantage (Proposition~\ref{prop_adv_know_more}),
	\[
		\E_{\pi\mid V}\left[\adv_{\pi}(f^{\oplus k}(X, Y)\mid X_k, Y_{<k}, \bM, V)\right]\geq \E_{\pi\mid V}\left[\adv_{\pi}(f^{\oplus k}(X, Y)\mid \bM, V)\right].
	\]
	
	Combining the inequalities and using the fact that $\log (1/\alpha)<C_A-c\log (r/a)$ and $\log (1/\alpha)<C_B-c\log (r/a)$, we have
	\begin{align*}
		\phi_{k, \pa}(\pi\mid U)&\leq \log \thec{\pi\mid V}{\mu^k}+\frac{\log (1/\alpha)}{C_A-c\log (r/\alpha)}\cdot \log \chis{\pi\mid V}{\mu^k,A}+\frac{\log (1/\alpha)}{C_B-c\log (r/\alpha)}\cdot \log \chis{\pi\mid V}{\mu^k,B} \\
		&\qquad+32\log\left(\E_{\pi\mid V}\left[f^{\oplus k}(X, Y)\mid \bM, V\right]^{-1}\right)+3\log(1/\pi(U\mid V))-16\log(1/\pi(U\mid V)) \\
		&=\phi_{k}(\pi\mid V)-13\log (1/\pi(U\mid V)).
	\end{align*}
	This proves the lemma.
\end{proof}
We need the following proposition in the later proof.
\begin{proposition}\label{prop_max_U}
	We have the following:
	\begin{enumerate}[(i)]
		\item for any $S\in\cS_{\pa}(\pi)$ and $S\subseteq U$, we have
		\begin{align*}
			&\kern-2em\pi(S)\cdot \E_{\pi\mid S}\left[\adv_{\pi}(f^{\oplus k}(X, Y)\mid X_k, Y_{<k}, \bM, S)\right] \\
			&\qquad\leq \pi(S\mid U)^{1/2}\cdot \pi(U)\cdot \E_{\pi\mid U}\left[\adv_{\pi}(f^{\oplus k}(X, Y)\mid X_k, Y_{<k}, \bM, U)\right].
		\end{align*}
		\item for any $S\in\cS_{\pa}(\pi)$ and $S\subseteq U$,
		if 
		\begin{align*}
			&\kern-2em\pi(U)\cdot \E_{\pi\mid U}\left[\adv_{\pi}(f^{\oplus k}(X, Y)\mid X_k, Y_{<k}, \bM, U)\right] \\
			&\leq s\cdot \pi(S)\cdot \E_{\pi\mid S}\left[\adv_{\pi}(f^{\oplus k}(X, Y)\mid X_k, Y_{<k}, \bM, S)\right],
		\end{align*}
		for some $s\geq 1$, then for any $t\leq 32$, we have
		\[
			\phi_{k,\pa}^{\adv}(\pi\mid U)+t\log (1/\pi(U))\geq \phi_{k,\pa}^{\adv}(\pi\mid S)+t\log (1/\pi(S))-32\log s.
		\]
	\end{enumerate}
\end{proposition}
\begin{proof}
	\begin{enumerate}[(i)]
		\item Since $U$ is the maximizer of Equation~\eqref{eqn_max_U}, $S\in \cS_{\pa}(\pi)$ and $S\subseteq U$, we have
		\begin{align*}
			&\kern-2em\pi(S)\cdot \E_{\pi\mid S}\left[\adv_{\pi}(f^{\oplus k}(X, Y)\mid X_k, Y_{<k}, \bM, S)\right] \\
			&\leq \pi(S)^{1/2}\cdot \pi(U)^{1/2}\cdot \E_{\pi\mid U}\left[\adv_{\pi}(f^{\oplus k}(X, Y)\mid X_k, Y_{<k}, \bM, U)\right] \\
			&=\pi(S\mid U)^{1/2}\cdot \pi(U)\cdot \E_{\pi\mid U}\left[\adv_{\pi}(f^{\oplus k}(X, Y)\mid X_k, Y_{<k}, \bM, U)\right].
		\end{align*}
		\item By taking the logarithm on both sides of the premise, we have
		\begin{align*}
			&\kern-2em\log \left(\E_{\pi\mid U}\left[\adv_{\pi}(f^{\oplus k}(X, Y)\mid X_k, Y_{<k}, \bM, U)\right]^{-1}\right)+\log (1/\pi(U)) \\
			&\geq \log \left(\E_{\pi\mid S}\left[\adv_{\pi}(f^{\oplus k}(X, Y)\mid X_k, Y_{<k}, \bM, S)\right]^{-1}\right)+\log (1/\pi(S))-\log s,
		\end{align*}
		i.e., (recall Definition~\ref{def_potential})
		\[
			\phi_{k,\pa}^{\adv}(\pi\mid U)+32\log (1/\pi(U))\geq \phi_{k,\pa}^{\adv}(\pi\mid S)+32\log (1/\pi(S))-32\log s.
		\]
		Since $\pi(U)\geq \pi(S)$, for any $t\leq 32$, 
		\[
			\phi_{k,\pa}^{\adv}(\pi\mid U)+t\log (1/\pi(U))\geq \phi_{k,\pa}^{\adv}(\pi\mid S)+t\log (1/\pi(S))-32\log s.
		\]
	\end{enumerate}
\end{proof}

Now we will divide the set of all $(X_k, Y_{<k}, \bM)$ with nonzero probability under $\pi$ into subsets based on the costs and the advantages of $\pi_{<k}$ and $\pi_k$.
Then we show that for each subset, there is a way to construct a generalized protocol for $f^{\oplus k-1}$ such that at least one of the protocols satisfies the requirements of Lemma~\ref{lem_induction}.
To analyze the costs of these protocols, which we will construct later in this section, we need the following two lemmas.

\begin{lemma}\label{lem_cost_drop_thec}
	Fix a set $S$ of triples $(X_k, Y_{<k}, \bM)$ and a parameter $\eta>0$.
	If for all $(X_k, Y_{<k}, \bM)\in S$,
	\[
		\E_{Y_k\sim\pi\mid X_k, Y_{<k}, \bM, U}\left[\thec{\pi_k}{\mu}[X_k, Y_k, \bM^{(\pi_k)}]\right]\geq \eta,
	\]
	then we have
	\[
		\log \thec{\pi_{<k}\mid S\cap U}{\mu^{k-1}}\leq \log \thec{\pi\mid U}{\mu^{k}}+\log (1/\pi(S\mid U))-\log \eta,
	\]
	where we abused the notation to let $S$ also denote the set $\{(X, Y, \bM): (X_k, Y_{<k}, \bM)\in S\}$.
\end{lemma}
\begin{proof}
	By Lemma~\ref{lem_sum_thetacost}, we have
	\begin{align*}
		&\quad\,\,\thec{\pi\mid S\cap U}{\mu^{k}} \\
		&=\E_{\pi}\left[\thec{\pi}{\mu^{k}}[X, Y, \bM]\mid S\cap U\right] \\
		&=\E_{\pi}\left[\thec{\pi_{<k}}{\mu^{k-1}}[X, Y, \bM]\cdot \thec{\pi_k}{\mu}[X, Y, \bM]\mid S\cap U\right].
	\end{align*}
	By the construction of $\pi_{<k}$ and $\pi_k$, $\thec{\pi_{<k}}{\mu^{k-1}}[X, Y, \bM]$ is a function of $(X, Y_{<k}, \bM)$ and does not depend on $Y_k$, and $\thec{\pi_k}{\mu}[X_k, Y, \bM]$ is a function of $(X_k, Y, \bM)$ does not depend on $X_{<k}$.
	Thus, it is equal to
	\[
		\E_{\pi}\left[\thec{\pi_{<k}}{\mu^{k-1}}[X, Y_{<k}, \bM]\cdot \thec{\pi_k}{\mu}[X_k, Y, \bM]\mid S\cap U\right].
	\]
	Since $\cS_{\pa}$ is closed under intersection and $S\in\cS_{\pa}$ by definition, we have that $S\cap U\in \cS_{\pa}$.
	Hence, $\pi\mid S\cap U$ has the partial rectangle property by Proposition~\ref{prop_cs_partial}(ii).
	Then $X_{<k}$ and $Y_k$ are independent conditioned on $(X_k, Y_{<k}, \bM, S\cap U)$ by Proposition~\ref{prop_independent_x<k_yk}.
	Hence, it is equal to
	\begin{align*}
		&\kern1.25em\E_{(X_k, Y_{<k}, \bM)\sim \pi\mid S\cap U}\left[\E_{X_{<k}\sim \pi\mid X_k, Y_{<k}, \bM, S\cap U}[\thec{\pi_{<k}}{\mu^{k-1}}[X, Y_{<k}, \bM]]\cdot \E_{Y_k\sim \pi\mid X_k, Y_{<k}, \bM, S\cap U}[\thec{\pi_k}{\mu}[X_k, Y, \bM]]\right] \\
		&=\E_{(X_k, Y_{<k}, \bM)\sim \pi\mid S\cap U}\left[\E_{X_{<k}\sim \pi\mid X_k, Y_{<k}, \bM, S\cap U}[\thec{\pi_{<k}}{\mu^{k-1}}[X_{<k}, Y_{<k}, \bM^{(\pi_{<k})}]]\right. \\
		&\strut\hspace{150pt} \times \left.\E_{Y_k\sim \pi\mid X_k, Y_{<k}, \bM, S\cap U}[\thec{\pi_k}{\mu}[X_k, Y_k, \bM^{(\pi_k)}]]\right].
	\end{align*}
	Since $S$ is a set of triples $(X_k, Y_{<k}, \bM)$, $(\pi\mid X_k, Y_{<k}, \bM, S\cap U)$ is the same as $(\pi\mid X_k, Y_{<k}, \bM, U)$ (for $(X_k, Y_{<k}, \bM)\in S$).
	It is equal to
	\begin{align*}
		&\kern1.25em\E_{(X_k, Y_{<k}, \bM)\sim \pi\mid S\cap U}\left[\E_{X_{<k}\sim \pi\mid X_k, Y_{<k}, \bM, S\cap U}[\thec{\pi_{<k}}{\mu^{k-1}}[X_{<k}, Y_{<k}, \bM^{(\pi_{<k})}]] \right. \\
		&\strut\hspace{150pt}\times \left. \E_{Y_k\sim \pi\mid X_k, Y_{<k}, \bM, U}[\thec{\pi_k}{\mu}[X_k, Y_k, \bM^{(\pi_k)}]]\right] \\
		&\geq \E_{(X_k, Y_{<k}, \bM)\sim \pi\mid S\cap U}\left[\E_{X_{<k}\sim\pi\mid X_k, Y_{<k}, \bM, S\cap U}[\thec{\pi_{<k}}{\mu^{k-1}}[X_{<k}, Y_{<k}, \bM^{(\pi_{<k})}]]\cdot \eta\right] \\
		&=\thec{\pi_{<k}\mid S\cap U}{\mu^{k-1}}\cdot \eta.
	\end{align*}
	Finally, by Proposition~\ref{prop_thec_condition_increase},
	\[
		\thec{\pi\mid S\cap U}{\mu^{k}}\leq \frac{\thec{\pi\mid U}{\mu^k}}{\pi(S\mid U)}.
	\]
	Hence, we have
	\[
		\log \thec{\pi_{<k}\mid S\cap U}{\mu^{k-1}}\leq \log \thec{\pi\mid U}{\mu^{k}}+\log (1/\pi(S\mid U))-\log \eta.
	\]
	This proves the lemma.
\end{proof}

Next, by applying Lemma~\ref{lem_sum_chicost} and Proposition~\ref{prop_chis_condition_increase} instead of Lemma~\ref{lem_sum_thetacost} and Proposition~\ref{prop_thec_condition_increase}, the same proof gives the following lemma for the $\chisq$-costs.
\begin{lemma}\label{lem_cost_drop_chis}
	Fix a set $S$ of triples $(X_k, Y_{<k}, \bM)$ and a parameter $\eta>0$.
	If for all $(X_k, Y_{<k}, \bM)\in S$,
	\[
		\E_{Y_k\sim \pi\mid X_k, Y_{<k}, \bM, U}\left[\chis{\pi_k}{\mu, A}[X_k, Y_k, \bM^{(\pi_k)}]\right]\geq \eta,
	\]
	then we have
	\[
		\log \chis{\pi_{<k}\mid S\cap U}{\mu^{k-1}, A}\leq \log\chis{\pi\mid U}{\mu^{k}, A}+\log (1/\pi(S\mid U))-\log \eta;
	\]
	similarly, if for all $(X_k, Y_{<k}, \bM)\in S$,
	\[
		\E_{Y_k\sim\pi\mid X_k, Y_{<k}, \bM, U}\left[\chis{\pi_k}{\mu, B}[X_k, Y_k, \bM^{(\pi_k)}]\right]\geq \eta,
	\]
	then we have
	\[
		\log \chis{\pi_{<k}\mid S\cap U}{\mu^{k-1}, B}\leq \log\chis{\pi\mid U}{\mu^{k}, B}+\log (1/\pi(S\mid U))-\log \eta.
	\]
\end{lemma}

We will also need the following lemma to relate the advantage for $f^{\oplus k-1}$ to the advantage for $f^{\oplus k}$.

\begin{lemma}\label{lem_cost_decs_adv}
	Fix a set $S\in \cS_{\pa}(\pi)$ such that $\pi(S\cap U)>0$.
	Suppose there exists $b\in\{0,1\}$ such that for any $(X_k, Y_{<k}, \bM)$ with $\pi(X_k, Y_{<k}, \bM, S\cap U)>0$, we have
	\[
		\pi(f^{\oplus k-1}(X_{<k}, Y_{<k})=b\mid X_k, Y_{<k}, \bM, S\cap U)\geq 1/2.
	\]
	Then we have 
	\[
		\E_{\pi\mid S\cap U}\left[\adv_{\pi}(f^{\oplus k-1}(X_{<k}, Y_{<k})\mid \bM^{(\pi_{<k})}, S\cap U)\right]\geq \E_{\pi\mid S\cap U}\left[\adv_{\pi}(f^{\oplus k}(X, Y)\mid X_k, Y_{<k}, \bM, S\cap U)\right].
	\]
	Moreover, if we further have 
	\[
		\E_{\pi\mid S\cap U}\left[\adv_{\pi}(f(X_k, Y_k)\mid X_k, Y_{<k}, \bM, S\cap U)\right]\leq \eta,
	\]
	for
	\[
		\eta^{1/4}\leq \frac{1}{2}\cdot \frac{\pi(S\cap U)^{1/2}\cdot \E_{\pi\mid S\cap U}\left[\adv_{\pi}(f^{\oplus k}(X, Y)\mid X_k, Y_{<k}, \bM, S\cap U)\right]}{\pi(U)^{1/2}\cdot \E_{\pi\mid U}\left[\adv_{\pi}(f^{\oplus k}(X, Y)\mid X_k, Y_{<k}, \bM, U)\right]},
	\]
	then we have
	\begin{align*}
		&\kern-2em\E_{\pi\mid S\cap U}\left[\adv_{\pi}(f^{\oplus k-1}(X_{<k}, Y_{<k})\mid \bM^{(\pi_{<k})}, S\cap U)\right] \\
		&\geq \frac{1}{2}\eta^{-1/2}\cdot \E_{\pi\mid S\cap U}\left[\adv_{\pi}(f^{\oplus k}(X, Y)\mid X_k, Y_{<k}, \bM, S\cap U)\right].
	\end{align*}
\end{lemma}
The first condition $\pi(f^{\oplus k-1}(X_{<k}, Y_{<k})=b\mid X_k, Y_{<k}, \bM, S\cap U)\geq 1/2$ is used to ensure that the expected advantage conditioned on $(\bM^{(\pi_{<k})},S\cap U)$ is the same as the expected advantage conditioned on $(Y_{<k}, \bM^{(\pi_{<k})}, S\cap U)$.

\begin{proof}
	We have
	\begin{align*}
		&\kern1.25em\adv_{\pi}(f^{\oplus k-1}(X_{<k}, Y_{<k})\mid \bM^{(\pi_{<k})}, S\cap U) \\
		&=\adv_{\pi}(f^{\oplus k-1}(X_{<k}, Y_{<k})\mid X_k, \bM, S\cap U) \\
		&=\left|2\pi(f^{\oplus k-1}(X_{<k}, Y_{<k})=b\mid X_k, \bM, S\cap U)-1\right| \\
		&=\left|2\sum_{Y_{<k}}\pi(Y_{<k}\mid X_k, \bM, S\cap U)\cdot\pi(f^{\oplus k-1}(X_{<k}, Y_{<k})=b\mid X_k, Y_{<k}, \bM, S\cap U)-1\right| \\
		&=\left|\sum_{Y_{<k}}\pi(Y_{<k}\mid X_k, \bM, S\cap U)\cdot(2\pi(f^{\oplus k-1}(X_{<k}, Y_{<k})=b\mid X_k, Y_{<k}, \bM, S\cap U)-1)\right|.
	\end{align*}
	By the assumption that $\pi(f^{\oplus k-1}(X_{<k}, Y_{<k})=b\mid X_k, Y_{<k}, \bM, S\cap U)\geq 1/2$, the absolute value of the sum is equal to the sum of absolute values:
	\begin{align*}
		&\kern-2em\left|\sum_{Y_{<k}}\pi(Y_{<k}\mid X_k, \bM, S\cap U)\cdot(2\pi(f^{\oplus k-1}(X_{<k}, Y_{<k})=b\mid X_k, Y_{<k}, \bM, S\cap U)-1)\right| \\
		&=\sum_{Y_{<k}}\pi(Y_{<k}\mid X_k, \bM, S\cap U)\cdot\left|2\pi(f^{\oplus k-1}(X_{<k}, Y_{<k})=b\mid X_k, Y_{<k}, \bM, S\cap U)-1\right| \\
		&=\E_{Y_{<k}\sim\pi\mid X_k, \bM, S\cap U}\left[\adv_{\pi}(f^{\oplus k-1}(X_{<k}, Y_{<k})\mid X_k, Y_{<k}, \bM, S\cap U)\right].
	\end{align*}
	By taking the expectation over $(X_k, \bM)$ conditioned on $S\cap U$, we obtain
	\begin{align}
		&\kern-2em\E_{\pi\mid {S\cap U}}\left[\adv_{\pi}(f^{\oplus k-1}(X_{<k}, Y_{<k})\mid \bM^{(\pi_{<k})}, S\cap U)\right] \nonumber\\
		&=\E_{\pi\mid S\cap U}\left[\adv_{\pi}(f^{\oplus k-1}(X_{<k}, Y_{<k})\mid X_k, Y_{<k}, \bM, S\cap U)\right].\label{eqn_adv_xm_xym}
	\end{align}
	Since $S, U\in\cS_{\pa}(\pi)$, we have $S\cap U\in \cS_{\pa}(\pi)$.
	Therefore, $X_{<k}$ and $Y_k$ are independent conditioned on $(X_k, Y_{<k}, \bM, S\cap U)$ by Proposition~\ref{prop_cs_partial}(ii) and Proposition~\ref{prop_independent_x<k_yk}.
	In particular, $f^{\oplus k-1}(X_{<k}, Y_{<k})$ and $f(X_k, Y_k)$ are independent conditioned on $(X_k, Y_{<k}, \bM, S\cap U)$.
	By the fact that $f(X, Y)=f^{\oplus k-1}(X_{<k}, Y_{<k})\oplus f(X_k, Y_k)$ and Proposition~\ref{prop_xor}, we have
	\begin{align}
		&\kern-2em\adv_{\pi}(f^{\oplus k}(X, Y)\mid X_k, Y_{<k}, \bM, S\cap U) \nonumber\\
		&=\adv_{\pi}(f^{\oplus k-1}(X_{<k}, Y_{<k})\mid X_k, Y_{<k}, \bM, S\cap U)\cdot \adv_{\pi}(f(X, Y)\mid X_k, Y_{<k}, \bM, S\cap U) \label{eqn_adv_product}\\
		&\leq \adv_{\pi}(f^{\oplus k-1}(X_{<k}, Y_{<k})\mid X_k, Y_{<k}, \bM, S\cap U). \nonumber
	\end{align}

	Thus, the expected advantage is at least
	\begin{align*}
		&\kern-2em\E_{\pi\mid {S\cap U}}\left[\adv_{\pi}(f^{\oplus k-1}(X_{<k}, Y_{<k})\mid \bM^{(\pi_{<k})}, S\cap U)\right] \\
		&\geq\E_{\pi\mid S\cap U}\left[\adv_{\pi}(f^{\oplus k}(X, Y)\mid X_k, Y_{<k}, \bM, S\cap U)\right].
	\end{align*}
	This proves the first part of the lemma.
	\bigskip

	For the second part, let $T$ be the set of all triples $(X_k, Y_{<k}, \bM)$ such that $$\pi(X_k, Y_{<k}, \bM, S\cap U)>0$$ {and}
	\[
		\adv_{\pi}(f(X_k, Y_k)\mid X_k, Y_{<k}, \bM, S\cap U)\geq \eta^{1/2}.
	\]
	Then by Markov's inequality, if we have
	\[
		\E_{\pi\mid S\cap U}\left[\adv_{\pi}(f(X_k, Y_k)\mid X_k, Y_{<k}, \bM, S\cap U)\right]\leq \eta,
	\]
	then $\pi(T\mid S\cap U)\leq \eta^{1/2}$.

	Hence, for $(X_k, Y_{<k}, \bM)\notin T$ and $\pi(X_k, Y_{<k}, \bM, S\cap U)>0$, Equation~\eqref{eqn_adv_product} implies that
	\begin{align*}
		&\kern-2em\adv_{\pi}(f^{\oplus k}(X, Y)\mid X_k, Y_{<k}, \bM, S\cap U) \\
		&\leq \eta^{1/2}\cdot \adv_{\pi}(f^{\oplus k-1}(X_{<k}, Y_{<k})\mid X_k, Y_{<k}, \bM, S\cap U). 
	\end{align*}

	Hence, we have
	\begin{align*}
		&\kern-2em\E_{\pi\mid S\cap U}\left[\adv_{\pi}(f^{\oplus k-1}(X_{<k}, Y_{<k})\mid X_k, Y_{<k}, \bM, S\cap U)\right] \\
		&\geq \sum_{(X_k, Y_{<k}, \bM)\notin T}\pi(X_k, Y_{<k}, \bM\mid S\cap U)\cdot \adv_{\pi}(f^{\oplus k-1}(X_{<k}, Y_{<k})\mid X_k, Y_{<k}, \bM, S\cap U) \\
		&\geq \eta^{-1/2}\cdot \sum_{(X_k, Y_{<k}, \bM)\notin T}\pi(X_k, Y_{<k}, \bM\mid S\cap U)\cdot \adv_{\pi}(f^{\oplus k}(X, Y)\mid X_k, Y_{<k}, \bM, S\cap U) \\
		&\geq \eta^{-1/2}\cdot \E_{\pi\mid S\cap U}\left[\adv_{\pi}(f^{\oplus k}(X, Y)\mid X_k, Y_{<k}, \bM, S\cap U)\right] \\
		&\qquad-\eta^{-1/2}\cdot \pi(T\mid S\cap U)\cdot \E_{\pi\mid T\cap S\cap U}\left[\adv_{\pi}(f^{\oplus k}(X, Y)\mid X_k, Y_{<k}, \bM, S\cap U)\right].
	\end{align*}
	Next we show that (the absolute value of) the second term is at most half of the first term.
	First since $T$ is a set of $(X_k, Y_{<k}, \bM)$, we have
	$$\adv_{\pi}(f^{\oplus k}(X, Y)\mid X_k, Y_{<k}, \bM, S\cap U)=\adv_{\pi}(f^{\oplus k}(X, Y)\mid X_k, Y_{<k}, \bM, T\cap S\cap U)$$ for any $(X_k, Y_{<k}, \bM)\in T$.
	Hence, by the fact that $T\cap S\cap U\in\cS_{\pa}(\pi)$ and $U$ maximizes Equation~\eqref{eqn_max_U}, we have that 
	\begin{align*}
		&\kern1.25em\pi(T\mid S\cap U)\cdot \E_{\pi\mid T\cap S\cap U}\left[\adv_{\pi}(f^{\oplus k}(X, Y)\mid X_k, Y_{<k}, \bM, S\cap U)\right] \\
		&=\pi(S\cap U)^{-1/2}\cdot \pi(T\mid S\cap U)^{1/2}\cdot \left(\pi(T\cap S\cap U)^{1/2}\cdot \E_{\pi\mid T\cap S\cap U}\left[\adv_{\pi}(f^{\oplus k}(X, Y)\mid X_k, Y_{<k}, \bM, T\cap S\cap U)\right]\right) \\
		&\leq \pi(S\cap U)^{-1/2}\cdot \eta^{1/4}\cdot \left(\pi(U)^{1/2}\cdot \E_{\pi\mid U}\left[\adv_{\pi}(f^{\oplus k}(X, Y)\mid X_k, Y_{<k}, \bM, U)\right]\right),
	\end{align*}
	which by the bound on $\eta$, is at most
	\[
		\frac{1}{2}\cdot \E_{\pi\mid S\cap U}\left[\adv_{\pi}(f^{\oplus k}(X, Y)\mid X_k, Y_{<k}, \bM, S\cap U)\right].
	\]
	Thus, we have
	\begin{align*}
		&\kern-2em\E_{\pi\mid S\cap U}\left[\adv_{\pi}(f^{\oplus k-1}(X_{<k}, Y_{<k})\mid X_k, Y_{<k}, \bM, S\cap U)\right] \\
		&\geq \frac{1}{2}\eta^{-1/2}\cdot \E_{\pi\mid S\cap U}\left[\adv_{\pi}(f^{\oplus k}(X, Y)\mid X_k, Y_{<k}, \bM, S\cap U)\right].
	\end{align*}
	Combining it with Equation~\eqref{eqn_adv_xm_xym} proves the lemma.
\end{proof}


\subsection{High costs}\label{sec_high_cost}
We first consider all $(X_k, Y_{<k}, \bM)$ at which $\pi_k$ has high costs. 
We will show that it leads to significant lower $\phi_{k-1}^{\cost}$.

\paragraph{High $\theta$-cost.}
The first set of triples consists of all $(X_k, Y_{<k}, \bM)$ such that
\begin{align*}
	\alpha^{-1/2}&\leq \E_{Y_k\sim \pi\mid X_k, Y_{<k}, \bM, U}\left[\thec{\pi_k}{\mu}[X_k, Y_k, \bM^{(\pi_k)}]\right], \\
	2^{-2^{-5}(C_A-c\log(r/\alpha))}\cdot \pi(U)&\leq\E_{Y_k\sim \pi\mid X_k, Y_{<k}, \bM, U}\left[\chis{\pi_k}{\mu,A}[X_k, Y_k, \bM^{(\pi_k)}]\right], \\
	2^{-2^{-5}(C_B-c\log(r/\alpha))}\cdot \pi(U)&\leq \E_{Y_k\sim \pi\mid X_k, Y_{<k}, \bM, U}\left[\chis{\pi_k}{\mu,B}[X_k, Y_k, \bM^{(\pi_k)}]\right].
\end{align*}
This is the set of triples at which $\pi_k$ has high $\theta$-cost and not-too-low $\chisq$-costs.

Note that $\thec{\pi_k}{\mu}[X_k, Y_k, \bM^{(\pi_k)}]$, $\chis{\pi_k}{\mu,A}[X_k, Y_k, \bM^{(\pi_k)}]$ and $\chis{\pi_k}{\mu,B}[X_k, Y_k, \bM^{(\pi_k)}]$ are functions of $(X_k, Y, \bM)$, hence, they do \emph{not} depend on $X_{<k}$.
Denote this set of $(X_k, Y_{<k}, \bM)$ by $S_{\reducetheta}$.
We will also abuse the notation, and use $S_{\reducetheta}$ to denote the set $\{(X, Y, \bM): (X_k, Y_{<k}, \bM)\in S_{\reducetheta}\}$, which can also be treated as an event.

By applying Lemma~\ref{lem_cost_drop_thec} and Lemma~\ref{lem_cost_drop_chis} to $S_{\reducetheta}$ and the corresponding $\eta$, we obtain the following bounds on the costs of $\pi_{<k}$ conditioned on $S_{\reducetheta\cap U}$:
\begin{align*}
	\log \thec{\pi_{<k}\mid S_{\reducetheta}\cap U}{\mu^{k-1}}&\leq \log \thec{\pi\mid U}{\mu^{k}}+\log (1/\pi(S_{\reducetheta}\mid U))-\frac{1}{2}\log (1/\alpha), \\
	\log \chis{\pi_{<k}\mid S_{\reducetheta}\cap U}{\mu^{k-1}, A}&\leq \log\chis{\pi\mid U}{\mu^{k}, A}+\log (1/\pi(S_{\reducetheta}\mid U))\\
	&\strut\hspace{100pt}+2^{-5}(C_A-c\log(r/\alpha))+\log(1/\pi(U)), \\
	\log \chis{\pi_{<k}\mid S_{\reducetheta}\cap U}{\mu^{k-1}, B}&\leq \log\chis{\pi\mid U}{\mu^{k}, B}+\log (1/\pi(S_{\reducetheta}\mid U))\\
	&\strut\hspace{100pt}+2^{-5}(C_B-c\log(r/\alpha))+\log(1/\pi(U)).
\end{align*}
Thus, it implies that (recall Definition~\ref{def_potential})
\begin{equation}\label{eqn_reducetheta}
	\phi_{k-1}^{\cost}(\pi_{<k}\mid S_{\reducetheta}\cap U)\leq \phi_{k}^{\cost}(\pi\mid U)+3\log (1/\pi(S_{\reducetheta}\mid U))+2\log(1/\pi(U))-\frac{1}{4}\log (1/\alpha),
\end{equation}
where we used the assumption that $C_A-c\log(r/\alpha)>\log(1/\alpha), C_B-c\log(r/\alpha)>\log (1/\alpha)$.

\paragraph{High $\chisq$-cost by Alice.}
The next set consists of all $(X_k, Y_{<k}, \bM)$ such that 
\begin{align*}
	\alpha^{2^{-5}}\cdot \pi(U)&\leq \E_{Y_k\mid X_k, Y_{<k}, \bM\sim \pi\mid U}\left[\thec{\pi_k}{\mu}[X_k, Y_k, \bM^{(\pi_k)}]\right]<\alpha^{-1/2}, \\
	2^{C_A-c\log(r/\alpha)}&\leq \E_{Y_k\mid X_k, Y_{<k}, \bM\sim \pi\mid U}\left[\chis{\pi_k}{\mu,A}[X_k, Y_k, \bM^{(\pi_k)}]\right], \\
	2^{-2^{-5}(C_B-c\log(r/\alpha))}\cdot \pi(U)&\leq \E_{Y_k\mid X_k, Y_{<k}, \bM\sim \pi\mid U}\left[\chis{\pi_k}{\mu,B}[X_k, Y_k, \bM^{(\pi_k)}]\right].
\end{align*}
This is the set of triples at which $\pi_k$ has high $\chisq$-cost by Alice and not-too-low $\theta$-cost and $\chisq$-cost by Bob.
Denote this set of $(X_k, Y_{<k}, \bM)$ by $S_{\reducechisA}$.
The upper bound on the $\theta$-cost ensures that it is disjoint from $S_{\reducetheta}$.
Similarly, we also abuse the notation to let $S_{\reducechisA}$ also denote the set $\{(X, Y, \bM): (X_k, Y_{<k}, \bM)\in S_{\reducechisA}\}$.

By applying Lemma~\ref{lem_cost_drop_thec} and Lemma~\ref{lem_cost_drop_chis} to $S_{\reducechisA}$ and the corresponding $\eta$, we have the following bounds:
\begin{align*}
	\log \thec{\pi_{<k}\mid S_{\reducechisA}\cap U}{\mu^{k-1}}&\leq \log \thec{\pi\mid U}{\mu^{k}}+\log (1/\pi(S_{\reducechisA}\mid U)) \\
	&\strut\hspace{100pt}+2^{-5}\log (1/\alpha)+\log(1/\pi(U)), \\
	\log \chis{\pi_{<k}\mid S_{\reducechisA}\cap U}{\mu^{k-1}, A}&\leq \log\chis{\pi\mid U}{\mu^{k}, A}+\log (1/\pi(S_{\reducechisA}\mid U))\\
	&\strut\hspace{100pt}-(C_A-c\log(r/\alpha)), \\
	\log \chis{\pi_{<k}\mid S_{\reducechisA}\cap U}{\mu^{k-1}, B}&\leq \log\chis{\pi\mid U}{\mu^{k}, B}+\log (1/\pi(S_{\reducechisA}\mid U))\\
	&\strut \hspace{100pt}+2^{-5}(C_B-c\log(r/\alpha))+\log(1/\pi(U)),
\end{align*}
which also implies
\begin{equation}\label{eqn_reducechisA}
	\phi_{k-1}^{\cost}(\pi_{<k}\mid S_{\reducechisA}\cap U)\leq \phi_{k}^{\cost}(\pi\mid U)+3\log (1/\pi(S_{\reducechisA}\mid U))+2\log(1/\pi(U))-\frac{1}{4}\log (1/\alpha).
\end{equation}

\paragraph{High $\chisq$-cost by Bob.}
The third case consists of all $(X_k, Y_{<k}, \bM)$ such that 
\begin{align*}
	\alpha^{2^{-5}}\cdot \pi(U)&\leq \E_{Y_k\sim \pi\mid X_k, Y_{<k}, \bM, U}\left[\thec{\pi_k}{\mu}[X_k, Y_k, \bM^{(\pi_k)}]\right]<\alpha^{-1/2}, \\
	2^{-2^{-5}(C_A-c\log(r/\alpha))}\cdot \pi(U)&\leq \E_{Y_k\sim \pi\mid X_k, Y_{<k}, \bM, U}\left[\chis{\pi_k}{\mu,A}[X_k, Y_k, \bM^{(\pi_k)}]\right]<2^{C_A-c\cdot\log(r/\alpha)}, \\
	2^{C_B-c\log(r/\alpha)}&\leq \E_{Y_k\sim \pi\mid X_k, Y_{<k}, \bM, U}\left[\chis{\pi_k}{\mu,B}[X_k, Y_k, \bM^{(\pi_k)}]\right].
\end{align*}
This is the set of triples at which $\pi_k$ has high $\chisq$-cost by Bob and not-too-low $\theta$-cost and $\chisq$-cost by Alice.
Denote this set of $(X_k, Y_{<k}, \bM)$ by $S_{\reducechisB}$.
It is disjoint from $S_{\reducetheta}$ and $S_{\reducechisA}$.
Similarly, we also use $S_{\reducechisB}$ to denote the set $\{(X, Y, \bM): (X_k, Y_{<k}, \bM)\in S_{\reducechisB}\}$.

By applying Lemma~\ref{lem_cost_drop_thec} and Lemma~\ref{lem_cost_drop_chis} to $S_{\reducechisB}$ and the appropriate $\eta$, we have the following bounds:
\begin{align*}
	\log \thec{\pi_{<k}\mid S_{\reducechisB}\cap U}{\mu^{k-1}}&\leq \log \thec{\pi\mid U}{\mu^{k}}+\log (1/\pi(S_{\reducechisB}\mid U))\\
	&\strut\hspace{100pt}+2^{-5}\log (1/\alpha)+\log(1/\pi(U)), \\
	\log \chis{\pi_{<k}\mid S_{\reducechisB}\cap U}{\mu^{k-1}, A}&\leq \log\chis{\pi\mid U}{\mu^{k}, A}+\log (1/\pi(S_{\reducechisB}\mid U)) \\
	&\strut\hspace{100pt}+2^{-5}(C_A-c\log(r/\alpha))+\log(1/\pi(U)), \\
	\log \chis{\pi_{<k}\mid S_{\reducechisB}\cap U}{\mu^{k-1}, B}&\leq \log\chis{\pi\mid U}{\mu^{k}, B}+\log (1/\pi(S_{\reducechisB}\mid U))\\
	&\strut\hspace{100pt}-(C_B-c\log(r/\alpha)),
\end{align*}
which also implies that
\begin{equation}\label{eqn_reducechisB}
	\phi_{k-1}^{\cost}(\pi_{<k}\mid S_{\reducechisB}\cap U)\leq \phi_{k}^{\cost}(\pi\mid U)+3\log (1/\pi(S_{\reducechisB}\mid U))+2\log(1/\pi(U))-\frac{1}{4}\log (1/\alpha).	
\end{equation}

\bigskip

Equation~\eqref{eqn_reducetheta},~\eqref{eqn_reducechisA} and~\eqref{eqn_reducechisB} implies that for $\beta\in\{\reducetheta, \reducechisA, \reducechisB\}$, we all have
\begin{align}
	\phi_{k-1}^{\cost}(\pi_{<k}\mid S_{\beta}\cap U)&\leq \phi_{k}^{\cost}(\pi\mid U)+3\log (1/\pi(S_{\beta}\mid U))+2\log(1/\pi(U))-\frac{1}{4}\log (1/\alpha) \nonumber\\
	&\leq \phi_{k}^{\cost}(\pi\mid U)+3\log (1/\pi(S_{\beta}\cap U))-\frac{1}{4}\log (1/\alpha).\label{eqn_cost_decs}
\end{align}

The main lemma of this subsection is the following, stating that if the above three sets contribute a nontrival amount of total advantage in $U$ (weighted by the probability), then we can construct a protocol for $f^{\oplus k-1}$ satisfying the requirements of Lemma~\ref{lem_induction} (by conditioning $\pi_{<k}$ on a carefully chosen event).
\begin{lemma}\label{lem_high_cost}
	Let $S_{\reducecost}$ be the union $S_{\reducetheta}\cup S_{\reducechisA}\cup S_{\reducechisB}$.
	If we have
	\begin{align*}
		&\kern-2em\pi(S_{\reducecost}\cap U)\cdot \E_{\pi\mid S_{\reducecost}\cap U}\left[\adv_{\pi}(f^{\oplus k}(X, Y)\mid X_k, Y_{<k}, \bM, S_{\reducecost}\cap U)\right]\\
		&\geq \frac{1}{3}\cdot \pi(U)\cdot \E_{\pi\mid U}\left[\adv_{\pi}(f^{\oplus k}(X, Y)\mid X_k, Y_{<k}, \bM, U))\right],
	\end{align*}
	then Lemma~\ref{lem_induction} holds.
\end{lemma}

Protocol $\pi_{<k}$ and event $S_{\beta}\cap U$ may be one potential choice for $\pen$ and $V_{\newn}$ in Lemma~\ref{lem_induction}.
However, Lemma~\ref{lem_induction} requires the probability of $V_{\newn}$ to be $\Omega(1)$, which is not necessarily true for any $\beta$, since $U$ may have very small probability.
On the other hand, we could also consider setting $\pen$ to the distribution of $\pi_{<k}$ conditioned on $S_{\beta}\cap U$ and $V_{\newn}$ to the entire sample space, but this protocol may have very large costs.

To resolve this issue, we will use the following lemma, which turns $(\pi_{<k}\mid S_{\beta}\cap U)$ into a protocol $(\pi_{<k})_{G}$ with bounded costs for some event $G\approx S_{\beta}\cap U$.
Moreover, by dividing $G$ into $G_0\cup G_1$ according to whether the function value is more likely to be $0$ or $1$ conditioned on $Y$ and $\bM$, the lemma guarantees that the costs conditioned on $G_b$ are also bounded (for $b=0,1$).
This will allow us to apply Lemma~\ref{lem_cost_decs_adv} later to lower bound the advantage.

\newcommand{\lemresolveeventcont}{
	Fix any $\gamma\in(0,1/2)$.
	Let $\rho$ be an $r$-round generalized protocol over $\cX\times \cY\times \cM$, $W$ be an event, $\nu$ be an input distribution and $h:\cX\times \cY\rightarrow \{0,1\}$ be a function of the inputs.
	Then there exists a \emph{partition} of $W$ into three events $G, B_0, B_1$ and a \emph{partition} of $\cY\times \cM$ into $E_0,E_1$ such that the following holds:
	\begin{enumerate}
		\item all three events $G$, $B_0$, $B_1$ have the form $W\cap S$ for some $S\in \cS_{\rect}(\rho)$;
		\item $\rho(B_0\cup B_1\mid W)\leq \gamma$;
		\item let $\rho_G$ be the protocol $(\rho\mid G)$, $G_0=G\cap E_0$ and $G_1=G\cap E_1$, then for $b=0,1$,
			\begin{align*}
				\log\thec{\rho_{G}\mid G_b}{\nu}&\leq \log \thec{\rho\mid W}{\nu}+(r+1)\log \left((r+3)/\gamma\right)+\log \left(1/((1-\gamma)\rho(G_b))\right), \\
				\log\chis{\rho_{G}\mid G_b}{\nu,A}&\leq \log \chis{\rho\mid W}{\nu,A}+\log ((r+3)/\gamma)+\log (1/((1-\gamma)\rho(G_b))), \\
				\log\chis{\rho_{G}\mid G_b}{\nu,B}&\leq \log \chis{\rho\mid W}{\nu,B}+\log ((r+3)/\gamma)+\log (1/((1-\gamma)\rho(G_b)));
			\end{align*}
		\item for $b=0,1$, and all $(Y, \bM)$ such that $\rho(Y, \bM\mid G_b)>0$, 
			\[
				\rho(h(X, Y)=b\mid Y, \bM, G_b)\geq 1/2.
			\]
	\end{enumerate}
}

\begin{lemma}\label{lem_resolve_event}
	\lemresolveeventcont
\end{lemma}

Note that we upper bound the costs of $\rho_G$ conditioned on $G_b$ by the costs of $\rho$ conditioned on $W$ (plus some small quantity).
Thus, for $\rho=\pi_{<k}$ and $W=S_{\beta}\cap U$, the costs are bounded due to Equation~\eqref{eqn_cost_decs}.
To focus on our main proof, we will defer the proof of Lemma~\ref{lem_resolve_event} to Section~\ref{sec_resolve_event}.
Now we use it to prove Lemma~\ref{lem_high_cost}.

\begin{proof}[Proof of Lemma~\ref{lem_high_cost}]
	We first fix some $\beta\in\{\reducetheta, \reducechisA, \reducechisB\}$.
	By applying Lemma~\ref{lem_resolve_event} to protocol $\rho=\pi_{<k}$, event $W=S_{\beta}\cap U$, input distribution $\nu=\mu^{k-1}$ and function $h=f^{\oplus k-1}$ for $\gamma=2^{-12}$, we obtain sets $G_{\beta}, B_{\beta,0}, B_{\beta,1}, E_{\beta,0}, E_{\beta,1}$.
	Let $G_{\beta,0}=G_{\beta}\cap E_{\beta,0}$, $G_{\beta,1}=G_{\beta}\cap E_{\beta,1}$, and $(\pi_{<k})_{G_{\beta}}$ be the distribution $\pi_{<k}$ conditioned on $G_{\beta}$.
	The lemma guarantees that $G_{\beta}, B_{\beta,0}, B_{\beta,1}$ all have the form $S_{\beta}\cap U\cap S$ for some $S\in \cS_{\rect}(\pi_{<k})\subseteq\cS_{\pa}(\pi)$ (Proposition~\ref{prop_pi<k_pik_S_rect}).
	Since $S_{\beta}\cap U\in \cS_{\pa}(\pi)$, we have that $G_{\beta}, B_{\beta,0}, B_{\beta,1}\in \cS_{\pa}(\pi)$.
	Since $E_{\beta,0}, E_{\beta,1}\in \cS_{\rect}(\pi_{<k})$, we also have $G_{\beta,0}, G_{\beta,1}\in \cS_{\pa}(\pi)$.
	\bigskip

	For each $\beta$, since $G_{\beta}\in\cS_{\pa}(\pi)$, Proposition~\ref{prop_cs_partial}(ii) implies that $(\pi\mid G_{\beta})$ has the partial rectangle property with respect to $\mu^k$.
	Then Proposition~\ref{prop_pi<k_pk_rect} implies that $(\pi_{<k}\mid G_{\beta})$, i.e., $(\pi_{<k})_{G_{\beta}}$, has the rectangle property with respect to $\mu^{k-1}$.
	For each $\beta,b$, since $E_{\beta,b}\in\cS_{\rect}(\pi_{<k})=\cS_{\rect}((\pi_{<k})_{G_{\beta}})$, the protocol $(\pi_{<k})_{G_{\beta}}$ and the event $E_{\beta,b}$ are one candidate for $\pen$ and $V_{\newn}$ in Lemma~\ref{lem_induction}.
	We will prove the following sufficient condition for them to satisfy the requirements of Lemma~\ref{lem_induction}.
	\begin{claim}\label{cl_high_cost}
		If we have 
		\begin{equation}\label{eqn_assumption_beta_b_adv}
			\begin{aligned}
			&\kern-2em\pi(G_{\beta,b})\cdot \E_{\pi\mid G_{\beta,b}}\left[\adv_{\pi}(f^{\oplus k}(X, Y)\mid X_k, Y_{<k}, \bM, G_{\beta,b}))\right] \\
			&\geq 2^{-6}\cdot \pi(U)\cdot \E_{\pi\mid U}\left[\adv_{\pi}(f^{\oplus k}(X, Y)\mid X_k, Y_{<k}, \bM, U))\right],
			\end{aligned}
		\end{equation}
		then $(\pi_{<k})_{G_{\beta}}$ and event $E_{\beta,b}$ satisfy the requirements of Lemma~\ref{lem_induction} for $\pen$ and $V_{\newn}$.
	\end{claim}
	Before proving the claim, we first show that it implies Lemma~\ref{lem_high_cost}.
	If Equation~\eqref{eqn_assumption_beta_b_adv} holds for any $\beta\in\{\reducetheta, \reducechisA, \reducechisB\}$ and $b\in\{0,1\}$, then the lemma holds.
	Otherwise we must have for every $\beta$ and $b$,
	\begin{align*}
		&\kern-2em\pi(G_{\beta,b})\cdot \E_{\pi\mid G_{\beta,b}}\left[\adv_{\pi}(f^{\oplus k}(X, Y)\mid X_k, Y_{<k}, \bM, G_{\beta,b}))\right] \\
		&<2^{-6}\cdot \pi(U)\cdot \E_{\pi\mid U}\left[\adv_{\pi}(f^{\oplus k}(X, Y)\mid X_k, Y_{<k}, \bM, U))\right].
	\end{align*}
	On the other hand, since $B_{\beta,b}\in\cS_{\pa}(\pi)$, $B_{\beta,b}\subseteq S_{\beta}\cap U$ and $\pi(B_{\beta,b}\mid U)\leq \pi(B_{\beta,b}\mid S_{\beta}\cap U)\leq 2^{-12}$, by Proposition~\ref{prop_max_U}(i), we also have
	\begin{align*}
		&\kern-2em\pi(B_{\beta,b})\cdot \E_{\pi\mid B_{\beta,b}}\left[\adv_{\pi}(f^{\oplus k}(X, Y)\mid X_k, Y_{<k}, \bM, B_{\beta,b}))\right] \\
		&\leq 2^{-6}\cdot \pi(U)\cdot \E_{\pi\mid U}\left[\adv_{\pi}(f^{\oplus k}(X, Y)\mid X_k, Y_{<k}, \bM, U))\right].
	\end{align*}
	Since $G_{\beta,0}\cup G_{\beta,1}\cup B_{\beta,0}\cup B_{\beta,1}=S_{\beta}\cap U$, $S_{\reducecost}=S_{\reducetheta}\cup S_{\reducechisA}\cup S_{\reducechisB}$, and all 12 sets are disjoint, by summing up the above inequalities for all $B_{\beta,b}$ and $G_{\beta,b}$ and applying Lemma~\ref{lem_adv_super_additivity}, we have
	\begin{align*}
		&\kern-2em\pi(S_{\reducecost}\cap U)\cdot \E_{\pi\mid S_{\reducecost}\cap U}\left[\adv_{\pi}(f^{\oplus k}(X, Y)\mid X_k, Y_{<k}, \bM, S_{\reducecost}\cap U))\right]\\
		&<\frac{1}{3}\cdot\pi(U)\cdot \E_{\pi\mid U}\left[\adv_{\pi}(f^{\oplus k}(X, Y)\mid X_k, Y_{<k}, \bM, U))\right],
	\end{align*}
	contradicting with the lemma premise.
	\bigskip

	Now it suffices to prove the claim.	
	We first observe that by Proposition~\ref{prop_max_U}(i), Equation~\eqref{eqn_assumption_beta_b_adv} also implies that $\pi\left(G_{\beta,b}\mid U\right)^{1/2}\geq 2^{-6}$, which in turn, implies that $\pi(E_{\beta,b}\mid G_{\beta})=\pi(G_{\beta,b}\mid G_{\beta})\geq 2^{-12}$, i.e., the probability of $E_{\beta,b}$ in the distribution $(\pi_{<k})_{G_{\beta}}$ is at least $2^{-12}$, as required by Lemma~\ref{lem_induction}.
	In the following, we show that the bound on $\phi_{k-1}\left((\pi_{<k})_{G_{\beta}}\mid E_{\beta,b}\right)=\phi_{k-1}\left((\pi_{<k})_{G_{\beta}}\mid G_{\beta,b}\right)$ also holds.

	\paragraph{Bounding $\phi_{k-1}^{\cost}((\pi_{<k})_{G_{\beta}}\mid G_{\beta,b})$.} We first bound its $\phi_{k-1}^{\cost}$ value.
	By Lemma~\ref{lem_resolve_event} and the fact that $\log(1/\alpha)<C_A-c\log(r/\alpha)$ and $\log(1/\alpha)<C_B-c\log(r/\alpha)$, we have
	\begin{align}
		&\kern1.25em\phi_{k-1}^{\cost}((\pi_{<k})_{G_{\beta}}\mid G_{\beta,b}) \nonumber\\
		&=\log\thec{(\pi_{<k})_{G_{\beta}}\mid G_{\beta,b}}{\mu^{k-1}}+\frac{\log(1/\alpha)}{C_A-c\log (r/\alpha)}\cdot\log\chis{(\pi_{<k})_{G_{\beta}}\mid G_{\beta,b}}{\mu^{k-1},A}\nonumber\\
		&\qquad+\frac{\log(1/\alpha)}{C_B-c\log (r/\alpha)}\cdot \log \chis{(\pi_{<k})_{G_{\beta}}\mid G_{\beta,b}}{\mu^{k-1},B} \nonumber\\
		&\leq \log\thec{\pi_{<k}\mid S_{\beta}\cap U}{\mu^{k-1}}+\frac{\log(1/\alpha)}{C_A-c\log (r/\alpha)}\cdot\log\chis{\pi_{<k}\mid S_{\beta}\cap U}{\mu^{k-1},A} \nonumber\\
		&\qquad+\frac{\log(1/\alpha)}{C_B-c\log (r/\alpha)}\cdot \log \chis{\pi_{<k}\mid S_{\beta}\cap U}{\mu^{k-1},B}\nonumber\\
		&\qquad+(r+3)\log\left((r+3)/\gamma\right)+3\log (1/\left((1-\gamma)\pi_{<k}(G_{\beta,b})\right)) \nonumber\\
		&=\phi_{k-1}^{\cost}(\pi_{<k}\mid S_{\beta}\cap U)+(r+3)\log\left((r+3)/\gamma\right)+3\log (1/\left((1-\gamma)\pi_{<k}(G_{\beta,b})\right)) \nonumber
		\intertext{which by Equation~\eqref{eqn_cost_decs}, is at most}
		&\leq \phi_{k}^{\cost}(\pi\mid U)+O(r\log r)+3\log (1/\pi_{<k}(G_{\beta,b}))+3\log (1/\pi(S_{\beta}\cap U))-\frac{1}{4}\log(1/\alpha) \nonumber \\
		&\leq \phi_{k}^{\cost}(\pi\mid U)+6\log (1/\pi(G_{\beta,b}))-\frac{1}{8}\log(1/\alpha),\label{eqn_cost_decs_2}
	\end{align}
	where we used the assumption that $\log (1/\alpha)>cr\log r$ for a sufficiently large $c$, and the fact that $\pi(S_{\beta}\cap U)\geq \pi(G_{\beta,b})$, and the fact that $G_{\beta,b}$ can also be viewed as an event in $\pi$.

	\paragraph{Bounding $\phi_{k-1}^{\adv}((\pi_{<k})_{G_{\beta}}\mid G_{\beta,b})$.}Next we bound its $\phi_{k-1}^{\adv}$ value.
	Lemma~\ref{lem_resolve_event} also guarantees that for all $(X_k, Y_{<k}, \bM)$ such that $\pi(X_k, Y_{<k}, \bM\mid G_{\beta,b})>0$ (recall that $\bM^{(\pi_{<k})}=(X_k, \bM)$), we have
	\[
		\pi(f^{\oplus k-1}(X_{<k}, Y_{<k})=b\mid X_k, Y_{<k}, \bM, G_{\beta,b})\geq 1/2.
	\]
	This allows us to bound its advantage by applying the first part of Lemma~\ref{lem_cost_decs_adv} for $S=G_{\beta,b}\in\cS_{\pa}(\pi)$ (note that $G_{\beta,b}\subseteq U$).
	Thus, it implies
	\begin{align*}
		&\E_{\pi\mid G_{\beta,b}}\left[\adv_\pi(f^{\oplus k-1}(X_{<k}, Y_{<k})\mid \bM^{(\pi_{<k})}, G_{\beta,b})\right]\\
		&\qquad\geq \E_{\pi\mid G_{\beta,b}}\left[\adv_\pi(f^{\oplus k}(X, Y)\mid X_k, Y_{<k}, \bM, G_{\beta,b})\right].
	\end{align*}
	Note that the LHS is exactly the expected advantage of protocol $(\pi_{<k})_{G_{\beta}}$ \emph{conditioned on} $G_{\beta,b}$:
	\begin{align*}
		&\kern-2em\E_{\pi\mid G_{\beta,b}}\left[\adv_\pi(f^{\oplus k-1}(X_{<k}, Y_{<k})\mid \bM^{(\pi_{<k})}, G_{\beta,b})\right]\\
		&=\E_{(\pi_{<k})\mid G_{\beta}, G_{\beta,b}}\left[\adv_{\pi_{<k}}(f^{\oplus k-1}(X_{<k}, Y_{<k})\mid \bM^{(\pi_{<k})}, G_{\beta}, G_{\beta,b})\right] \\
		&=\E_{(\pi_{<k})_{G_{\beta}}\mid G_{\beta,b}}\left[\adv_{(\pi_{<k})_{G_{\beta}}}(f^{\oplus k-1}(X_{<k}, Y_{<k})\mid \bM^{((\pi_{<k})_{G_{\beta}})}, G_{\beta,b})\right].
	\end{align*}
	Thus, by definition, that is
	\begin{equation}\label{eqn_adv_decs}
		\phi_{k-1}^{\adv}((\pi_{<k})_{G_{\beta}}\mid G_{\beta,b})\leq \phi_{k,\pa}^{\adv}(\pi\mid {G_{\beta,b}}).
	\end{equation}

	\paragraph{Bounding $\phi_{k-1}((\pi_{<k})_{G_{\beta}}\mid G_{\beta,b})$.} Now we sum up the two parts of the potential function.
	By Equation~\eqref{eqn_cost_decs_2} and~\eqref{eqn_adv_decs}, we have
	\begin{align}
		&\kern1.25em\phi_{k-1}((\pi_{<k})_{G_{\beta}}\mid G_{\beta,b}) \nonumber\\
		&=\phi_{k-1}^{\cost}((\pi_{<k})_{G_{\beta}}\mid G_{\beta,b})+\phi_{k-1}^{\adv}((\pi_{<k})_{G_{\beta}}\mid G_{\beta,b})\nonumber\\
		&\leq \phi_{k}^{\cost}(\pi\mid U)+\phi_{k,\pa}^{\adv}(\pi\mid {G_{\beta,b}})+6\log (1/\pi(G_{\beta,b}))-\frac{1}{8}\log(1/\alpha) \nonumber\\
		&=\phi_{k,\pa}(\pi\mid U)-\phi_{k,\pa}^{\adv}(\pi\mid U)+\left(\phi_{k,\pa}^{\adv}(\pi\mid {G_{\beta,b}})+6\log (1/\pi(G_{\beta,b}))\right)-\frac{1}{8}\log(1/\alpha)\nonumber
		\intertext{which by Lemma~\ref{lem_max_U}, is}
		&\leq \phi_k(\pi\mid V)-13\log (1/\pi(U\mid V))-\phi_{k,\pa}^{\adv}(\pi\mid U)+\left(\phi_{k,\pa}^{\adv}(\pi\mid {G_{\beta,b}})+6\log (1/\pi(G_{\beta,b}))\right)-\frac{1}{8}\log(1/\alpha) \nonumber
		\intertext{which by the fact that $\log(1/\pi(U\mid V))\geq 0$ and $\pi(V)\geq 2^{-12}$, is}
		&\leq \phi_k(\pi\mid V)-6(\log (1/\pi(U))-\log (1/\pi(V)))-\phi_{k,\pa}^{\adv}(\pi\mid U)+\left(\phi_{k,\pa}^{\adv}(\pi\mid {G_{\beta,b}})+6\log (1/\pi(G_{\beta,b}))\right) \nonumber \\
		&\qquad-\frac{1}{8}\log(1/\alpha) \nonumber \\
		&\leq \phi_k(\pi\mid V)-\left(\phi_{k,\pa}^{\adv}(\pi\mid U)+6\log (1/\pi(U))\right)+\left(\phi_{k,\pa}^{\adv}(\pi\mid {G_{\beta,b}})+6\log (1/\pi(G_{\beta,b}))\right) \nonumber\\
		&\qquad-\frac{1}{8}\log(1/\alpha)+72.\label{eqn_phi_decs}
	\end{align}
	Finally, by Equation~\eqref{eqn_assumption_beta_b_adv} and Proposition~\ref{prop_max_U}(ii) (for $s=2^6$ and $t=6$), we have
	\[
		\left(\phi_{k,\pa}^{\adv}(\pi\mid U)+6\log (1/\pi(U))\right)\geq\left(\phi_{k,\pa}^{\adv}(\pi\mid {G_{\beta,b}})+6\log (1/\pi(G_{\beta,b}))\right)-192.
	\]
	Plugging it into Equation~\eqref{eqn_phi_decs}, we obtain $$\phi_{k-1}((\pi_{<k})_{G_{\beta}}\mid G_{\beta,b})\leq \phi_k(\pi\mid V)-\frac{1}{16}\log (1/\alpha),$$ since $\alpha\leq r^{-cr}$ for a sufficiently large constant $c$.
	Hence, $(\pi_{<k})_{G_{\beta}}$ and $E_{\beta,b}$ satisfy the requirements of Lemma~\ref{lem_induction}.
	This proves the claim, completing the proof of Lemma~\ref{lem_high_cost}.
\end{proof}


\subsection{Low costs}\label{sec_low_cost}
Now we consider the case where $\pi_k$ has low costs.
It consists of all $(X_k, Y_{<k}, \bM)$ such that
\begin{align*}
	\alpha^{2^{-5}}\cdot \pi(U)&\leq \E_{Y_k\sim \pi\mid X_k, Y_{<k}, \bM, U}\left[\thec{\pi_k}{\mu}[X_k, Y_k, \bM^{(\pi_k)}]\right]<\alpha^{-1/2}, \\
	2^{-2^{-5}(C_A-c\log(r/\alpha))}\cdot \pi(U)&\leq \E_{Y_k\sim \pi\mid X_k, Y_{<k}, \bM, U}\left[\chis{\pi_k}{\mu,A}[X_k, Y_k, \bM^{(\pi_k)}]\right]<2^{C_A-c\log(r/\alpha)}, \\
	2^{-2^{-5}(C_B-c\log(r/\alpha))}\cdot \pi(U)&\leq \E_{Y_k\sim \pi\mid X_k, Y_{<k}, \bM, U}\left[\chis{\pi_k}{\mu,B}[X_k, Y_k, \bM^{(\pi_k)}]\right]<2^{C_B-c\log(r/\alpha)}.
\end{align*}
This is the set of triples at which the costs of $\pi_k$ are not high, nor too low.
Denote this set of $(X_k, Y_{<k}, \bM)$ by $S_{\good}$.
By definition, it is disjoint from $S_{\reducecost}$.
Similarly, we also use $S_{\good}$ to denote set $\{(X, Y, \bM): (X_k, Y_{<k}, \bM)\in S_{\good}\}$.

By applying Lemma~\ref{lem_cost_drop_thec} and Lemma~\ref{lem_cost_drop_chis} to $S_{\good}$ with the appropriate $\eta$, we have the following bounds:
\begin{align*}
	\log \thec{\pi_{<k}\mid S_{\good}\cap U}{\mu^{k-1}}&\leq \log \thec{\pi\mid U}{\mu^{k}}+\log (1/\pi(S_{\good}\mid U))\\
	&\strut\hspace{100pt}+2^{-5}\log (1/\alpha)+\log(1/\pi(U)), \\
	\log \chis{\pi_{<k}\mid S_{\good}\cap U}{\mu^{k-1}, A}&\leq \log\chis{\pi\mid U}{\mu^{k}, A}+\log (1/\pi(S_{\good}\mid U))\\
	&\strut\hspace{100pt}+2^{-5}(C_A-c\log(r/\alpha))+\log(1/\pi(U)), \\
	\log \chis{\pi_{<k}\mid S_{\good}\cap U}{\mu^{k-1}, B}&\leq \log\chis{\pi\mid U}{\mu^{k}, B}+\log (1/\pi(S_{\good}\mid U))\\
	&\strut\hspace{100pt}+2^{-5}(C_B-c\log(r/\alpha))+\log(1/\pi(U)).
\end{align*}
Therefore, we have
\begin{align}
	\phi_{k-1}^{\cost}(\pi_{<k}\mid S_{\good}\cap U)&\leq \phi_{k}^{\cost}(\pi\mid U)+3\log (1/\pi(S_{\good}\mid U))+3\log(1/\pi(U))+3\cdot 2^{-5}\log(1/\alpha) \nonumber\\
	&=\phi_{k}^{\cost}(\pi\mid U)+3\log (1/\pi(S_{\good}\cap U))+3\cdot 2^{-5}\log(1/\alpha).\label{eqn_cost_low_cost}
\end{align}

The main lemma of this subsection is the following, stating that if $S_{\good}$ contributes a nontrival amount of total advantage in $U$, then we can construct a protocol for $f^{\oplus k-1}$ satisfying the requirements of Lemma~\ref{lem_induction}.
\begin{lemma}\label{lem_low_cost}
	If we have
	\begin{align*}
		&\kern-2em\pi(S_{\good}\cap U)\cdot \E_{\pi\mid S_{\good}\cap U}\left[\adv_{\pi}(f^{\oplus k}(X, Y)\mid X_k, Y_{<k}, \bM, S_{\good}\cap U)\right] \\
		&\geq \frac{1}{3}\cdot \pi(U)\cdot \E_{\pi\mid U}\left[\adv_{\pi}(f^{\oplus k}(X, Y)\mid X_k, Y_{<k}, \bM, U))\right],
	\end{align*}
	then Lemma~\ref{lem_induction} holds.
\end{lemma}

\newcommand{\lemcompressioncont}{
	Let $\delta_1,\delta_2\in(0,1/2)$ be any fixed parameter.
	Let $\rho$ be an $r$-round generalized protocol and let $W$ be an event such that $(\rho\mid W)$ has the rectangle property with respect to $\mu$.
	Then for any function $f:\cX\times \cY\rightarrow\{0,1\}$, there is an $r$-round \emph{standard} protocol $\tau$ such that
	\begin{itemize}
		\item in odd rounds of $\tau$, Alice sends a message of at most $\log \chis{\rho\mid W}{\mu,A}+O(\log (r/\delta_1\delta_2)+\log\log \thec{\rho\mid W}{\mu})$ bits;
		\item in even rounds of $\tau$, Bob sends a message of at most $\log\chis{\rho\mid W}{\mu,B}+O(\log (r/\delta_1\delta_2))$ bits;
		\item $\tau$ computes $f$ correctly under input distribution $\mu$ with probability at least
		\[
			\frac{1}{2}+\frac{\delta_1}{32\thec{\rho\mid W}{\mu}}\left(\rho(W)\cdot\E_{\rho\mid W}\left[\adv_{\rho}(f(X, Y)\mid X, \bM, W)\right]-6\delta_1\right)-2r\delta_2.
		\]
	\end{itemize}
}

The proof will use the following lemma that converts a generalized protocol with low costs to a standard protocol with low communication.
\begin{lemma}\label{lem_compression}
	\lemcompressioncont
\end{lemma}
To focus on the main proof, we will defer the proof of Lemma~\ref{lem_compression} to Section~\ref{sec_compression}.

\begin{proof}[Proof of Lemma~\ref{lem_low_cost}]
	Similar to the proof of Lemma~\ref{lem_high_cost}, we first apply Lemma~\ref{lem_resolve_event} to $\rho=\pi_{<k}$, event $W=S_{\good}\cap U$, input distribution $\nu=\mu^{k-1}$ and function $h=f^{\oplus k-1}$ for $\gamma=2^{-12}$.
	We obtain sets $G_{\good}$, $B_{\good,0}$, $B_{\good,1}$ and $E_{\good,0}$, $E_{\good,1}$.
	Let $G_{\good,0}=G_{\good}\cap E_{\good,0}, G_{\good,1}=G_{\good}\cap E_{\good, 1}$, and $(\pi_{<k})_{G_{\good}}$ be $\pi_{<k}$ conditioned on $G_{\good}$.
	Again, we have that $G_{\good,0}$, $G_{\good,1}$, $B_{\good,0}$, $B_{\good,1}$ and $G_{\good}\in \cS_{\pa}(\pi)$, $(\pi_{<k})_{G_{\good}}$ has the rectangle property with respect to $\mu^{k-1}$ and for $b=0,1$, $E_{\good,b}\in \cS_{\rect}((\pi_{<k})_{G_{\good}})$.

	Thus, the protocol $(\pi_{<k})_{G_{\good}}$ and the event $E_{\good,b}$ are one possible candidate for Lemma~\ref{lem_induction}.
	We will prove the following sufficient condition for them to satisfy the requirements of Lemma~\ref{lem_induction}.
	\begin{claim}\label{cl_low_cost}
		If we have
		\begin{equation}\label{eqn_assumption_good_b_adv}
			\begin{aligned}
				&\kern-2em\pi(G_{\good,b})\cdot \E_{\pi\mid G_{\good,b}}\left[\adv_{\pi}(f^{\oplus k}(X, Y)\mid X_k, Y_{<k}, \bM, G_{\good,b}))\right] \\
				&\geq 2^{-6}\cdot \pi(U)\cdot \E_{\pi\mid U}\left[\adv_{\pi}(f^{\oplus k}(X, Y)\mid X_k, Y_{<k}, \bM, U))\right],
			\end{aligned}
		\end{equation}
		then $(\pi_{<k})_{G_{\good}}$ and $E_{\good,b}$ satisfy the requirements of Lemma~\ref{lem_induction}  for $\pen$ and $V_{\newn}$.
	\end{claim}

	Similar to the proof of Lemma~\ref{lem_high_cost}, before proving the claim, we first show that it implies the lemma.
	If Equation~\eqref{eqn_assumption_good_b_adv} holds for either $b=0$ or $b=1$, then the lemma holds.
	Otherwise, we have
	\begin{align*}
		&\kern-2em\pi(G_{\good,b})\cdot \E_{\pi\mid G_{\good,b}}\left[\adv_{\pi}(f^{\oplus k}(X, Y)\mid X_k, Y_{<k}, \bM, G_{\good,b}))\right] \\
		&< 2^{-6}\cdot \pi(U)\cdot \E_{\pi\mid U}\left[\adv_{\pi}(f^{\oplus k}(X, Y)\mid X_k, Y_{<k}, \bM, U))\right],
	\end{align*}
	for $b=0,1$.
	Lemma~\ref{lem_resolve_event} guarantees that $\pi(B_{\good,b}\mid U)\leq \pi(B_{\good,b}\mid S_{\good}\cap U)\leq 2^{-12}$.
	By Proposition~\ref{prop_max_U}(i), we also have 
	\begin{align*}
		&\kern-2em\pi(B_{\good,b})\cdot \E_{\pi\mid B_{\good,b}}\left[\adv_{\pi}(f^{\oplus k}(X, Y)\mid X_k, Y_{<k}, \bM, B_{\good,b}))\right] \\
		&\leq 2^{-6}\cdot \pi(U)\cdot \E_{\pi\mid U}\left[\adv_{\pi}(f^{\oplus k}(X, Y)\mid X_k, Y_{<k}, \bM, U))\right].
	\end{align*}
	By summing up the inequalities and applying Lemma~\ref{lem_adv_super_additivity}, we obtain
	\begin{align*}
		&\kern-2em\pi(S_{\good}\cap U)\cdot \E_{\pi\mid S_{\good}\cap U}\left[\adv_{\pi}(f^{\oplus k}(X, Y)\mid X_k, Y_{<k}, \bM, S_{\good}\cap U))\right] \\
		&<2^{-4}\cdot \pi(U)\cdot \E_{\pi\mid U}\left[\adv_{\pi}(f^{\oplus k}(X, Y)\mid X_k, Y_{<k}, \bM, U))\right],
	\end{align*}
	contracting with the lemma premise.

	\bigskip
	Now it suffices to prove the claim.
	By Proposition~\ref{prop_max_U}(i), Equation~\eqref{eqn_assumption_good_b_adv} implies that $\pi(G_{\good,b}\mid U)\geq 2^{-12}$.
	Therefore, the probability of $E_{\good,b}$ in the distribution $(\pi_{<k})_{G_{\good}}$ is at least $2^{-12}$ as required by Lemma~\ref{lem_induction}.
	In the following, we show that the bound on $\phi_{k-1}((\pi_{<k})_{G_{\good}}\mid E_{\good,b})=\phi_{k-1}((\pi_{<k})_{G_{\good}}\mid G_{\good,b})$ holds.

	\paragraph{Bounding $\phi_{k-1}^{\cost}((\pi_{<k})_{G_{\good}}\mid G_{\good,b})$.}
	We first bound its $\phi_{k-1}^{\cost}$ value.
	Similar to the proof of Lemma~\ref{lem_high_cost}, Lemma~\ref{lem_resolve_event} guarantees that
	\begin{align}
		&\kern-2em\phi_{k-1}^{\cost}((\pi_{<k})_{G_{\good}}\mid G_{\good,b}) \nonumber\\
		&\leq \phi_{k-1}^{\cost}(\pi_{<k}\mid S_{\good}\cap U)+O(r\log r)+3\log (1/\pi(G_{\good,b})) \nonumber
		\intertext{which by Equation~\eqref{eqn_cost_low_cost}, is}
		&\leq \phi_{k}^{\cost}(\pi\mid U)+O(r\log r)+3\log (1/\pi(G_{\good,b}))+3\log (1/\pi(S_{\good}\cap U))+3\cdot 2^{-5}\log(1/\alpha) \nonumber\\
		&\leq \phi_{k}^{\cost}(\pi\mid U)+6\log (1/\pi(G_{\good,b}))+2^{-3}\log(1/\alpha),\label{eqn_cost_good}
	\end{align}
	where we use the fact that $\alpha>r^{cr}$ for a sufficiently $c$, and $\pi(G_{\good,b})\leq \pi(S_{\good}\cap U)$.

	\bigskip
	To bound its $\phi_{k-1}^{\adv}$ and then $\phi_{k-1}$, we will consider two cases: $\pi(U)\leq \alpha^{1/8}$ and $\pi(U)> \alpha^{1/8}$.

	\paragraph{Bounding $\phi_{k-1}^{\adv}((\pi_{<k})_{G_{\good}}\mid G_{\good,b})$ when $\pi(U)\leq \alpha^{1/8}$.}
	We first bound its $\phi_{k-1}^{\adv}$ when $\pi(U)\leq \alpha^{1/8}$.
	Similar to the proof of Lemma~\ref{lem_high_cost}, Lemma~\ref{lem_resolve_event} implies that
	\[
		\pi(f^{\oplus k-1}(X_{<k}, Y_{<k})=b\mid X_k, Y_{<k}, \bM, G_{\good,b})\geq 1/2.
	\]
	Thus, the first part of Lemma~\ref{lem_cost_decs_adv} for $S=G_{\good,b}$ implies that
	\begin{align*}
		&\E_{\pi\mid G_{\good,b}}\left[\adv_\pi(f^{\oplus k-1}(X_{<k}, Y_{<k})\mid X_k, \bM, G_{\good,b})\right]\\
		&\qquad\geq \E_{\pi\mid G_{\good,b}}\left[\adv_\pi(f^{\oplus k}(X, Y)\mid X_k, Y_{<k}, \bM, G_{\good,b})\right].
	\end{align*}
	That is,
	\begin{equation}\label{eqn_adv_good}
		\phi_{k-1}^{\adv}((\pi_{<k})_{G_{\good}}\mid G_{\good,b})\leq \phi_{k,\pa}^{\adv}(\pi\mid G_{\good,b}).
	\end{equation}

	\paragraph{Bounding $\phi_{k-1}((\pi_{<k})_{G_{\good}}\mid G_{\good,b})$ when $\pi(U)\leq \alpha^{1/8}$.} By Equation~\eqref{eqn_cost_good} and~\eqref{eqn_adv_good}, we have
	\begin{align*}
		&\kern1.25em\phi_{k-1}((\pi_{<k})_{G_{\good}}\mid G_{\good,b}) \\
		&=\phi_{k-1}^{\cost}((\pi_{<k})_{G_{\good}}\mid G_{\good,b})+\phi_{k-1}^{\adv}((\pi_{<k})_{G_{\good}}\mid G_{\good,b}) \\
		&\leq \phi_{k}^{\cost}(\pi\mid U)+\phi_{k,\pa}^{\adv}(\pi\mid G_{\good,b})+6\log (1/\pi(G_{\good,b}))+2^{-3}\log(1/\alpha) \\
		&=\phi_{k, \pa}(\pi\mid U)-\phi_{k,\pa}^{\adv}(\pi\mid U)+\phi_{k,\pa}^{\adv}(\pi\mid G_{\good,b})+6\log (1/\pi(G_{\good,b}))+2^{-3}\log(1/\alpha)
		\intertext{which by Lemma~\ref{lem_max_U}, is}
		&\leq \phi_k(\pi\mid V)-(\phi_{k,\pa}^{\adv}(\pi\mid U)+6\log (1/\pi(U)))+(\phi_{k,\pa}^{\adv}(\pi\mid G_{\good,b})+6\log (1/\pi(G_{\good,b})))\\
		&\qquad+2^{-3}\log(1/\alpha)-7\log(1/\pi(U))+13\log (1/\pi(V))
		\intertext{which by the assumption that $\pi(U)\leq \alpha^{1/8}$ and $\pi(V)\geq 2^{-12}$, is}
		&\leq \phi_k(\pi\mid V)-(\phi_{k,\pa}^{\adv}(\pi\mid U)+6\log (1/\pi(U)))+(\phi_{k,\pa}^{\adv}(\pi\mid G_{\good,b})+6\log (1/\pi(G_{\good,b})))\\
		&\qquad-\frac{3}{4}\log(1/\alpha)+156.
	\end{align*}
	By Proposition~\ref{prop_max_U}(ii), Equation~\eqref{eqn_assumption_good_b_adv} implies that
	\[
		\phi_{k,\pa}^{\adv}(\pi\mid U)+6\log (1/\pi(U))\geq \phi_{k,\pa}^{\adv}(\pi\mid {G_{\good,b}})+6\log (1/\pi(G_{\good,b}))-192.
	\]
	Thus, $\phi_{k-1}((\pi_{<k})_{G_{\good}}\mid G_{\good,b})\leq \phi_k(\pi\mid V)-\frac{1}{4}\log(1/\alpha)$, as $\alpha<r^{cr}$ for a large $c$.
	This proves Claim~\ref{cl_low_cost} when $\pi(U)\leq \alpha^{1/8}$.

	\paragraph{Bounding $\phi_{k-1}^{\adv}((\pi_{<k})_{G_{\good}}\mid G_{\good,b})$ when $\pi(U)>\alpha^{1/8}$.}
	Next we consider the case where $\pi(U)>\alpha^{1/8}$.
	To bound $\phi_{k-1}^{\adv}((\pi_{<k})_{G_{\good}}\mid G_{\good,b})$ in this case, we will apply the second part of Lemma~\ref{lem_cost_decs_adv}.
	To this end, we will first upper bound the advantage of $\pi_k$ for computing $f(X_k, Y_k)$ by applying Lemma~\ref{lem_compression} to $\rho=\pi_k$ and $W=G_{\good,b}$ and using the assumption on the communication complexity of $f$.

	To verify the premises of Lemma~\ref{lem_compression} are satisfied, note that $G_{\good, b}$ is in $\cS_{\pa}(\pi)$, hence, $(\pi\mid G_{\good,b})$ has the partial rectangle property with respect to $\mu^{k}$ by Proposition~\ref{prop_cs_partial}(ii).
	Hence, $(\pi_k\mid G_{\good, b})$ has the rectangle property with respect to $\mu$ by Proposition~\ref{prop_pi<k_pk_rect}.
	To bound the costs of $\pi_k$ conditioned on $G_{\good,b}$, note that since $S_{\good}\cap U\in\cS_{\pa}(\pi)$, we have $Y_k$ and $X_{<k}$ are independent conditioned on $(X_k, Y_{<k}, \bM, S_{\good}\cap U)$ by Proposition~\ref{prop_independent_x<k_yk}.
	Note that $G_{\good, b}\subseteq S_{\good}\cap U$, and note that whether the event $G_{\good, b}$ happens is determined by $(X_{<k}, Y_{<k}, (X_k, \bM), S_{\good}\cap U)$, and whether $S_{\good}$ happens is determined by $(X_k, Y_{<k}, \bM, U)$.
	Therefore, when $\pi(X, Y_{<k}, \bM, G_{\good, b})>0$, the distribution of $Y_k$ conditioned on $(X, Y_{<k}, \bM, G_{\good, b})$ is the same as the distribution of $Y_k$ conditioned on $(X_k, Y_{<k}, \bM, U)$, because
	\begin{align*}
		&\kern1.325em\pi(Y_k=y_k\mid X, Y_{<k}, \bM, G_{\good,b}) \\
		&=\pi(Y_k=y_k\mid X, Y_{<k}, \bM, S_{\good}\cap U, G_{\good,b}) && \textrm{($G_{\good,b}\subseteq S_{\good}\cap U$)} \\
		&=\pi(Y_k=y_k\mid X, Y_{<k}, \bM, S_{\good}\cap U) && \textrm{($G_{\good,b}$ implied by $(X, Y_{<k}, \bM, S_{\good}\cap U)$)}\\
		&=\pi(Y_k=y_k\mid X_k, Y_{<k}, \bM, S_{\good}\cap U) && \textrm{($X_{<k}\bot Y_k$ given $(X_k, Y_{<k}, \bM, S_{\good}\cap U)$)} \\
		&=\pi(Y_k=y_k\mid X_k, Y_{<k}, \bM, U, S_{\good}\cap U) && \textrm{($S_{\good}\cap U\subseteq U$)} \\
		&=\pi(Y_k=y_k\mid X_k, Y_{<k}, \bM, U) && \textrm{($S_{\good}$ implied by $(X_k, Y_{<k}, \bM)$)}.
	\end{align*}
	Thus, the $\theta$-cost of $\pi_k$ conditioned on $G_{\good,b}$ is at most
	\begin{align*}
		&\kern-2em\thec{\pi_k\mid G_{\good, b}}{\mu} \\
		&=\E_{(X, Y, \bM)\sim\pi\mid G_{\good, b}}\left[\thec{\pi_k}{\mu}[X, Y, \bM]\right] \\
		&=\E_{(X, Y_{<k}, \bM)\sim\pi\mid G_{\good, b}}\left[\E_{Y_k\sim\pi\mid (X, Y_{<k}, \bM, G_{\good, b})}\left[\thec{\pi_k}{\mu}[X, Y, \bM]\right]\right] \\
		&=\E_{(X, Y_{<k}, \bM)\sim\pi\mid G_{\good, b}}\left[\E_{Y_k\sim\pi\mid (X_k, Y_{<k}, \bM, U)}\left[\thec{\pi_k}{\mu}[X, Y, \bM]\right]\right] \\
		&<\E_{(X, Y_{<k}, \bM)\sim\pi\mid G_{\good, b}}\left[\alpha^{-1/2}\right] \\
		&=\alpha^{-1/2}.
	\end{align*}
	For the same reason, its $\chisq$-cost by Alice is at most $2^{C_A-c \log (r/\alpha)}$, and its $\chisq$-cost by Bob is at most $2^{C_B-c \log (r/\alpha)}$.
	By applying Lemma~\ref{lem_compression} to $\rho=\pi_k$ and event $W=G_{\good, b}$ for $\delta_1=\alpha^{1/4}$ and $\delta_2=\alpha\cdot r^{-1}$, we obtain a \emph{standard} $r$-round protocol $\tau$.
	Since $c$ is a sufficiently large constant, we have that in $\tau$,
	\begin{itemize}
		\item Alice sends at most \[
			C_A-c\cdot \log (r/\alpha)+O(\log (r/\delta_1\delta_2)+\log\log \thec{\pi_k\mid G_{\good, b}}{\mu})\leq C_A
		\] bits in every odd round;
		\item Bob sends at most \[
			C_B-c\cdot \log (r/\alpha)+O(\log (r/\delta_1\delta_2))\leq C_B
		\] bits in every even round;
		\item $\tau$ computes $f$ correctly under input distribution $\mu$ with probability at least
		\[
			\frac{1}{2}+\frac{\delta_1}{32\thec{\pi_k\mid G_{\good, b}}{\mu}}\cdot \left(\pi_k(G_{\good, b})\cdot \E_{\pi_k\mid G_{\good, b}}\left[\adv_{\pi_k}(f(X_k, Y_k)\mid X_k, \bM^{(\pi_k)}, G_{\good, b})\right]-6\delta_1\right)-2r\delta_2.
		\]
	\end{itemize}

	By the our assumption on the communication complexity of $f$, the expected advantage of $\tau$ must be at most $\alpha$:
	\[
		\frac{\delta_1}{16\thec{\pi_k\mid G_{\good, b}}{\mu}}\cdot \left(\pi_k(G_{\good, b})\cdot \E_{\pi_k\mid G_{\good, b}}\left[\adv_{\pi_k}(f(X_k, Y_k)\mid X_k, \bM^{(\pi_k)}, G_{\good, b})\right]-6\delta_1\right)-4r\delta_2\leq \alpha.
	\]
	It implies that 
	\begin{align*}
		&\kern-2em\pi_k(G_{\good, b})\cdot \E_{\pi_k\mid G_{\good, b}}\left[\adv_{\pi_k}(f(X_k, Y_k)\mid X_k, Y_{<k}, \bM, G_{\good, b})\right] \\
		&\leq (\alpha+4r\delta_2)\cdot \frac{16\thec{\pi_k\mid G_{\good, b}}{\mu}}{\delta_1}+6\delta_1\\
		&\leq (\alpha+4\alpha)\cdot \frac{16\alpha^{-1/2}}{\alpha^{1/4}}+6\alpha^{1/4} \\
		&=86\alpha^{1/4}.
	\end{align*}

	Since $G_{\good,b}\subseteq U$, we apply the second part of Lemma~\ref{lem_cost_decs_adv} for $S=G_{\good, b}$ and $$\eta=86\alpha^{1/4}\cdot \pi(G_{\good, b})^{-1}.$$
	The premises of Lemma~\ref{lem_cost_decs_adv} are satisfied, because
	\begin{enumerate}[(a)]
		\item Lemma~\ref{lem_resolve_event} gives that 
		\[
			\pi(f^{\oplus k-1}(X_{<k}, Y_{<k})=b\mid X_k, Y_{<k}, \bM, G_{\good,b})\geq 1/2;
		\]
		\item by Proposition~\ref{prop_max_U}(i), Equation~\eqref{eqn_assumption_good_b_adv} implies that $\pi(G_{\good,b}\mid U)^{1/2}\geq 2^{-6}$, hence, $$\pi(G_{\good,b})\geq 2^{-12}\cdot \pi(U)\geq 2^{-12}\cdot\alpha^{1/8};$$
		\item then we have
			\begin{align*}
				\eta^{1/4}&=\left(86\alpha^{1/4}\cdot \pi(G_{\good, b})^{-1}\right)^{1/4} \\
				&\leq \left(2^{20}\alpha^{1/8}\right)^{1/4} \\
				&=2^{5}\alpha^{1/32};
			\end{align*}
		\item Equation~\eqref{eqn_assumption_good_b_adv} implies
		\begin{align*}
			&\kern1.25em\frac{1}{2}\cdot \frac{\pi(G_{\good, b})^{1/2}\cdot \E_{\pi\mid G_{\good, b}}\left[\adv_{\pi}(f^{\oplus k}(X, Y)\mid X_k, Y_{<k}, \bM, G_{\good, b})\right]}{\pi(U)^{1/2}\cdot \E_{\pi\mid U}\left[\adv_{\pi}(f^{\oplus k}(X, Y)\mid X_k, Y_{<k}, \bM, U)\right]} \\
			&= \frac{1}{2\cdot \pi(G_{\good,b}\mid U)^{1/2}}\cdot \frac{\pi(G_{\good, b})\cdot \E_{\pi\mid G_{\good, b}}\left[\adv_{\pi}(f^{\oplus k}(X, Y)\mid X_k, Y_{<k}, \bM, G_{\good, b})\right]}{\pi(U)\cdot \E_{\pi\mid U}\left[\adv_{\pi}(f^{\oplus k}(X, Y)\mid X_k, Y_{<k}, \bM, U)\right]} \\
			&\geq \frac{1}{2\cdot \pi(G_{\good,b}\mid U)^{1/2}}\cdot 2^{-6}\\
			&\geq 2^{-7},
		\end{align*}
		which is at least $\eta^{1/4}$ by the upper bound on $\eta^{1/4}$ in (c) and the fact that $\alpha$ is sufficiently small.
	\end{enumerate}

	Hence, we obtain the following by the second part of Lemma~\ref{lem_cost_decs_adv},
	\begin{align*}
		&\kern1.5em\E_{\pi\mid G_{\good, b}}\left[\adv_{\pi}(f^{\oplus k-1}(X_{<k}, Y_{<k})\mid \bM^{(\pi_{<k})}, G_{\good, b})\right] \\
		&\geq \frac{1}{2}\cdot \eta^{-1/2}\cdot \E_{\pi\mid G_{\good, b}}\left[\adv_{\pi}(f^{\oplus k}(X, Y)\mid X_k, Y_{<k}, \bM, G_{\good, b})\right]
		\intertext{which by the above upper bound on $\eta^{1/4}$ in (c), is}
		&\geq 2^{-11}\cdot \alpha^{-1/16}\cdot \E_{\pi\mid G_{\good, b}}\left[\adv_{\pi}(f^{\oplus k}(X, Y)\mid X_k, Y_{<k}, \bM, G_{\good, b})\right]
	\end{align*}
	i.e.,
	\begin{equation}\label{eqn_adv_good_2}
		\phi_{k-1}^{\adv}((\pi_{<k})_{G_{\good}}\mid G_{\good,b})\leq \phi_{k,\pa}^{\adv}(\pi\mid G_{\good,b})-2\log(1/\alpha)+352.
	\end{equation}
	\paragraph{Bounding $\phi_{k-1}((\pi_{<k})_{G_{\good}}\mid G_{\good,b})$ when $\pi(U)>\alpha^{1/8}$.}
	Combining Equation~\eqref{eqn_cost_good} and~\eqref{eqn_adv_good_2}, we have
	\begin{align*}
		&\kern1.25em\phi_{k-1}((\pi_{<k})_{G_{\good}}\mid G_{\good,b}) \\
		&=\phi_{k-1}^{\cost}((\pi_{<k})_{G_{\good}}\mid G_{\good,b})+\phi_{k-1}^{\adv}((\pi_{<k})_{G_{\good}}\mid G_{\good,b}) \\
		&\leq \phi_{k}^{\cost}(\pi\mid U)+\phi_{k,\pa}^{\adv}(\pi\mid G_{\good,b})+6\log (1/\pi(G_{\good,b}))+2^{-3}\log(1/\alpha)-2\log(1/\alpha)+352 \\
		&=\phi_{k,\pa}(\pi \mid U)-\phi_{k,\pa}^{\adv}(\pi\mid U)+\phi_{k,\pa}^{\adv}(\pi\mid G_{\good,b})+6\log (1/\pi(G_{\good,b}))-\frac{15}{8}\log(1/\alpha)+352
		\intertext{which by Lemma~\ref{lem_max_U} and the fact that $\log(1/\pi(U\mid V))\geq 0$, is }
		&\leq \phi_k(\pi \mid V)-(\phi_{k,\pa}^{\adv}(\pi\mid U)+6\log (1/\pi(U)))+(\phi_{k,\pa}^{\adv}(\pi\mid G_{\good,b})+6\log (1/\pi(G_{\good,b}))) \\
		&\qquad-\frac{15}{8}\log(1/\alpha)+352+6\log (1/\pi(V)) \\
		\intertext{which by Proposition~\ref{prop_max_U}(ii) and the fact that $\pi(V)\geq 2^{-12}$ and $\alpha$ is sufficiently small, is}
		&\leq \phi_k(\pi \mid V)-\log(1/\alpha).
	\end{align*}
	This proves Claim~\ref{cl_low_cost} when $\pi(U)>\alpha^{1/8}$, and completes the proof of Lemma~\ref{lem_low_cost}.
\end{proof}


\subsection{Putting together}\label{sec_put_together_induction}
Now we are ready to prove Lemma~\ref{lem_induction}.
The two main lemmas in the previous two subsections show that if either $S_{\reducecost}$ or $S_{\good}$ contributes a nontrivial advantage in $U$, then Lemma~\ref{lem_induction} holds.
We will show that the complement of their union has very low probability, hence contributes a small amount of advantage by Proposition~\ref{prop_max_U}(i).
Then the superadditivity of weighted advantage implies the lemma.
\begin{restate}[Lemma~\ref{lem_induction}]
	\leminductioncont
\end{restate}
\begin{proof}
	If the premise of Lemma~\ref{lem_high_cost} or Lemma~\ref{lem_low_cost} holds, then Lemma~\ref{lem_induction} holds.

	Otherwise, we have that
	\begin{align}
		&\kern-2em\pi(S_{\reducecost}\cap U)\cdot \E_{\pi\mid S_{\reducecost}\cap U}\left[\adv_{\pi}(f^{\oplus k}(X, Y)\mid X_k, Y_{<k}, \bM, S_{\reducecost}\cap U)\right]\nonumber\\
		&<\frac{1}{3}\cdot \pi(U)\cdot \E_{\pi\mid U}\left[\adv_{\pi}(f^{\oplus k}(X, Y)\mid X_k, Y_{<k}, \bM, U))\right],\label{eqn_induction_high_cost}
	\end{align}
	and
	\begin{align}
		&\kern-2em\pi(S_{\good}\cap U)\cdot \E_{\pi\mid S_{\good}\cap U}\left[\adv_{\pi}(f^{\oplus k}(X, Y)\mid X_k, Y_{<k}, \bM, S_{\good}\cap U)\right] \nonumber\\
		&<\frac{1}{3}\cdot \pi(U)\cdot \E_{\pi\mid U}\left[\adv_{\pi}(f^{\oplus k}(X, Y)\mid X_k, Y_{<k}, \bM, U))\right].\label{eqn_induction_low_cost}
	\end{align}

	On the other hand, by construction, the complement of $S_{\reducecost}\cup S_{\good}$ is the set of all triples $(X_k, Y_{<k}, \bM)$ such that either
	\begin{align*}
		\E_{Y_k\sim \pi\mid X_k, Y_{<k}, \bM, U}\left[\thec{\pi_k}{\mu}[X_k, Y_k, \bM^{(\pi_k)}]\right]&<\alpha^{2^{-5}}\cdot \pi(U), \textrm{or}\\
		\E_{Y_k\sim \pi\mid X_k, Y_{<k}, \bM, U}\left[\chis{\pi_k}{\mu,A}[X_k, Y_k, \bM^{(\pi_k)}]\right]&<2^{-2^{-5}(C_A-c\log(r/\alpha))}\cdot \pi(U), \textrm{or}\\
		\E_{Y_k\sim \pi\mid X_k, Y_{<k}, \bM, U}\left[\chis{\pi_k}{\mu,B}[X_k, Y_k, \bM^{(\pi_k)}]\right]&<2^{-2^{-5}(C_B-c\log(r/\alpha))}\cdot \pi(U).
	\end{align*}
	Denote this set by $S_{\lowprob}$.
	Clearly, we also have $S_{\lowprob}\cap U\in \cS_{\pa}(\pi)$.

	However, by Proposition~\ref{prop_expect_theta_inv}, we have
	\[
		\E_{\pi\mid U}\left[\thec{\pi_k}{\mu}[X_k, Y_k, \bM^{(\pi_k)}]^{-1}\right]\leq \pi(U)^{-1}\cdot \E_{\pi}\left[\thec{\pi_k}{\mu}[X_k, Y_k, \bM^{(\pi_k)}]^{-1}\right]=\pi(U)^{-1}.
	\]
	Since $\thec{\pi_k}{\mu}[X_k, Y_k, \bM^{(\pi_k)}]$ is a function of $(X_k, Y, \bM)$, by the convexity of $x^{-1}$, we also have
	\[
		\left(\E_{Y_k\sim \pi\mid X_k, Y_{<k}, \bM, U}\left[\thec{\pi_k}{\mu}[X_k, Y_k, \bM^{(\pi_k)}]\right]\right)^{-1}\leq \E_{Y_k\sim \pi\mid X_k, Y_{<k}, \bM, U}\left[\thec{\pi_k}{\mu}[X_k, Y_k, \bM^{(\pi_k)}]^{-1}\right],
	\]
	and hence,
	\[
		\E_{(X_k, Y_{<k}, \bM)\sim\pi\mid U}\left[\left(\E_{Y_k\sim \pi\mid X_k, Y_{<k}, \bM, U}\left[\thec{\pi_k}{\mu}[X_k, Y_k, \bM^{(\pi_k)}]\right]\right)^{-1}\right]\leq \pi(U)^{-1}.
	\]
	By Markov's inequality, we obtain
	\[
		\Pr_{(X_k, Y_{<k}, \bM)\sim\pi\mid U}\left[\E_{Y_k\sim \pi\mid X_k, Y_{<k}, \bM, U}\left[\thec{\pi_k}{\mu}[X_k, Y_k, \bM^{(\pi_k)}]\right]<\alpha^{2^{-5}}\cdot \pi(U)\right]<{\alpha^{2^{-5}}}.
	\]
	Similarly, by invoking Proposition~\ref{prop_expect_chis_inv}, we have
	\[
		\Pr_{(X_k, Y_{<k}, \bM)\sim\pi\mid U}\left[\E_{Y_k\sim \pi\mid X_k, Y_{<k}, \bM, U}\left[\chis{\pi_k}{\mu,A}[X_k, Y_k, \bM^{(\pi_k)}]\right]<2^{-2^{-5}(C_A-c\log(r/\alpha))}\cdot \pi(U)\right]<2^{-2^{-5}(C_A-c\log(r/\alpha))},
	\]
	and
	\[
		\Pr_{(X_k, Y_{<k}, \bM)\sim\pi\mid U}\left[\E_{Y_k\sim \pi\mid X_k, Y_{<k}, \bM, U}\left[\chis{\pi_k}{\mu,B}[X_k, Y_k, \bM^{(\pi_k)}]\right]<2^{-2^{-5}(C_B-c\log(r/\alpha))}\cdot \pi(U)\right]<2^{-2^{-5}(C_B-c\log(r/\alpha))}.
	\]
	Thus, by union bound, we have
	\[
		\pi(S_{\lowprob}\mid U)<{\alpha^{2^{-5}}}+2^{-2^{-5}(C_A-c\log(r/\alpha))}+2^{-2^{-5}(C_B-c\log(r/\alpha))}<1/9,
	\]
	since $\alpha<r^{cr}$, $C_A,C_B>2c\log(r/\alpha)$ for a sufficiently large $c$.
	By applying Proposition~\ref{prop_max_U}(i) on $S_{\lowprob}\cap U$, we obtain
	\begin{align}
		&\kern-2em\pi(S_{\lowprob}\cap U)\cdot \E_{\pi\mid S_{\lowprob}\cap U}\left[\adv_{\pi}(f^{\oplus k}(X, Y)\mid X_k, Y_{<k}, \bM, S_{\lowprob}\cap U)\right]\nonumber\\
		&<\frac{1}{3}\cdot \pi(U)\cdot \E_{\pi\mid U}\left[\adv_{\pi}(f^{\oplus k}(X, Y)\mid X_k, Y_{<k}, \bM, U))\right].\label{eqn_induction_low_prob}
	\end{align}
	By summing up Equation~\eqref{eqn_induction_high_cost},~\eqref{eqn_induction_low_cost} and~\eqref{eqn_induction_low_prob}, we get a contradiction with Lemma~\ref{lem_adv_super_additivity}.
	This completes the proof of Lemma~\ref{lem_induction}.
\end{proof}


\subsection{Proof of Lemma~\ref{lem_resolve_event}}\label{sec_resolve_event}
In this subsection, we prove Lemma~\ref{lem_resolve_event}, which lets us convert a protocol conditioned on an event to a generalized protocol with bounded costs.
\begin{restate}[Lemma~\ref{lem_resolve_event}]
	\lemresolveeventcont
\end{restate}
\begin{proof}
	Ideally, we could simply let $G_0$ be the intersection of $W$ and all $(Y, \bM)$ such that $\rho(h(X, Y)=0\mid Y, \bM, W)\geq 1/2$, and $G_1$ be the intersection of $W$ and all other $(Y, \bM)$.
	In this way, the last line of the lemma holds, since $G_b$ is a set that depends only on $(Y, \bM, W)$ and is a subset of $W$,
	\[
		\rho(h(X, Y)=b\mid Y, \bM, G_b)=\rho(h(X, Y)=b\mid Y, \bM, W, G_b)=\rho(h(X, Y)=b\mid Y, \bM, W).
	\]
	However, the new protocol $\rho_G$ (for $G=W$ in this case) may not have low costs, since the denominators in the definitions may become arbitrarily small (recall Definition~\ref{def_theta_cost} and Definition~\ref{def_chi_cost}).
	To ensure that the costs of the new protocol $\rho_{G}$ are bounded, we will identify all $(X, \bM)$ and $(Y, \bM)$ at which the denominators in the definition of $\theta$-cost and $\chisq$-costs becomes much smaller, and repeatedly remove such pairs from the support.

	More specifically, we repeatedly remove from the support of $\rho$, all $M_0$ whose probability becomes much smaller after conditioning on $G$, we also remove pairs $(X, \bM)$ [resp. $(Y, \bM)$] such that either $\rho(X, \bM)$ or for some odd $i$, $\rho(M_i\mid X, M_{<i})$ [resp. $\rho(Y, \bM)$ or for some even $i$, $\rho(M_i\mid Y, M_{<i})$] becomes much smaller after conditioning.

	Formally, consider the following process:\footnote{Note that $W$ may not necessarily be a subset of $\cX\times \cY\times \cM$, it could be any event.}
	\begin{compactenum}
		\item $W_{\tmp}\leftarrow W$ \qquad // the current $W$
		\item $B_{X,M}\leftarrow\varnothing$ \qquad // the bad $(X, \bM)$ pairs
		\item $B_{Y,M}\leftarrow\varnothing$ \qquad // the bad $(Y, \bM)$ pairs
		\item repeat
		\item \qquad if $\exists m_0$ such that $0<\rho(m_0\mid W_{\tmp})<\frac{\gamma}{r+3}\cdot \rho(m_0)$
		\item \qquad\qquad $B_{X,M}\leftarrow B_{X,M}\cup (\cX\times \{\bM: M_0=m_0\})$
		\item \qquad\qquad $W_{\tmp}\leftarrow W_{\tmp}\setminus (\cX\times \cY\times \{\bM: M_0=m_0\})$
		\item \qquad if $\exists x, \bm$ such that $0<\rho(x, \bm\mid W_{\tmp})<\frac{\gamma}{r+3}\cdot \rho(x, \bm)$
		\item \qquad\qquad $B_{X,M}\leftarrow B_{X,M}\cup \left\{(x, \bm)\right\}$
		\item \qquad\qquad $W_{\tmp}\leftarrow W_{\tmp}\setminus (\{x\}\times \cY\times \{\bm\})$
		\item \qquad if $\exists y, \bm$ such that $0<\rho(y, \bm\mid W_{\tmp})<\frac{\gamma}{r+3}\cdot \rho(y, \bm)$
		\item \qquad\qquad $B_{Y,M}\leftarrow B_{Y,M}\cup \left\{(y, \bm)\right\}$
		\item \qquad\qquad $W_{\tmp}\leftarrow W_{\tmp}\setminus (\cX\times \{y\} \times \{\bm\})$
		\item \qquad if $\exists x, \odd i\in[r], m_{\leq i}$ such that $0<\rho(m_i\mid x, m_{<i}, W_{\tmp})<\frac{\gamma}{r+3}\cdot\rho(m_i\mid x, m_{<i})$
		\item \qquad\qquad $B_{X,M}\leftarrow B_{X,M}\cup (\{x\}\times \left\{\bM: M_{\leq i}=m_{\leq i}\right\})$
		\item \qquad\qquad $W_{\tmp}\leftarrow W_{\tmp}\setminus (\{x\}\times \cY\times \left\{\bM: M_{\leq i}=m_{\leq i}\right\})$
		\item \qquad if $\exists y, \even i\in[r], m_{\leq i}$ such that $0<\rho(m_i\mid y, m_{<i}, W_{\tmp})<\frac{\gamma}{r+3}\cdot\rho(m_i\mid y, m_{<i})$
		\item \qquad\qquad $B_{Y,M}\leftarrow B_{Y,M}\cup (\{y\}\times \left\{\bM: M_{\leq i}=m_{\leq i}\right\})$
		\item \qquad\qquad $W_{\tmp}\leftarrow W_{\tmp}\setminus (\cX\times \{y\}\times \left\{\bM: M_{\leq i}=m_{\leq i}\right\})$
		\item until none of line 5,8,11,14,17 holds
		\item $G\leftarrow W_{\tmp}$
		\item $E_0\leftarrow\{(Y, \bM): \rho(Y, \bM\mid G)=0 \vee \rho(h=0\mid Y, \bM, G)\geq 1/2\}$
		\item $E_1\leftarrow\{(Y, \bM): \rho(Y, \bM\mid G)>0 \wedge \rho(h=1\mid Y, \bM, G)> 1/2\}$
		\item $G_0\leftarrow G\cap E_0$
		\item $G_1\leftarrow G\cap E_1$
		\item $B_0\leftarrow W\cap \{(X, Y, \bM): (X, \bM)\in B_{X, M}\}$
		\item $B_1\leftarrow W\cap \{(X, Y, \bM): (X, \bM)\notin B_{X, M}, (Y, \bM)\in B_{Y, M}\}$
		\item return $(G, B_0, B_1, E_0, E_1)$
	\end{compactenum}

	By construction, $(G, B_0, B_1)$ is a partition of $W$ and $(E_0, E_1)$ is a partition of $\cY\times \cM$.
	To see that Item 1 holds, note that the set $B_{X, M}$ and its complement are in $\cU_{X, M}$, the set $B_{Y, M}$ is in $\cU_{Y, M}$ (recall Definition~\ref{def_cU_pi}), and note that $G$ is also the set $W\cap \{(X, Y, \bM): (X, \bM)\notin B_{X, M}, (Y, \bM)\notin B_{Y, M}\}$.
	Thus, $G, B_0, B_1$ have the form $W\cap S$ for some $S\in\cS_{\rect}$.

	Since for $b=0,1$, for all $(Y, \bM)$ such that $\rho(Y, \bM\mid G_b)>0$, we have
	\[
		\rho(h(X, Y)=b\mid Y, \bM, G_b)=\rho(h(X, Y)=b\mid Y, \bM, G, G_b)=\rho(h(X, Y)=b\mid Y, \bM, G)\geq 1/2.
	\]
	Item 4 also holds.

	It remains to bound $\rho(B_0\cup B_1)$, and bound the costs of $\rho_G$ conditioned on $G_b$.

	\begin{claim}\label{cl_prob_drop}
		We have $\rho(B_0\cup B_1)\leq \gamma\cdot\rho(W)$.
	\end{claim}
	To see this, first observe that by construction, $B_0\cup B_1$ contains all triples that are ``removed from $W$'' in the whole process.
	Let us first focus on step 5-7, and upper bound the total probability of all $m_0$ that are removed in step 7.
	Observe that each $m_0$ can only be removed at most once.
	Each time we remove a $m_0$, $\rho(W_{\tmp})$ decreases by $\rho(m_0, W_{\tmp})$ (for the $W_{\tmp}$ at the time of the removal).
	Since $\rho(m_0\mid W_{\tmp})<\frac{\gamma}{r+3}\cdot \rho(m_0)$ at the time of the removal, we have
	\begin{align*}
		\rho(m_0, W_{\tmp})&=\rho(W_{\tmp})\cdot \rho(m_0\mid W_{\tmp})\\
		&\leq \rho(W)\cdot \left(\frac{\gamma}{r+3}\cdot \rho(m_0)\right).
	\end{align*}
	Therefore, during the entire process, $\rho(W_{\tmp})$ can decrease in step 7 by at most
	\begin{align*}
		\sum_{m_0}\rho(W)\cdot \left(\frac{\gamma}{r+3}\cdot \rho(m_0)\right)=\frac{\gamma}{r+3}\cdot \rho(W).
	\end{align*}
	Similarly, in step 10 and step 13, $\rho(W_{\tmp})$ can also decrease by at most $\frac{\gamma}{r+3}\cdot \rho(W)$ respectively.

	Next, consider step 14-16, and fix an odd $i$.
	Each time we remove a pair $(x, m_{\leq i})$, $\rho(W_{\tmp})$ decreases by $\rho(x, m_{\leq i}, W_{\tmp})$.
	Since $\rho(m_i\mid x, m_{<i}, W_{\tmp})<\frac{\gamma}{r+3}\cdot\rho(m_i\mid x, m_{<i})$, we have
	\begin{align*}
		\rho(x, m_{\leq i}, W_{\tmp})&=\rho(x, m_{<i}, W_{\tmp})\cdot \rho(m_i\mid x, m_{<i}, W_{\tmp}) \\
		&\leq\rho(x, m_{<i}, W)\cdot \left(\frac{\gamma}{r+3}\cdot\rho(m_i\mid x, m_{<i})\right).
	\end{align*}
	Therefore, during the whole process, the probability $\rho(W_{\tmp})$ can decrease in step 16 for a fixed $i$ by at most 
	\begin{align*}
		&\kern-2em\sum_{(x, m_{\leq i})}\rho(x, m_{<i}, W)\cdot \left(\frac{\gamma}{r+3}\cdot\rho(m_i\mid x, m_{<i})\right) \\
		&=\frac{\gamma}{r+3}\cdot \sum_{(x, m_{<i})}\rho(x, m_{<i}, W) \\
		&=\frac{\gamma}{r+3}\cdot \rho(W).
	\end{align*}
	Similarly, $\rho(W_{\tmp})$ can decrease in step 19 for a fixed $i$ by at most $\frac{\gamma}{r+3}\cdot \rho(W)$.

	Hence, summing over all $i$ for step 16 and 19 and over all steps, $\rho(W_{\tmp})$ decreases by at most $\gamma\cdot \rho(W)$, i.e., $\rho(B_0\cup B_1)\leq \gamma\cdot\rho(W)$, proving Claim~\ref{cl_prob_drop}.
	Equivalently, $\rho(B_0\cup B_1\mid W)\leq \gamma$, and $\rho(G\mid W)\geq 1-\gamma$.
	Hence, Item 2 holds.

	\bigskip

	Finally, it remains to bound the costs of $\rho_G$ conditioned on $G_b$ for Item 3.
	Since $G$ is the final $W_{\tmp}$, which passes line 20, $\rho_G$ must satisfy
	\begin{equation}\label{eqn_m0}
		\rho_G(M_0)\geq \frac{\gamma}{r+3}\cdot\rho(M_0)
	\end{equation}
	for $M_0$ with $\rho_G(M_0)>0$,
	\begin{equation}\label{eqn_xm2}
		\rho_G(X, \bM)\geq \frac{\gamma}{r+3}\cdot\rho(X, \bM)
	\end{equation}
	for $(X, \bM)$ with $\rho_G(X, \bM)>0$,
	\begin{equation}\label{eqn_ym2}
		\rho_G(Y, \bM)\geq \frac{\gamma}{r+3}\cdot\rho(Y, \bM)
	\end{equation}
	for $(Y, \bM)$ with $\rho_G(Y, \bM)>0$,
	\begin{equation}\label{eqn_xm}
		\rho_G(M_i\mid X, M_{<i})\geq \frac{\gamma}{r+3}\cdot\rho(M_i\mid X, M_{<i})
	\end{equation}
	for all odd $i\in[r]$ and $(X, M_{\leq i})$ with $\rho_G(X, M_{\leq i})>0$, and
	\begin{equation}\label{eqn_ym}
		\rho_G(M_i\mid Y, M_{<i})\geq \frac{\gamma}{r+3}\cdot\rho(M_i\mid Y, M_{<i})
	\end{equation}
	for all even $i\in[r]$ and $(Y, M_{\leq i})$ with $\rho_G(Y, M_{\leq i})>0$.

	The $\theta$-cost of $\rho_G$ conditioned on $G_b$ is
	\begin{align*}
		&\kern1.25em\log \thec{\rho_G\mid G_b}{\nu}\\
		&=\log \E_{(X, Y, \bM)\sim \rho_G\mid G_b}\left[\frac{\rho_G(X, Y, \bM)}{\rho_G(M_0)\cdot\nu(X, Y)\cdot\prod_{\odd i\in[r]}\rho_G(M_i\mid X, M_{<i})\cdot\prod_{\even i\in[r]}\rho_G(M_i\mid Y, M_{<i})}\right] \\
		\intertext{which by~\eqref{eqn_m0},~\eqref{eqn_xm} and~\eqref{eqn_ym}, is}
		&\leq\log \E_{(X, Y, \bM)\sim \rho\mid G_b}\left[\left(\frac{r+3}{\gamma}\right)^{r+1}\cdot\frac{\rho(X, Y, \bM\mid G)}{\rho(M_0)\cdot \nu(X, Y)\cdot\prod_{\odd i\in[r]}\rho(M_i\mid X, M_{<i})\cdot\prod_{\even i\in[r]}\rho(M_i\mid Y, M_{<i})}\right] \\
		&\leq\log \E_{(X, Y, \bM)\sim \rho\mid G_b}\left[\left(\frac{r+3}{\gamma}\right)^{r+1}\cdot\frac{\rho(X, Y, \bM)/\rho(G)}{\rho(M_0)\cdot \nu(X, Y)\cdot\prod_{\odd i\in[r]}\rho(M_i\mid X, M_{<i})\cdot\prod_{\even i\in[r]}\rho(M_i\mid Y, M_{<i})}\right] \\
		&=\log \E_{(X, Y, \bM)\sim \rho\mid G_b}\left[\thec{\rho}{\nu}[X, Y, \bM]\right]+(r+1)\log \left((r+3)/\gamma\right)+\log \left(1/\rho(G)\right) \\
		&=\log \thec{\rho\mid G_b}{\nu}+(r+1)\log \left((r+3)/\gamma\right)+\log \left(1/\rho(G)\right)
		\intertext{which by Proposition~\ref{prop_thec_condition_increase} and the fact that $G_b\subseteq W$, is}
		&\leq \log \thec{\rho\mid W}{\nu}+\log (1/\rho(G_b\mid W))+(r+1)\log \left((r+3)/\gamma\right)+\log \left(1/\rho(G)\right) \\
		&=\log \thec{\rho\mid W}{\nu}+(r+1)\log \left((r+3)/\gamma\right)+\log (1/\rho(G_b))+\log (1/\rho(G\mid W)) \\
		&\leq \log \thec{\rho\mid W}{\nu}+(r+1)\log \left((r+3)/\gamma\right)+\log (1/((1-\gamma)\rho(G_b))).
	\end{align*}

	For the $\chisq$-cost by Alice, we have
	\begin{align*}
		\log \chis{\rho_G\mid G_b}{\nu,A}&=\log \E_{(X, Y, \bM)\sim \rho_G\mid G_b}\left[\frac{\rho_G(X\mid \bM,Y)}{\nu(X\mid Y)}\right] \\
		&=\log \E_{(X, Y, \bM)\sim \rho\mid G_b} \left[\frac{\rho(X,Y,\bM\mid G)}{\rho_G(\bM,Y)\nu(X\mid Y)}\right]
		\intertext{which by~\eqref{eqn_ym2}, is}
		&\leq \log \E_{(X, Y, \bM)\sim \rho\mid G_b} \left[\frac{\rho(X, Y, \bM)/\rho(G)}{(\gamma/(r+3))\rho(\bM,Y)\nu(X\mid Y)}\right] \\
		&=\log \chis{\rho\mid G_b}{\nu,A}+\log ((r+3)/\gamma)+\log (1/\rho(G))
		\intertext{which by Proposition~\ref{prop_chis_condition_increase}, is}
		&\leq \log \chis{\rho\mid W}{\nu,A}+\log (1/\rho(G_b\mid W))+\log ((r+3)/\gamma)+\log (1/\rho(G)) \\
		&\leq \log \chis{\rho\mid W}{\nu,A}+\log ((r+3)/\gamma)+\log (1/((1-\gamma)\rho(G_b))).
	\end{align*}
	Similarly, $\log \chis{\rho_G}{\nu,B}\leq \log \chis{\rho\mid W}{\nu,B}+\log ((r+3)/\gamma)+\log (1/((1-\gamma)\rho(G_b)))$.
	This proves the lemma.
\end{proof}

\section{Compression of Generalized Protocols: Proof of Lemma~\ref{lem_compression}}\label{sec_compression}
In this subsection, we design a standard protocol with lower communication from a generalized protocol with low costs.
We will use the following lemma as a subroutine, whose proof is similar to Theorem 4.1 in~\cite{BR11}.
The lemma lets the players sample from a distribution $P$ with low communication, where only Alice knows $P$, and Bob knows a different distribution $Q$.
The success probability depends on how ``close'' the two distributions are.
\begin{lemma}\label{lem_cor_sample_pq}
	Let $P,Q$ be two distributions over $\cU$, such that Alice knows $P$ and Bob knows $Q$.
	For any $C>0$ and $\delta\in(0,1/2)$, there is a (standard) one-way communication protocol with shared public random bits, where Alice sends one message of $C+O(\log (1/\delta))$ bits to Bob.
	Then Alice and Bob simultaneously output an element in $\cU$ such that Alice outputs $x$ with probability $P(x)$ for every $x\in\cU$; conditioned on Alice outputting $x$, Bob outputs 
	\begin{itemize}
		\item the same $x$ with probability \emph{at least} $\min\{1, 2^C\cdot Q(x)/P(x)\}-\delta$,
		\item some different $x\in\cU$ with probability \emph{at most} $\delta$,
		\item $\bot$ otherwise.
	\end{itemize}
	
\end{lemma}
In particular, for each $x$ such that $P(x)\leq 2^C Q(x)$, the players will agree on $x$ with probability at least $(1-\delta)P(x)$.
\begin{proof}
	Let $t=\lceil 2\log (2/\delta)\rceil$.
	Consider the following protocol.
	\begin{code}{sample}[Protocol][(P; Q)]
		\item[]\textbf{Part I: Alice samples $x$}
		\item Alice and Bob view the public random bits as a sequence of $2t\cdot \left|\cU\right|$ \emph{uniform} samples $(x_i, p_i)$ in $\cU\times [0, 1]$
		\item\label{sample_step_2} Alice finds the first pair $(x_i, p_i)$ such that $p_i\leq P(x_i)$
		\item if such pair does not exist
		\item\qquad Alice outputs an $x\sim P$, and sends ``0'' to Bob \hfill // ``$0$'' indicates ``fail''
		\item\qquad\label{sample_step_5} upon receiving ``0'', Bob outputs $\bot$, and the protocol aborts
		\item otherwise, Alice outputs $x_i$, and sends 1\\ \strut\ctn
	\end{code}
	So far, the protocol describes how Alice samples $x$.
	Now, we show that Alice indeed samples $x$ according to $P$.
	Fix $x\in\cU$, for each sample $(x_i, p_i)$, the probability that $x_i=x$ and $p_i\leq P(x)$ is
	\[
		\frac{1}{\left|\cU\right|}\cdot P(x).
	\]
	Summing over all possible $x$, the probability that $p_i\leq P(x_i)$ is equal to $\frac{1}{\left|\cU\right|}$.
	The probability that Alice outputs $x$ is 
	\begin{align*}
		&\sum_{i=1}^{2t\cdot \left|\cU\right|}\left(1-\frac{1}{\left|\cU\right|}\right)^{i-1}\cdot \frac{1}{\left|\cU\right|}\cdot P(x) \\
		&=P(x)\cdot \left(1-\left(1-\frac{1}{\left|\cU\right|}\right)^{2t\cdot \left|\cU\right|}\right).
	\end{align*}
	Note that $\left(1-1/\left|\cU\right|\right)^{2t\cdot \left|\cU\right|}\leq e^{-2\log (2/\delta)}<\delta/2$.
	The probability that Alice finds a pair in step~\ref{sample_step_2} is at least $1-\delta/2$, and conditioned on finding such a pair, $x_i$ is distributed according to $P$.
	Next, Alice tries to inform Bob by hashing the index $i$.
	\begin{code}{sample}
		\item[]\textbf{Part II: Bob outputs $x$}
		\item the players view the remaining public random bits as a uniformly random hash function $h: [2t\cdot \left|\cU\right|]\rightarrow \{0,1\}^{C+\lceil\log (2/\delta)+\log (2t)\rceil}$
		\item Alice sends $v=h(i)$ to Bob
		\item\label{sample_step_8} upon receiving $v$, Bob finds all $i\in[2t\cdot \left|\cU\right|]$ such that $h(i)=v$ and $p_i\leq 2^C\cdot Q(x_i)$
		\item if there is only one such $i$
		\item\qquad Bob outputs $x_i$
		\item else
		\item \qquad Bob outputs $\bot$
	\end{code}

	It is clear that the communication cost is at most
	\[
		C+\log (1/\delta)+\log\log (1/\delta)+7\leq C+O(\log(1/\delta))
	\]
	bits.

	Conditioned on Alice outputting $x_i$, $p_i$ is uniform in $[0, P(x_i)]$.
	Hence, the correct $i$ will be found in step~\ref{sample_step_8} with probability $\min\{1, 2^C\cdot Q(x_i)/P(x_i)\}$.
	Next, we bound the probability that Bob finds any other $j\neq i$, conditioned on Alice outputting $x_i$.

	For each $(x_j, p_j)$, if $j>i$, the probability that $p_j\leq 2^C\cdot Q(x_j)$ is $\min\left\{1, 2^C\cdot Q(x_j)\right\}\leq 2^C\cdot Q(x_j)$.
	If $j<i$, conditioned on Alice outputting $x_i$, $p_j$ is uniform in $(P(x_j), 1]$.
	The probability that $p_j\leq 2^C\cdot Q(x_j)$ is
	\[
		\frac{\max\{0,\min\left\{1, 2^C\cdot Q(x_j)\right\}-P(x_j)\}}{1-P(x_j)}\leq 2^C\cdot Q(x_j).
	\]
	Independently, the probability that $h(j)=v$ is equal to $2^{-(C+\lceil\log (2/\delta)+\log (2t)\rceil)}$.

	Therefore, the probability $(x_j, p_j)$ satisfies both conditions is at most
	\begin{align*}
		\sum_{x}\Pr[x_j=x]\cdot2^{-(C+\lceil\log (2/\delta)+\log (2t)\rceil)}\cdot 2^C\cdot Q(x)<\frac{\delta}{4t\cdot \left|\cU\right|}.
	\end{align*}
	By union bound, the probability that any other $(x_j, p_j)$ satisfies both conditions is at most $\delta/2$.

	To conclude, Bob outputs the same $x_i$ when
	\begin{itemize}
		\item Bob does not output $\bot$ in step~\ref{sample_step_5} (with probability $\geq 1-\delta/2$), and
		\item Bob finds the correct $i$ in step~\ref{sample_step_8} (with probability $\min\{1, 2^C\cdot Q(x_i)/P(x_i)\}$), and
		\item Bob does not find any other $j\neq i$ in step~\ref{sample_step_8} (with probability $\geq 1-\delta/2$).
	\end{itemize}
	By union bound, Bob outputs the same $x_i$ with probability at least $\min\{1, 2^C\cdot Q(x_i)/P(x_i)\}-\delta$.
	Bob outputs some different $x_i$ \emph{only} when
	\begin{itemize}
		\item Bob does not output $\bot$ in step~\ref{sample_step_5}, and
		\item Bob does not find the correct $i$ in step~\ref{sample_step_8}, and
		\item Bob find some other $j\neq i$ (and $x_j\neq x_i$) in step~\ref{sample_step_8} (with probability $\leq \delta/2$).
	\end{itemize}
	Bob outputs some different $x_i$ with probability at most $\delta/2$.
	Otherwise, Bob outputs $\bot$.
	This proves the lemma.
\end{proof}

We will use the above lemma to sample messages $M_i$ given $M_{<i}$.
The next lemma proves that most of time, in Alice and Bob's view, the probabilities of $M_i$ are not too different.

\begin{lemma}\label{lem_small_prob}
	Let $\rho$ be an $r$-round generalized protocol and $W$ be an event such that $(\rho\mid W)$ has the rectangle property with respect to $\mu$, and let $(X, Y, \bM)\sim\rho\mid W$.
	Then for any $T>1$, the probability that 
	\begin{itemize}
		\item there exists an odd $i\in [r]$ such that
		\[
			\frac{\rho(M_i\mid X, M_{<i})}{\rho(M_i\mid Y, M_{<i})}>T\cdot \chis{\rho\mid W}{\mu,A},
		\]
		or
		\item there exists an even $i\in [r]$ such that
		\[
			\frac{\rho(M_i\mid Y, M_{<i})}{\rho(M_i\mid X, M_{<i})}>T\cdot \chis{\rho\mid W}{\mu,B},
		\]
	\end{itemize}
	is at most $6r\cdot T^{-1/5}\cdot \rho(W)^{-1}$.
\end{lemma}
\begin{proof}
	We first fix an odd $i\in[r]$, and upper bound the probability that $\frac{\rho(M_i\mid X, M_{<i})}{\rho(M_i\mid Y, M_{<i})}>T\cdot \chis{\rho\mid W}{\mu,A}$.
	Recall that 
	\[
		\chis{\rho\mid W}{\mu,A}=\E_{\rho\mid W}\left[\frac{\rho(X\mid Y, \bM)}{\mu(X\mid Y)}\right],
	\]
	and we have
	\begin{align}
		\frac{\rho(X\mid Y, \bM)}{\mu(X\mid Y)}&=\frac{\rho(X\mid Y, \bM)}{\rho(X\mid Y)}\cdot \frac{\rho(X\mid Y)}{\mu(X\mid Y)} \nonumber\\
		&=\frac{\rho(\bM\mid X, Y)}{\rho(\bM\mid Y)}\cdot \frac{\rho(X\mid Y)}{\mu(X\mid Y)} \nonumber\\
		&=\frac{\rho(M_{<i}\mid X, Y)}{\rho(M_{<i}\mid Y)}\cdot \frac{\rho(M_i\mid M_{<i}, X, Y)}{\rho(M_i\mid M_{<i}, Y)}\cdot \frac{\rho(M_{>i}\mid M_{\leq i}, X, Y)}{\rho(M_{>i}\mid M_{\leq i}, Y)}\cdot \frac{\rho(X\mid Y)}{\mu(X\mid Y)} \nonumber\\
		&=\frac{\rho(M_i\mid M_{<i}, X)}{\rho(M_i\mid M_{<i}, Y)}\cdot \left(\frac{\rho(M_{<i}\mid X, Y)}{\rho(M_{<i}\mid Y)}\cdot \frac{\rho(M_i\mid M_{<i}, X, Y)}{\rho(M_i\mid M_{<i}, X)}\cdot \frac{\rho(M_{>i}\mid M_{\leq i}, X, Y)}{\rho(M_{>i}\mid M_{\leq i}, Y)}\cdot \frac{\rho(X\mid Y)}{\mu(X\mid Y)}\right) \nonumber\\
		&=:\frac{\rho(M_i\mid M_{<i}, X)}{\rho(M_i\mid M_{<i}, Y)}\cdot \left(F_1\cdot F_2\cdot F_3\cdot F_4\right), \label{eqn_ratio_F1234}
	\end{align}
	where $F_1,F_2,F_3,F_4$ denote the four fractions in the parenthesis respectively.
	Note that the fraction outside the parenthesis is what we want to upper bound.
	We now show that $F_1,F_2,F_3,F_4$ are all not-too-small with high probability.

	For $F_1$, we have
	\begin{align*}
		\E_{\rho}\left[1/F_1\right] &= \E_{\rho}\left[\frac{\rho(M_{<i}\mid Y)}{\rho(M_{<i}\mid X, Y)}\right] \\
		&=\sum_{X, Y, M_{<i}} \rho(X, Y, M_{<i})\cdot \frac{\rho(M_{<i}\mid Y)}{\rho(M_{<i}\mid X, Y)} \\
		&=\sum_{X, Y, M_{<i}} \rho(X, Y)\cdot \rho(M_{<i}\mid Y) \\
		&=\sum_{X, Y} \rho(X, Y) \\
		&=1.
	\end{align*}
	Similarly, we can show that
	\[
		\E_{\rho}\left[1/F_2\right]=\E_{\rho}\left[1/F_3\right]=1,
	\]
	and
	\[
		\E_{\rho}\left[1/F_4\right]=\sum_{X, Y}\rho(X, Y)\cdot \frac{\mu(X\mid Y)}{\rho(X\mid Y)}=\sum_{X, Y}\rho(Y)\mu(X\mid Y)=1.
	\]
	Since $F_1,F_2,F_3,F_4$ are all nonnegative, by Markov's inequality, we have
	\[
		\Pr_{\rho}\left[1/F_j\geq \delta^{-1}\right]\leq \delta,
	\]
	for $j=1,2,3,4$ and any $\delta\in (0, 1)$.
	Thus, $\Pr_{\rho\mid W}\left[1/F_j\geq \delta^{-1}\right]\leq \delta/\rho(W)$.

	By union bound and plugging into~\eqref{eqn_ratio_F1234}, we have
	\[
		\Pr_{\rho\mid W}\left[\frac{\rho(M_i\mid M_{<i}, X)}{\rho(M_i\mid M_{<i}, Y)}\geq \delta^{-4}\cdot \frac{\rho(X\mid Y, \bM)}{\mu(X\mid Y)}\right]\leq 4\delta/\rho(W).
	\]

	Thus by union bound over all odd $i\in[r]$, we have
	\[
		\Pr_{\rho\mid W}\left[\exists \odd i\in[r], \frac{\rho(M_i\mid M_{<i}, X)}{\rho(M_i\mid M_{<i}, Y)}\geq \delta^{-4}\cdot \frac{\rho(X\mid Y, \bM)}{\mu(X\mid Y)}\right]\leq 4 \lceil r/2\rceil\cdot \delta/\rho(W).
	\]
	By Markov's inequality again, we have
	\[
		\Pr_{\rho\mid W}\left[\frac{\rho(X\mid Y, \bM)}{\mu(X\mid Y)}\geq \delta^{-1}\cdot \chis{\rho\mid W}{\mu,A}\right]\leq \delta.
	\]
	Combining the two inequalities, we obtain
	\[
		\Pr_{\rho\mid W}\left[\exists \odd i\in[r], \frac{\rho(M_i\mid M_{<i}, X)}{\rho(M_i\mid M_{<i}, Y)}\geq \delta^{-5}\cdot \chis{\rho\mid W}{\mu,A}\right]\leq 4 \lceil r/2\rceil\cdot \delta/\rho(W)+\delta.
	\]

	Similarly, for even $i$, we can prove that
	\[
		\Pr_{\rho\mid W}\left[\exists \even i\in[r], \frac{\rho(M_i\mid M_{<i}, Y)}{\rho(M_i\mid M_{<i}, X)}\geq \delta^{-5}\cdot \chis{\rho\mid W}{\mu,B}\right]\leq 4 \lfloor r/2\rfloor\cdot \delta/\rho(W)+\delta.
	\]
	Finally, by setting $\delta=T^{-1/5}$ and applying a union bound on the odd and the even case, the probability is at most
	\[
		4r\cdot T^{-1/5}/\rho(W)+2T^{-1/5}\leq 6rT^{-1/5}/\rho(W).
	\]
	This proves the lemma.
\end{proof}

We will also use the following lemma in the proof.

\begin{lemma}\label{lem_small_prob_2}
	Let $\rho$ be an $r$-round generalized protocol and $W$ be an event such that $(\rho\mid W)$ has the rectangle property with respect to $\mu$, and let $(X, Y, \bM)\sim\rho\mid W$.
	Then for any $T>1$, the probability that 
	\[
		\thec{\rho}{\mu}[X, Y, \bM]\cdot\frac{\rho(X, Y, \bM, W)}{\rho(X, Y, \bM)}>T\cdot \thec{\rho\mid W}{\mu},
	\]
	or
	\[
		\thec{\rho}{\mu}[X, Y, \bM]\cdot\frac{\rho(X, Y, \bM, W)}{\rho(X, Y, \bM)}<1/T,
	\]
	is at most $2\cdot T^{-1}\cdot \rho(W)^{-1}$.
\end{lemma}
\begin{proof}
	The first half is bounded using an application of Markov's inequality and the fact that $\frac{\rho(X, Y, \bM, W)}{\rho(X, Y, \bM)}\leq 1$:
	\[
		\Pr_{\rho\mid W}\left[\thec{\rho}{\mu}[X, Y, \bM]>T\cdot \thec{\rho\mid W}{\mu}\right]<1/T,
	\]
	implying that
	\[
		\Pr_{\rho\mid W}\left[\thec{\rho}{\mu}[X, Y, \bM]\cdot \frac{\rho(X, Y, \bM, W)}{\rho(X, Y, \bM)}>T\cdot \thec{\rho\mid W}{\mu}\right]<1/T.
	\]

	For the second half, similar to the proof of Lemma~\ref{lem_small_prob}, we have
	\begin{align*}
		&\kern-2em\E_{\rho\mid W}\left[\thec{\rho}{\mu}[X, Y, \bM]^{-1}\cdot\frac{\rho(X, Y, \bM)}{\rho(X, Y, \bM, W)}\right] \\
		&=\sum_{X, Y, \bM}\rho(X, Y, \bM\mid W)\cdot \frac{\rho(M_0)\cdot\mu(X, Y)\cdot \prod_{\odd i\in[r]} \rho(M_i\mid X, M_{<i})\cdot\prod_{\even i\in[r]} \rho(M_i\mid Y, M_{<i})}{\rho(X, Y, \bM, W)}\\
		&=\frac{1}{\rho(W)}\cdot\sum_{X, Y, \bM}\rho(M_0)\cdot \mu(X, Y)\cdot \prod_{\odd i\in[r]} \rho(M_i\mid X, M_{<i})\cdot\prod_{\even i\in[r]} \rho(M_i\mid Y, M_{<i})\\
		&=\frac{1}{\rho(W)}\cdot\sum_{X, Y, M_{<r}}\rho(M_0)\cdot \mu(X, Y)\cdot \prod_{\odd i\in[r-1]} \rho(M_i\mid X, M_{<i})\cdot\prod_{\even i\in[r-1]} \rho(M_i\mid Y, M_{<i})\\
		&=\cdots\\
		&=\frac{1}{\rho(W)}\cdot\sum_{X, Y, M_0} \rho(M_0)\cdot \mu(X, Y) \\
		&=\frac{1}{\rho(W)}.
	\end{align*}
	Hence, by Markov's inequality, 
	\[
		\Pr_{\rho}\left[\thec{\rho}{\mu}[X, Y, \bM]^{-1}\cdot\frac{\rho(X, Y, \bM)}{\rho(X, Y, \bM, W)}>T\right]<T^{-1}\cdot \rho(W)^{-1}.
	\]
	We prove the lemma by an application of the union bound.
\end{proof}

Finally, we are ready to prove Lemma~\ref{lem_compression}.
\begin{restate}[Lemma~\ref{lem_compression}]
	\lemcompressioncont
\end{restate}
\begin{proof}
	Let us first consider the following ``ideal protocol'' $\tau^*$ that cannot necessarily be implemented in the standard communication setting.
	But we can still analyze the probability that $\tau^*$ computes $f(X, Y)$.
	Then we construct a standard protocol $\tau$ with low communication and statistically close to $\tau^*$ when $(X, Y)$ is sampled from $\mu$.

	The ideal protocol consists of two parts: In the first part, the players generate a transcript $\bM$ given the inputs $(X, Y)$; in the second part, they use rejection sampling, and accept $\bM$ with some carefully chosen probability (and output a random bit if they reject).
	\begin{code}{ideal-tau}[``Ideal protocol''][(X; Y)][\tau^*]
		\item[]\textbf{Part I}
		\item Alice and Bob use public random bits to sample $M_0$ from $\rho(M_0)$
		\item for $i=1,\ldots,r-1$
		\item \qquad if $i$ is odd, Alice samples $M_i$ from $\rho(M_i\mid X, M_{<i})$ and sends it to Bob
		\item \qquad if $i$ is even, Bob samples $M_i$ from $\rho(M_i\mid Y, M_{<i})$ and sends it to Alice
		\item Bob locally samples $M_r$ from $\rho(M_r\mid Y, M_{<r})$ \hfill // recall that $M_r\in\{0,1\}$
		\item for $j=0,1$
		\item\qquad Alice examines the distribution of $\rho(f(X, Y)\mid X, \bM, W)$ pretending $M_r=j$
		\item\qquad Alice sends Bob the more likely value $p_j\in\{0,1\}$ of $f(X, Y)$ in this conditional distribution
		\item[]\textbf{Part II}
		\item Alice and Bob accept $\bM$ with probability equal to $\gamma\cdot \thec{\mu}{\rho}[X, Y, \bM]\cdot \frac{\rho(X, Y, \bM, W)}{\rho(X, Y, \bM)}$ (assuming it is at most $1$), for some fixed parameter $\gamma$
		\item if the players decide to accept
		\item \qquad Bob sends $p_{M_r}$
		\item else
		\item \qquad Bob sends a random bit
	\end{code}

	\paragraph{Success probability of $\tau^*$.}
	Given $X, Y$, Alice and Bob generate transcript $\bM$ with probability
	\[
		\rho(M_0)\cdot \prod_{\odd i\in [r]}\rho(M_i\mid X, M_{<i})\cdot \prod_{\even i\in[r]}\rho(M_i\mid Y, M_{<i}),
	\]
	where $M_r$ is only known to Bob.
	Then it is accepted with probability $\gamma\cdot \thec{\mu}{\rho}[X, Y, \bM]\cdot \frac{\rho(X, Y, W\mid \bM)}{\rho(X, Y\mid \bM)}$.
	Recall that
	\[
		\thec{\rho}{\mu}[X, Y, \bM]=\frac{\rho(X, Y, \bM)}{\mu(X, Y)\cdot \rho(M_0)\cdot \prod_{\odd i\in[r]} \rho(M_i\mid X, M_{<i})\cdot\prod_{\even i\in[r]} \rho(M_i\mid Y, M_{<i})}.
	\]
	Hence, for $(X, Y)\sim\mu$, the probability that the players generate and accept $(X, Y, \bM)$ is 
	\begin{align*}
		&\kern-2em\mu(X, Y)\cdot\rho(M_0)\cdot \prod_{\odd i\in [r]}\rho(M_i\mid X, M_{<i})\cdot \prod_{\even i\in[r]}\rho(M_i\mid Y, M_{<i})\cdot \gamma\cdot \thec{\mu}{\rho}[X, Y, \bM]\cdot \frac{\rho(X, Y, \bM, W)}{\rho(X, Y, \bM)} \\
		&=\gamma\cdot \rho(X, Y, \bM, W).
	\end{align*}
	Thus, the probability that $\tau^*$ accepts is $\gamma\cdot \rho(W)$.
	Alice does not know $M_r$, so she sends the more likely value $p_j$ conditioned on $(X, \bM)$ for both possibilities of $M_r$, and Bob outputs this bit when they accept.

	Since conditioned on accepting, $(X, Y, \bM)$ follows the distribution of $\rho(X, Y, \bM\mid W)$, and Alice has told Bob in advance what is the more like value of $f(X, Y)$ conditioned on $(X, \bM, W)$.
	Intuitively, this should imply that the overall advantage should be $\gamma\cdot \rho(W)\cdot \E_{\rho\mid W}[\adv(f(X, Y)\mid X, \bM, W)]$.
	We now formally prove that this holds.
	We use $\tau^*(\bR)$ to denote the probability of $\bR$ in the distribution induced by running $\tau^*$ on $(X, Y)\sim \mu$.
	Note that the transcript of $\tau^*$ is $(M_{<r}, p_0, p_1, p)$.
	The expected overall advantage of $\tau^*$ is at least
	\begin{align*}
		&\kern-2em \sum_{M_{<r}, p}\tau^*(M_{<r}, p)\cdot \left|2\tau^*(f(X, Y)=1\mid M_{<r}, p)-1\right| \\
		&\geq \sum_{M_{<r}, p}\tau^*(M_{<r}, p)\cdot \left(2\tau^*(f(X, Y)=p\mid M_{<r}, p)-1\right) \\
		&=2\sum_{M_{<r}, p}\tau^*(f(X, Y)=p, M_{<r}, p)-1 \\
		&=\left(2\sum_{M_{<r}, p}\tau^*(f(X, Y)=p, M_{<r}, p, \textnormal{accept})\right)+\left(2\sum_{M_{<r}, p}\tau^*(f(X, Y)=p, M_{<r}, p, \textnormal{reject})-1\right).
	\end{align*}
	The first term is
	\begin{align*}
		&\quad\,\, 2\sum_{X, Y, \bM, p=f(X, Y)}\tau^*(X, Y, \bM, p_{M_r}=p, \textnormal{accept}) \\
		&=2\sum_{X, Y, \bM: p_{M_r}=f(X, Y)}\gamma\cdot\rho(X, Y, \bM, W) \\
		&=2\sum_{X, \bM} \gamma\cdot \rho(X, \bM, W)\cdot \rho(p_{M_r}=f(X, Y)\mid X, \bM, W)
		\intertext{which by the fact that $p_{M_r}$ is the more likely value of $f(X, Y)$ conditioned on $(X, \bM, W)$, is}
		&=2\sum_{X, \bM} \gamma\cdot \rho(X, \bM, W)\cdot \left(\frac{1}{2}+\frac{1}{2}\cdot \adv_{\rho}(f(X, Y)\mid X, \bM, W)\right) \\
		&=\gamma\cdot\rho(W)+\gamma\cdot\rho(W)\cdot\E_{\rho\mid W}\left[\adv_{\rho}(f(X, Y)\mid X, \bM, W)\right].
	\end{align*}
	The second term is equal to
	\begin{align*}
		&\kern-2em2\sum_{M_{<r}, p}\frac{1}{2}\cdot\tau^*(f(X, Y)=p, M_{<r}, \textnormal{reject})-1 \\
		&=-\tau^*(\textnormal{accept}) \\
		&=-\gamma\cdot\rho(W).
	\end{align*}
	The two terms sum up to $\gamma\cdot\rho(W)\cdot\E_{\rho\mid W}\left[\adv_{\rho}(f(X, Y)\mid X, \bM, W)\right]$.

	Hence, we have proved the following claim.
	\begin{claim}\label{cl_ideal_tau}
		If the probability that a protocol generates and accepts a triple $(X, Y, \bM)$ is equal to $\gamma\cdot \rho(X, Y, \bM, W)$, and it outputs a random bit otherwise, then this protocol computes $f(X, Y)$ correctly with probability at least
		\[
			\frac{1}{2}+\frac{\gamma}{2}\cdot \rho(W)\cdot\E_{\rho\mid W}\left[\adv_{\rho}(f(X, Y)\mid X, \bM, W)\right].
		\]
	\end{claim}

	\bigskip

	\paragraph{Standard protocol $\tau$.}
	Next, we will construct a standard protocol $\tau$ that simulates $\tau^*$.
	Similar to $\tau^*$, protocol $\tau$ also has two parts: In the first part, the players generate a transcript $\bM$; in the second part, they decide if they will accept $\bM$.

	For the first part, the players first use public randomness to sample $M_0$.
	Then for the subsequent messages $M_i$, Alice knows the distribution $\rho(M_i\mid X, M_{<i})$, and Bob knows the distribution $\rho(M_i\mid Y, M_{<i})$.
	For odd $i$, the players use Lemma~\ref{lem_cor_sample_pq} to sample from $\rho(M_i\mid X, M_{<i})$ where Alice sends a message; for even $i$, they sample from $\rho(M_i\mid Y, M_{<i})$ with Bob sending a message.
	Finally, Bob locally samples the last message $M_r$.
	We will show that Lemma~\ref{lem_small_prob} guarantees that the probability of sampling $\bM$ is a good approximation of $$\rho(M_0)\cdot \prod_{\odd i\in [r]}\rho(M_i\mid X, M_{<i})\cdot \prod_{\even i\in[r]}\rho(M_i\mid Y, M_{<i}).$$
	\begin{code}{tau}[Protocol][(X; Y)][\tau]
		\item[]\textbf{Part I}
		\item fix parameters $\delta_1,\delta_2\in(0,1/2)$
		\item Alice and Bob use public random bits to sample $M_0$ from $\rho(M_0)$
		\item for $i=1,\ldots,r-1$
		\item\label{step_Alice_sample_message} \qquad for odd $i$, use Lemma~\ref{lem_cor_sample_pq} to sample $M_i$ from $\rho(M_i\mid X, M_{<i})$ given that Bob only knows $\rho(M_i\mid Y, M_{<i})$, where we set $C:=C_0=\log \chis{\rho\mid W}{\mu,A}+5\log(6r/\delta_1)$ and $\delta:=\delta_2$ \hfill // Alice sends one message
		\item\label{step_Bob_sample_message} \qquad for even $i$, use Lemma~\ref{lem_cor_sample_pq} to sample $M_i$ from $\rho(M_i\mid Y, M_{<i})$ given that Alice only knows $\rho(M_i\mid X, M_{<i})$, where we set $C:=C_1=\log \chis{\rho\mid W}{\mu,B}+5\log(6r/\delta_1)$ and $\delta:=\delta_2$ \hfill // Bob sends one message
		\item in Bob's local memory: $\acc\leftarrow1$ \hfill // the final value of $\acc$ indicates if they will accept
		\item if any player outputs $\bot$ in any round
		\item \qquad $\acc\leftarrow0$
		\item Bob \emph{locally} samples $M_r$ from $\rho(M_r\mid Y, M_{<r})$ \ctn
	\end{code}
	Each player will send one extra bit indicating whether they output $\bot$ in the previous round.
	Hence, Bob knows if any player outputs $\bot$ in the first $r-1$ rounds (including round $r-1$, for which he does not need to send the extra bit).

	\bigskip
	Next, we use rejection sampling, and accept $\bM$ with probability roughly $\gamma\cdot \thec{\rho}{\mu}[X, Y, \bM]\cdot \frac{\rho(X, Y, \bM, W)}{\rho(X, Y, \bM)}$ for some carefully chosen $\gamma>0$.
	The rectangle property of $(\rho\mid W)$ ensures that this rejection sampling can be done approximately using very little communication.

	More specifically, by the rectangle property of $(\rho \mid W)$ with respect to $\mu$ (see Definition~\ref{def_rect}), there exists $g_1,g_2$ such that $\rho(X, Y, \bM\mid W)=\mu(X, Y)\cdot g_1(X, \bM)\cdot g_2(Y, \bM)$.
	Hence, $\thec{\rho}{\mu}[X, Y, \bM]\cdot\frac{\rho(X, Y, \bM, W)}{\rho(X, Y, \bM)}$ can be written as
	\[
		\thec{\rho}{\mu}[X, Y, \bM]\cdot\frac{\rho(X, Y, \bM, W)}{\rho(X, Y, \bM)}=g_A(X, \bM)\cdot g_B(Y, \bM),
	\]
	by letting $g_A(X, \bM):=\frac{\rho(W)\cdot g_1(X, \bM)}{\prod_{\odd i\in[r]} \rho(M_i\mid X, M_{<i})}$ and $g_B(Y, \bM):=\frac{g_2(Y, \bM)}{\rho(M_0)\cdot \prod_{\even i\in[r]} \rho(M_i\mid Y, M_{<i})}$.

	Suppose we let Alice accept with probability $\gamma_A\cdot g_A(X, \bM)$ and Bob accept with probability $\gamma_B\cdot g_B(Y, \bM)$ for $\gamma_A\gamma_B=\gamma$, then they will be able to accept with the correct probability by sending only one bit, i.e., whether they accept locally.
	We will also need to choose $\gamma_A$ and $\gamma_B$ carefully so that both probabilities are at most one.
	This is done by applying Lemma~\ref{lem_small_prob_2}, which ensures that most of the time $\thec{\rho}{\mu}[X, Y, \bM]\cdot\frac{\rho(X, Y, \bM, W)}{\rho(X, Y, \bM)}$ is between $\delta$ and $\thec{\rho\mid W}{\mu}/\delta$ for small $\delta>0$.
	Thus, they can coordinate the values of $\gamma_A$ and $\gamma_B$ by Alice sending a small hash value of some approximation of $g_A(X, \bM)$.
	\begin{code}{tau}
		\item[]\textbf{Part II}
		\item for $j=0,1$, Alice computes $g_A(X, \bM)$ pretending $M_r=j$, and computes $e_{A,j}:=\lfloor \log g_A(X, \bM)\rfloor$ \\
		 Bob computes $g_B(Y, \bM)$ and $e_B:=\lfloor \log g_B(Y, \bM)\rfloor$
		\item let $R:=\lceil \log (32\thec{\rho\mid W}{\mu}/\delta_1^2)/\delta_2 \rceil$, Alice and Bob use public random bits to sample a hash function $h:\mathbb{Z}\rightarrow [R]$
		\item for $j=0,1$
		\item\label{step_sample_b_A} \qquad Alice samples a bit $b_{A, j}$ such that $\Pr[b_{A, j}=1]=g_A(X, \bM)\cdot 2^{-(e_{A, j}+1)}$ pretending $M_r=j$
		\item\qquad Alice examines the distribution of $\rho(f(X, Y)\mid X, \bM, W)$ pretending $M_r=j$
		\item\qquad Alice sets $p_j\in\{0,1\}$ to the more likely value of $f(X, Y)$ in this conditional distribution
		\item Alice appends $(h(e_{A, 0}), h(e_{A, 1}), b_{A, 0}, b_{A, 1}, p_0, p_1)$ to her last message $M_{r-1}$\\
		\item let $L_1:=\lceil \log (4/\delta_1)\rceil$ and $L_2:=\lceil \log (\thec{\rho\mid W}{\mu}/\delta_1) \rceil$ 
		\item upon receiving $(v_0, v_1, b_{A, 0}, b_{A,1}, p_0, p_1)$, Bob checks:  \\
		if there is one \emph{unique} integer $e_A'\in[-e_B-L_1, -e_B+L_2]$ such that $h(e_A')=v_{M_r}$ 
		\item\label{step_sample_b_B}\qquad Bob samples a bit $b_B$ such that $\Pr[b_B=1]=g_B(Y, \bM)\cdot 2^{e_A'-L_2-1}$
		\item set $\acc\leftarrow0$ if there is no such $e_A'$, or $e_A'$ is not unique, or $b_{A, M_r}=0$, or $b_B=0$
		\item if $\acc=1$
		\item\qquad Bob sends $p_{M_r}$
		\item else
		\item\qquad Bob sends a random bit
	\end{code}
	Note that Alice's new messages are sent before Bob starts sending the last message, hence, it is still part of round $r-1$.
	In step~\ref{step_sample_b_A}, since $g_A(X, \bM)<2^{e_{A, j}+1}$ for $j=0,1$, the probability is at most $1$.
	In step~\ref{step_sample_b_B}, since $e_A'-L_2-1\leq -(e_B+1)$, the probability is also at most $1$.
	Hence, the protocol is well-defined.

	\paragraph{Communication cost.} By Lemma~\ref{lem_cor_sample_pq}, in odd rounds, Alice sends a message of length at most
	\[
		\log \chis{\rho\mid W}{\mu,A}+5\log(6r/\delta_1)+O(\log (1/\delta_2)),
	\]
	in even rounds, Bob sends a message of length at most
	\[
		\log \chis{\rho\mid W}{\mu,B}+5\log (6r/\delta_1)+O(\log (1/\delta_2)).
	\]
	They also send one extra bit indicating if the lemma outputs $\bot$ in the previous round.
	In Alice's last message (round $r-1$), Alice further sends two hash values $h(e_{A, 0}), h(e_{A, 1})$ and the bits $b_{A, 0},b_{A,1},p_0,p_1$, which takes at most 
	\[
		2\lceil \log R\rceil+4\leq O(\log \log \thec{\rho\mid W}{\mu} +\log\log (1/\delta_1)+\log (1/\delta_2))
	\]
	bits in total.
	This proves the communication bound of $\tau$ we claimed.
	
	\paragraph{The first part of $\tau$.}
	We first analyze the first part of $\tau$ and estimate the probability that we generate a triple $(X, Y, \bM)$.
	By Lemma~\ref{lem_small_prob}, for $(X, Y, \bM)\sim \rho\mid W$, the probability that (recall the value of $C_0$ in line~\ref{step_Alice_sample_message} and the value of $C_1$ in line~\ref{step_Bob_sample_message} of $\tau$) 
	\begin{itemize}
		\item there exists an odd $i\in [r]$ such that
		\[
			\frac{\rho(M_i\mid X, M_{<i})}{\rho(M_i\mid Y, M_{<i})}>2^{C_0},
		\]
		or
		\item there exists an even $i\in [r]$ such that
		\[
			\frac{\rho(M_i\mid Y, M_{<i})}{\rho(M_i\mid X, M_{<i})}>2^{C_1},
		\]
	\end{itemize}
	is at most $6r\cdot \left(2^{5\log(6r/\delta_1)}\right)^{-1/5}\cdot \rho(W)^{-1}=\delta_1\cdot \rho(W)^{-1}$.
	Denote this set of $(X, Y, \bM)$ by $\cB_1$, hence, $\rho(\cB_1\mid W)\leq \delta_1\cdot \rho(W)^{-1}$.

	Now consider any $(X, Y, \bM)\notin \cB_1$, and we estimate the probability that $\bM$ is generated by the players given $X, Y$.
	By Lemma~\ref{lem_cor_sample_pq}, conditioned on $(X, Y, M_{<i})$, for odd $i\in[r]$, the probability that both players agree on $M_i$ in step~\ref{step_Alice_sample_message} is at least
	\[
		\left(\min\left\{1, 2^{C_0}\cdot \frac{\rho(M_i\mid Y, M_{<i})}{\rho(M_i\mid X, M_{<i})}\right\}-\delta_2\right)\rho(M_i\mid X, M_{<i})\geq (1-\delta_2)\rho(M_i\mid X, M_{<i}),
	\]
	as $\frac{\rho(M_i\mid X, M_{<i})}{\rho(M_i\mid Y, M_{<i})}\leq 2^{C_0}$ for $(X, Y, \bM)\notin \cB_1$.
	Similarly, for even $i\in[r]$, both players agree on $M_i$ in step~\ref{step_Bob_sample_message} is at least 
	\[
		(1-\delta_2)\rho(M_i\mid Y, M_{<i}).
	\]
	Bob generates the last message $M_r$ with probability $\rho(M_r\mid Y, M_{<r})$.
	Thus, conditioned on $(X, Y)$, the probability that the players generate and agree on $\bM$ is at least
	\begin{align*}
		(1-(r-1)\delta_2)\rho(M_0)\cdot \prod_{\odd i\in [r]}\rho(M_i\mid X, M_{<i})\cdot \prod_{\even i\in[r]}\rho(M_i\mid Y, M_{<i}),
	\end{align*}
	where we used the fact that $(1-\delta_2)^{r-1}\geq (1-(r-1)\delta_2)$.

	On the other hand, for \emph{all} $(X, Y, \bM)$ (not necessarily in $\overline{\cB_1}$), the probability that the players agree on $M_i$ is at most
	$\rho(M_i\mid X, M_{<i})$ for odd $i\in[r]$ (since this is the probability that Alice outputs $M_i$ by Lemma~\ref{lem_cor_sample_pq}), and $\rho(M_i\mid Y, M_{<i})$ for even $i\in[r]$.
	Thus, the probability that they agree on $\bM$ is at most
	\[
		\rho(M_0)\cdot \prod_{\odd i\in [r]}\rho(M_i\mid X, M_{<i})\cdot \prod_{\even i\in[r]}\rho(M_i\mid Y, M_{<i}).
	\]

	Also, by Lemma~\ref{lem_cor_sample_pq} and union bound, the probability that the players do not agree on the same $M_i$ in some round is at most $(r-1)\delta_2$.
	Otherwise, some player outputs $\bot$, and $\acc$ is set to $0$.
	Thus, we obtain the following claim.
	\begin{claim}\label{cl_part1_tau}
		There is a set $\cB_1$ such that $\rho(\cB_1\mid W)\leq \delta_1\cdot \rho(W)^{-1}$, and given $(X, Y)$, the protocol $\tau$ generates $\bM$ in Part I of $\tau$ with probability at most
		\[
			\rho(M_0)\cdot \prod_{\odd i\in [r]}\rho(M_i\mid X, M_{<i})\cdot \prod_{\even i\in[r]}\rho(M_i\mid Y, M_{<i});
		\]
		furthermore, if $(X, Y, \bM)\notin \cB_1$, $\tau$ generates $\bM$ with probability at least
		\[
			(1-(r-1)\delta_2)\rho(M_0)\cdot \prod_{\odd i\in [r]}\rho(M_i\mid X, M_{<i})\cdot \prod_{\even i\in[r]}\rho(M_i\mid Y, M_{<i}).
		\]
		Moreover, the probability that the players do not agree on the same $\bM$ is at most $(r-1)\delta_2$.
	\end{claim}
	
	\paragraph{The second part of $\tau$.}

	Consider a triple $(X, Y, \bM)$, we analyze the probability that it is accepted in the second part, \emph{conditioned on} it being generated in the first part.
	Alice does not know $M_r$, but it has only two possible values.
	Alice pretends that $M_r=j$ for $j=0,1$, and computes the corresponding $g_A(X, \bM), e_{A, j}, b_{A, j}$ and $p_j$.
	She sends both copies (for $j=0,1$) to Bob, and Bob only looks at the copy corresponding to the actual $M_r$.
	In terms of the correctness, this is equivalent to Alice knowing $M_r$.
	For simplicity of notations, we will omit the subscript $j$, and use $g_A(X, \bM), e_A, b_A, p$ to denote the copy for the actual $M_r$.

	If for a triple $(X, Y, \bM)$, we have $-L_1\leq e_A+e_B\leq L_2$ (note that $e_A, e_B$ are determined by the triple), then the probability that there is a \emph{unique} integer $e_A'\in [-e_B-L_1,-e_B+L_2]$ such that $h(e_A')=h(e_A)$ is \emph{equal to}
	\[
		(1-1/R)^{L_1+L_2},
	\]
	and in this case, we must have $e_A'=e_A$.
	Then the probability that $b_A=1$ is
	\[
		g_A(X, \bM)\cdot 2^{-(e_A+1)},
	\]
	and the probability that $b_B=1$ is
	\[
		g_B(Y, \bM)\cdot 2^{e_A-L_2-1}.
	\]
	The players do not set $\acc$ to $0$ in Part II with probability
	\begin{align*}
		&(1-1/R)^{L_1+L_2}\cdot g_A(X, \bM)\cdot 2^{-(e_A+1)}\cdot g_B(Y, \bM)\cdot 2^{e_A-L_2-1}\\
		&=(1-1/R)^{L_1+L_2}\cdot 2^{-L_2-2}\cdot \thec{\rho}{\mu}[X, Y, \bM]\cdot\frac{\rho(X, Y, \bM, W)}{\rho(X, Y, \bM)}.
	\end{align*}
	
	On the other hand, if for a triple $(X, Y, \bM)$, either $e_A+e_B<-L_1$ or $e_A+e_B>L_2$, then the probability that the players accept it conditioned on it being generated is \emph{at most} the probability that there exists some $e_A'$ that matches the hash value of $e_A$, which by union bound, is at most
	\begin{align*}
		(L_1+L_2+1)/R&=(\lceil \log (4/\delta_1)\rceil +\lceil \log (\thec{\rho\mid W}{\mu}/\delta_1)\rceil+1)/\lceil \log(32\thec{\rho\mid W}{\mu}/\delta_1^2)/\delta_2\rceil \\
		&\leq (\log (4\thec{\rho\mid W}{\mu}/\delta_1^2)+3)/(\log (32\thec{\rho\mid W}{\mu}/\delta_1^2)/\delta_2) \\
		&=\delta_2.
	\end{align*}

	We denote the set of $(X, Y, \bM)$ such that either $e_A+e_B<-L_1$ or $e_A+e_B>L_2$ by $\cB_2$.
	If $e_A+e_B>L_2$, then we have
	\begin{align*}
		&\log \left(\thec{\rho}{\mu}[X, Y, \bM]\cdot\frac{\rho(X, Y, \bM, W)}{\rho(X, Y, \bM)}\right) \\
		&=\log (g_A(X, \bM)g_B(Y, \bM)) \\
		&\geq e_A+e_B \\
		&>L_2 \\
		&>\log(\thec{\rho\mid W}{\mu}/\delta_1).
	\end{align*}
	If $e_A+e_B<-L_1$, then we have
	\begin{align*}
		&\log \left(\thec{\rho}{\mu}[X, Y, \bM]\cdot\frac{\rho(X, Y, \bM, W)}{\rho(X, Y, \bM)}\right) \\
		&=\log (g_A(X, \bM)g_B(Y, \bM)) \\
		&< e_A+e_B+2 \\
		&<-L_1+2 \\
		&\leq\log\delta_1.
	\end{align*}

	However, by Lemma~\ref{lem_small_prob_2}, for $(X, Y, \bM)\sim\rho\mid W$, the probability that
	\[
		\thec{\rho}{\mu}[X, Y, \bM]\cdot\frac{\rho(X, Y, \bM, W)}{\rho(X, Y, \bM)}>\thec{\rho\mid W}{\mu}/\delta_1,
	\]
	or
	\[
		\thec{\rho}{\mu}[X, Y, \bM]\cdot\frac{\rho(X, Y, \bM, W)}{\rho(X, Y, \bM)}<\delta_1
	\]
	is at most $2\delta_1\cdot \rho(W)^{-1}$.
	This implies that $\rho(\cB_2\mid W)\leq 2\delta_1\cdot \rho(W)^{-1}$.
	Hence, we obtain the following claim.

	\begin{claim}\label{cl_part2_tau}
		There is a set $\cB_2$ such that $\rho(\cB_2\mid W)\leq 2\delta_1\cdot \rho(W)^{-1}$, and for $(X, Y, \bM)\notin\cB_2$, the probability that $\tau$ accepts $(X, Y, \bM)$ conditioned on $\tau$ generating $(X, Y, \bM)$ is \emph{equal to}
		\[
			(1-1/R)^{L_1+L_2}\cdot 2^{-L_2-2}\cdot \thec{\rho}{\mu}[X, Y, \bM]\cdot\frac{\rho(X, Y, \bM, W)}{\rho(X, Y, \bM)},
		\]
		for $(X, Y, \bM)\in\cB_2$, the probability that $\tau$ accepts $(X, Y, \bM)$ conditioned on $\tau$ generating $(X, Y, \bM)$ is \emph{at most} $\delta_2$.
	\end{claim}

	\paragraph{Overall success probability.}
	If all $(X, Y, \bM)$ were generated in the first part with probability equal to
	\[
		\rho(M_0)\cdot \prod_{\odd i\in [r]}\rho(M_i\mid X, M_{<i})\cdot \prod_{\even i\in[r]}\rho(M_i\mid Y, M_{<i}),
	\]
	and accepted in the second part with probability equal to
	\[
		(1-1/R)^{L_1+L_2}\cdot 2^{-L_2-2}\cdot \thec{\rho}{\mu}[X, Y, \bM]\cdot\frac{\rho(X, Y, \bM, W)}{\rho(X, Y, \bM)},
	\]
	then $\tau$ would be the ideal protocol in Claim~\ref{cl_ideal_tau} for $\gamma=(1-1/R)^{L_1+L_2}\cdot 2^{-L_2-2}$.
	Hence, to lower bound the overall success probability, it suffices to compare $\tau$ with $\tau^*$, and bound the total probability difference in generating and accepting a triple $(X, Y, \bM)$.

	By Claim~\ref{cl_part1_tau} and Claim~\ref{cl_part2_tau}, for $(X, Y, \bM)\notin \cB_1\cup\cB_2$, the probability that it is generated and accepted is at most
	\[
		(1-1/R)^{L_1+L_2}\cdot 2^{-L_2-2}\cdot \rho(X, Y, \bM, W)=\gamma\cdot \rho(X, Y, \bM, W),
	\]
	and at least
	\[
		(1-(r-1)\delta_2)\cdot (1-1/R)^{L_1+L_2}\cdot 2^{-L_2-2}\cdot \rho(X, Y, \bM, W)\geq (1-(r-1)\delta_2)\gamma\cdot \rho(X, Y, \bM, W).
	\]
	Hence, the total probability difference between $\tau$ and $\tau^*$ for these $(X, Y, \bM)$ is at most
	\begin{equation}\label{eqn_tau_diff_1}
		\sum_{(X, Y, \bM)\notin \cB_1\cup \cB_2}(r-1)\delta_2 \gamma\cdot \rho(X, Y, \bM, W)\leq (r-1)\delta_2\gamma\cdot\rho(W).
	\end{equation}
	For $(X, Y, \bM)\in \cB_1\setminus\cB_2$, the probability that it is generated and accepted is at most
	\[
		\gamma\cdot \rho(X, Y, \bM, W).
	\]
	Hence, the total probability difference for these $(X, Y, \bM)$ is at most
	\begin{equation}\label{eqn_tau_diff_2}
		\sum_{(X, Y, \bM)\in \cB_1\setminus\cB_2}\gamma\cdot \rho(X, Y, \bM, W)\leq \gamma\cdot \rho(\cB_1\mid W)\cdot \rho(W)\leq \gamma\delta_1.
	\end{equation}
	For $(X, Y, \bM)\in\cB_2$, the probability that it is generated and accepted is at most
	\[
		\delta_2\cdot \rho(M_0)\cdot \prod_{\odd i\in [r]}\rho(M_i\mid X, M_{<i})\cdot \prod_{\even i\in[r]}\rho(M_i\mid Y, M_{<i}).
	\]
	Hence, the total probability difference for these $(X, Y, \bM)$ is at most
	\begin{align}
		&\sum_{(X, Y, \bM)\in\cB_2}\max\left\{\delta_2\cdot \rho(M_0)\cdot \prod_{\odd i\in [r]}\rho(M_i\mid X, M_{<i})\cdot \prod_{\even i\in[r]}\rho(M_i\mid Y, M_{<i}), \,\,\gamma\cdot \rho(X, Y, \bM, W)\right\} \nonumber\\
		&\leq \sum_{(X, Y, \bM)\in\cB_2}\left(\delta_2\cdot \rho(M_0)\cdot \prod_{\odd i\in [r]}\rho(M_i\mid X, M_{<i})\cdot \prod_{\even i\in[r]}\rho(M_i\mid Y, M_{<i}) + \gamma\cdot \rho(X, Y, \bM, W)\right) \nonumber\\
		&\leq \delta_2+\gamma\cdot \rho(\cB_2\mid W)\cdot\rho(W) \nonumber\\
		&\leq \delta_2+2\gamma\delta_1.\label{eqn_tau_diff_3}
	\end{align}
	Finally, the players do not agree on the same $\bM$ with probability at most $(r-1)\delta_2$.

	Summing up Equation~\eqref{eqn_tau_diff_1},~\eqref{eqn_tau_diff_2},~\eqref{eqn_tau_diff_3} and the probability that they do no agree, the statistical distance between $\tau$ and $\tau^*$ is at most $3\gamma\delta_1+2r\delta_2$, where we used the fact that $\gamma\leq1$ and $\rho(W)\leq 1$.
	Combining it with Claim~\ref{cl_ideal_tau}, we obtain that $\tau$ computes $f$ correctly with probability at least
	\[
		\frac{1}{2}+\frac{\gamma}{2}\cdot \left(\rho(W)\cdot\E_{\rho\mid W}\left[\adv_\rho(f(X, Y)\mid X, \bM, W)\right]-6\delta_1\right)-2r\delta_2.
	\]
	Finally, note that if $\rho(W)\cdot\E_{\rho}\left[\adv_\rho(f(X, Y)\mid X, \bM)\right]-6\delta_1<0$, then the lemma holds trivially by setting $\tau$ to the protocol that outputs a random bit, otherwise, we have
	\begin{align*}
		\gamma&=(1-1/R)^{L_1+L_2}\cdot 2^{-L_2-2} \\
		&\geq (1-(L_1+L_2)/R)\cdot 2^{-L_2-2} \\
		&\geq (1-\delta_2)\cdot \frac{1}{8\thec{\rho\mid W}{\mu}/\delta_1} \\
		&\geq \frac{\delta_1}{16\thec{\rho\mid W}{\mu}}.
	\end{align*}
	Hence, the success probability is as claimed in the statement.
	This finishes the proof of the lemma.
\end{proof}

	\bibliographystyle{alpha}
	\bibliography{ref}

\newcommand{\etalchar}[1]{$^{#1}$}
\begin{thebibliography}{BRWY13b}

\bibitem[AN21]{AN21}
Sepehr Assadi and Vishvajeet N.
\newblock Graph streaming lower bounds for parameter estimation and property
  testing via a streaming {XOR} lemma.
\newblock In Samir Khuller and Virginia~Vassilevska Williams, editors, {\em
  {STOC} '21: 53rd Annual {ACM} {SIGACT} Symposium on Theory of Computing,
  Virtual Event, Italy, June 21-25, 2021}, pages 612--625. {ACM}, 2021.

\bibitem[BBCR13]{BBCR13}
Boaz Barak, Mark Braverman, Xi~Chen, and Anup Rao.
\newblock How to compress interactive communication.
\newblock {\em {SIAM} J. Comput.}, 42(3):1327--1363, 2013.

\bibitem[BJKS04]{BYJKS04}
Ziv Bar{-}Yossef, T.~S. Jayram, Ravi Kumar, and D.~Sivakumar.
\newblock An information statistics approach to data stream and communication
  complexity.
\newblock {\em J. Comput. Syst. Sci.}, 68(4):702--732, 2004.

\bibitem[BKLS20]{BKLS20}
Joshua Brody, Jae~Tak Kim, Peem Lerdputtipongporn, and Hariharan Srinivasulu.
\newblock A strong {XOR} lemma for randomized query complexity.
\newblock {\em CoRR}, abs/2007.05580, 2020.

\bibitem[BR11]{BR11}
Mark Braverman and Anup Rao.
\newblock Information equals amortized communication.
\newblock In Rafail Ostrovsky, editor, {\em {IEEE} 52nd Annual Symposium on
  Foundations of Computer Science, {FOCS} 2011, Palm Springs, CA, USA, October
  22-25, 2011}, pages 748--757. {IEEE} Computer Society, 2011.

\bibitem[BRWY13a]{BRWY13a}
Mark Braverman, Anup Rao, Omri Weinstein, and Amir Yehudayoff.
\newblock Direct product via round-preserving compression.
\newblock In Fedor~V. Fomin, Rusins Freivalds, Marta~Z. Kwiatkowska, and David
  Peleg, editors, {\em Automata, Languages, and Programming - 40th
  International Colloquium, {ICALP} 2013, Riga, Latvia, July 8-12, 2013,
  Proceedings, Part {I}}, volume 7965 of {\em Lecture Notes in Computer
  Science}, pages 232--243. Springer, 2013.

\bibitem[BRWY13b]{BRWY13}
Mark Braverman, Anup Rao, Omri Weinstein, and Amir Yehudayoff.
\newblock Direct products in communication complexity.
\newblock In {\em 54th Annual {IEEE} Symposium on Foundations of Computer
  Science, {FOCS} 2013, 26-29 October, 2013, Berkeley, CA, {USA}}, pages
  746--755. {IEEE} Computer Society, 2013.

\bibitem[BW15]{BW15}
Mark Braverman and Omri Weinstein.
\newblock An interactive information odometer and applications.
\newblock In Rocco~A. Servedio and Ronitt Rubinfeld, editors, {\em Proceedings
  of the Forty-Seventh Annual {ACM} on Symposium on Theory of Computing, {STOC}
  2015, Portland, OR, USA, June 14-17, 2015}, pages 341--350. {ACM}, 2015.

\bibitem[CKP{\etalchar{+}}21]{CKPSSY21}
Lijie Chen, Gillat Kol, Dmitry Paramonov, Raghuvansh~R. Saxena, Zhao Song, and
  Huacheng Yu.
\newblock Almost optimal super-constant-pass streaming lower bounds for
  reachability.
\newblock In Samir Khuller and Virginia~Vassilevska Williams, editors, {\em
  {STOC} '21: 53rd Annual {ACM} {SIGACT} Symposium on Theory of Computing,
  Virtual Event, Italy, June 21-25, 2021}, pages 570--583. {ACM}, 2021.

\bibitem[CSWY01]{CSWY01}
Amit Chakrabarti, Yaoyun Shi, Anthony Wirth, and Andrew~Chi{-}Chih Yao.
\newblock Informational complexity and the direct sum problem for simultaneous
  message complexity.
\newblock In {\em 42nd Annual Symposium on Foundations of Computer Science,
  {FOCS} 2001, 14-17 October 2001, Las Vegas, Nevada, {USA}}, pages 270--278.
  {IEEE} Computer Society, 2001.

\bibitem[Dru12]{Drucker12}
Andrew Drucker.
\newblock Improved direct product theorems for randomized query complexity.
\newblock {\em Comput. Complex.}, 21(2):197--244, 2012.

\bibitem[GNW11]{GNW11}
Oded Goldreich, Noam Nisan, and Avi Wigderson.
\newblock On yao's xor-lemma.
\newblock In Oded Goldreich, editor, {\em Studies in Complexity and
  Cryptography. Miscellanea on the Interplay between Randomness and Computation
  - In Collaboration with Lidor Avigad, Mihir Bellare, Zvika Brakerski, Shafi
  Goldwasser, Shai Halevi, Tali Kaufman, Leonid Levin, Noam Nisan, Dana Ron,
  Madhu Sudan, Luca Trevisan, Salil Vadhan, Avi Wigderson, David Zuckerman},
  volume 6650 of {\em Lecture Notes in Computer Science}, pages 273--301.
  Springer, 2011.

\bibitem[Imp95]{Impagliazzo95}
Russell Impagliazzo.
\newblock Hard-core distributions for somewhat hard problems.
\newblock In {\em 36th Annual Symposium on Foundations of Computer Science,
  Milwaukee, Wisconsin, USA, 23-25 October 1995}, pages 538--545. {IEEE}
  Computer Society, 1995.

\bibitem[IW97]{IW97}
Russell Impagliazzo and Avi Wigderson.
\newblock \emph{P = BPP} if \emph{E} requires exponential circuits:
  Derandomizing the {XOR} lemma.
\newblock In Frank~Thomson Leighton and Peter~W. Shor, editors, {\em
  Proceedings of the Twenty-Ninth Annual {ACM} Symposium on the Theory of
  Computing, El Paso, Texas, USA, May 4-6, 1997}, pages 220--229. {ACM}, 1997.

\bibitem[JPY12]{JPY12}
Rahul Jain, Attila Pereszl{\'{e}}nyi, and Penghui Yao.
\newblock A direct product theorem for the two-party bounded-round public-coin
  communication complexity.
\newblock In {\em 53rd Annual {IEEE} Symposium on Foundations of Computer
  Science, {FOCS} 2012, New Brunswick, NJ, USA, October 20-23, 2012}, pages
  167--176. {IEEE} Computer Society, 2012.

\bibitem[Lev87]{Levin87}
Leonid~A. Levin.
\newblock One-way functions and pseudorandom generators.
\newblock {\em Comb.}, 7(4):357--363, 1987.

\bibitem[Sha03]{Shaltiel03}
Ronen Shaltiel.
\newblock Towards proving strong direct product theorems.
\newblock {\em Comput. Complex.}, 12(1-2):1--22, 2003.

\bibitem[She11]{Sherstov11}
Alexander~A. Sherstov.
\newblock Strong direct product theorems for quantum communication and query
  complexity.
\newblock In Lance Fortnow and Salil~P. Vadhan, editors, {\em Proceedings of
  the 43rd {ACM} Symposium on Theory of Computing, {STOC} 2011, San Jose, CA,
  USA, 6-8 June 2011}, pages 41--50. {ACM}, 2011.

\bibitem[VW08]{VW08}
Emanuele Viola and Avi Wigderson.
\newblock Norms, {XOR} lemmas, and lower bounds for polynomials and protocols.
\newblock {\em Theory Comput.}, 4(1):137--168, 2008.

\bibitem[Yao82]{Yao82a}
Andrew~Chi{-}Chih Yao.
\newblock Theory and applications of trapdoor functions (extended abstract).
\newblock In {\em 23rd Annual Symposium on Foundations of Computer Science,
  Chicago, Illinois, USA, 3-5 November 1982}, pages 80--91. {IEEE} Computer
  Society, 1982.

\end{thebibliography}

	\appendix

\section{Theorem~\ref{thm_main} Implies Shaltiel's XOR Lemma}\label{app_disc}
Recall that the discrepancy of a function $f:\cX\times \cY\rightarrow\{-1,1\}$ is\footnote{For $\{0,1\}$-valued functions, we map the value to $\{-1,1\}$ then apply this definition.}
\[
	\disc(f):=\frac{1}{\left|\cX\right|\cdot \left|\cY\right|}\cdot \max_{\bR=\bX\times \bY: \bX\subseteq \cX, \bY\subseteq \cY}\left\{\left|\sum_{x\in \bX,y\in\bY}f(x,y)\right|\right\}.
\]
Shaltiel's XOR lemma states that
\[
	\disc(f^{\oplus n})= O(\disc(f))^{\Omega(n)}.
\]

We show that Theorem~\ref{thm_main} implies this bound.
First by viewing a protocol with $C$ bits of communication as a partitioning of $\cX\times \cY$ into $2^C$ combinatorial rectangles, any $C$-bit communication protocol cannot compute $f$ with probability better than
\[
	\frac{1}{2}+2^C\cdot \disc(f),
\]
when the inputs are sampled from the uniform distribution $\mu$ over $\cX\times \cY$.
In particular, it applies to $2$-round protocols, and we obtain that
\[
	\suc_{\mu}\left(f; \frac{1}{4}\log (1/\disc(f)), \frac{1}{4}\log (1/\disc(f)), 2\right)\leq \frac{1}{2}+\disc(f)^{1/2}.
\]
Thus, for $\disc(f)$ smaller than a sufficiently small constant (otherwise, the bound holds trivially), Theorem~\ref{thm_main} with $C_A=C_B=\frac{1}{4}\log (1/\disc(f))$, $\alpha=\disc(f)^{1/16c}(\geq 2\disc(f)^{1/2})$ and $r=2$ implies that no $2$-round protocol with communication $O(n\log (1/\disc(f)))$ solves $f^{\oplus n}$ with probability
\[
	\frac{1}{2}+\disc(f)^{\Omega(n)}.
\]
In particular, no protocol with \emph{two} bits of communication can solve $f^{\oplus n}$ with this probability.

Finally, as pointed out in Remark 3.12 in~\cite{VW08}, the discrepancy of a function is equal to the maximum advantage over $1/2$ that a $2$-bit protocol can achieve on the uniform distribution (up to a constant factor).
This proves that $\disc(f^{\oplus n})=O(\disc(f))^{\Omega(n)}$.
\end{document}